\long\def\commsingle #1\commsingleend{}
\long\def\commdouble #1\commdoubleend{#1}

\long\def\commabs #1\commabsend{#1}
\long\def\commshort #1\commshortend{#1}
\long\def\commlong #1\commlongend{}

\long\def\oldProbLowerBound #1\oldProbLowerBoundEnd{}

\long\def\commEREZDELA #1\commEREZDELAend{}

\documentclass[11pt]{article}
\usepackage[dvips]{graphicx}
\usepackage{epsfig,graphicx}
\usepackage{amsmath, amssymb}

\def\inline#1:{\par\vskip 7pt\noindent{\bf #1:}\hskip 10pt}
\def\midinline#1:{\par\noindent{\bf #1:}\hskip 10pt}
\def\dnsinline#1:{\par\vskip -7pt\noindent{\bf #1:}\hskip 10pt}
\def\ddnsinline#1:{\newline{\bf #1:}\hskip 10pt}

\newtheorem{theorem}{Theorem}[section]

\newtheorem{lem}[theorem]{Lemma}
\newtheorem{claim}[theorem]{Claim}
\newtheorem{corollary}[theorem]{Corollary}
\newtheorem{observation}[theorem]{Observation}

\def\proof{\par\noindent{\bf Proof:~}}

\def\blackslug{\hbox{\hskip 1pt \vrule width 4pt height 8pt
    depth 1.5pt \hskip 1pt}}

\def\QED{\quad\blackslug\lower 8.5pt\null\par}

\newcommand{\rot}[0]{v_0}

\newcommand{\MCD}{\mbox{\sc MCD}}
\newcommand{\DMCD}{\mbox{\sc DMCD}}

\newcommand{\StRSA}{\mbox{\sc SRSA}}
\newcommand{\SRSA}{\mbox{\sc SRSA}}
\newcommand{\RSA}{\mbox{\sc RSA}}

\newcommand{\commit}[0]{commit}

\newcommand{\neigh}[0]{N}

\newcommand{\stayactive}[0]{\mbox{stays-active}}
\newcommand{\Active}[0]{\mbox{active}}

\newcommand{\Inter}[2]{{I^{#1}_{#2}}}
\newcommand{\Interk}[2]{{I^{#2}(#1)}}
\newcommand{\lfun}[1]{\ell(#1)}

\newcommand{\cost}[0]{cost}

\newcommand{\lrangle}[1]{\langle #1\rangle}

\newcommand{\calX}{\mathcal{X}}

\newcommand{\calL}{\mathcal{L}}

\newcommand{\calV}{\mathcal{V}}
\newcommand{\calE}{\mathcal{E}}
\newcommand{\calH}{\mathcal{H}}
\newcommand{\calA}{\mathcal{A}}
\newcommand{\calI}{\mathcal{I}}
\newcommand{\calF}{\mathcal{F}}
\newcommand{\calR}{\mathcal{R}}
\newcommand{\calQ}{\mathcal{Q}}
\newcommand{\calP}{\mathcal{P}}
\newcommand{\calC}{\mathcal{C}}

\newcommand{\calK}{\mathcal{K}}
\newcommand{\calS}{\mathcal{S}}

\newcommand{\naturals}{\mathbb{N}}

\newcommand{\opt}{\mbox{\textsc{opt}}}
\newcommand{\alg}{\textsc{alg}}

\newcommand{\lineon}{\mbox{\sc Line}^{\mbox{\upshape on}}}

\newcommand{\onRSAmn}{\mbox{\sc rsa}^{\mbox{\upshape on}}_{\MM,\nn,p}}
\newcommand{\onRSAmnk}[1]{\mbox{\sc rsa}^{\mbox{\upshape on}}_{\MM_{#1},\nn_{#1},p_0^{#1}}}
\newcommand{\onRSAmnp}[0]{\mbox{\sc rsa}^{\mbox{\upshape on}}_{\MM,\nn,p}}
\newcommand{\onRSAn}{\mbox{\sc rsa}^{\mbox{\upshape on}}_{\nn}}

\newcommand{\onRSA}{\mbox{\sc rsa}^{\mbox{\upshape on}}}

\newcommand{\Base}[0]{\mbox{\sc Base}}
\newcommand{\DBase}[0]{{\mbox{\sc Base}}}

\newcommand{\Von}[0]{\calV^{\mbox{\upshape on}}}
\newcommand{\Aon}[0]{\calA^{\mbox{\upshape on}}}

\newcommand{\RSAFon}[0]{\calF^{\mbox{\upshape St}}_{\MM,\nn}}
\newcommand{\xmaxQ}[0]{\mbox{max}_x\calQ}

\newcommand{\Square}[0]{\mbox{\sc Square}}

\newcommand{\MM}[0]{M}

\newcommand{\linev}[1]{L_{ver}\langle #1\rangle}
\newcommand{\lineh}[1]{L_{hor}\langle #1\rangle}

\newcommand{\rep}[2]{(#1,#2)}
\newcommand{\rr}{r}

\newcommand{\NN}{N}

\newcommand{\nn}{n}

\newcommand{\qon}{q^{\mbox{\upshape on}}}
\newcommand{\uon}{u^{\mbox{\upshape on}}}

\newcommand{\PHon}[0]{\calP_\calH^{\mbox{\upshape\small on}}}
\newcommand{\PVon}[0]{\calP_\calV^{\mbox{\upshape\small on}}}

\newcommand{\tminus}[2]{t^{-}_{#1}#2}
\newcommand{\tplus}[1]{t^{+}_{#1}}

\newcommand{\payer}[0]{payer}

\newcommand{\distinf}[1]{{dist}^{\rightarrow}_{\infty}(#1)}

\newcommand{\bfDlineon}[0]{\mbox{{\bf D-Line}}^{\mbox{\bf on}}}
\newcommand{\Dlineon}[0]{\mbox{{\sc D-Line}}^{\mbox{\upshape on}}}
\newcommand{\NID}[2]{\overrightarrow{\neigh}^{#1}(#2)}

\newcommand{\Dron}[1]{{\rho}^{\mbox{\upshape\small on}}_{ #1}}
\newcommand{\DPHon}[0]{\calH^{\mbox{\upshape\small on}}}
\newcommand{\DPVon}[0]{\calV^{\mbox{\upshape\small on}}}

\newcommand{\Dqon}{{q}^{\mbox{\upshape on}}}
\newcommand{\Duon}{{u}^{\mbox{\upshape on}}}

\newcommand{\qSQ}{q^{\mbox{\small\upshape serve}}}
\newcommand{\uSQ}{u^{\mbox{\small\upshape serve}}}
\newcommand{\sSQ}{s^{\mbox{\small\upshape serve}}}

\newcommand{\qclose}{q^{\mbox{\small\upshape close}}}
\newcommand{\uclose}{u^{\mbox{\small\upshape close}}}
\newcommand{\sclose}{s^{\mbox{\small\upshape close}}}

\newcommand{\tail}[0]{\mbox{\sc tail}}
\newcommand{\FSQ}[0]{\calF^{\mbox{\small\upshape SQ}}}
\newcommand{\HSQ}[0]{\calH^{\mbox{\small\upshape SQ}}}
\newcommand{\ASQ}[0]{\calA^{\mbox{\small\upshape SQ}}}
\newcommand{\ADon}[0]{\overrightarrow{\calA}^{\mbox{\upshape on}}}
\newcommand{\HDon}[0]{\overrightarrow{\calH}^{\mbox{\upshape on}}}
\newcommand{\FDon}[0]{\overrightarrow{\calF}^{\mbox{\upshape on}}}

\newcommand{\pivot}[0]{pivot}

\newcommand{\NDI}[1]{\overrightarrow{\neigh}(#1)}
\newcommand{\DNLI}[1]{\overrightarrow{\neigh}^{L}(#1)}

\newcommand{\DCOMMIT}[0]{\mbox{\sc commit}}
\newcommand{\Dcommit}[0]{commit}

\newcommand{\eqdf}{\stackrel{\scriptscriptstyle \triangle}{=}}

\newcommand{\rhoSQ}[0]{\rho^{\mbox{\small\upshape SQ}}}
\newcommand{\Jnter}[3]{I\langle#1\rangle_{#2}^{#3}}

\newcommand{\seq}[1]{\calS[#1]}

\newcommand{\reck}[0]{\mbox{\sc rec}}

\newcommand{\Tree}[0]{\mbox{\sc tree}}

\newcommand{\parent}[0]{parent}
\newcommand{\suns}[0]{\mbox{\sc children}}
\newcommand{\roots}[0]{\mbox{\sc roots}}

\newcommand{\SQball}[0]{Q{\mbox{\sc -ball}}^{\mbox{\sc sq}}}
\newcommand{\uncover}[0]{\mbox{\sc uncover}}
\newcommand{\cover}[0]{\mbox{\sc cover}}
\newcommand{\Qball}[0]{Q{\mbox{\sc -ball}}}

\newcommand{\qlast}[0]{q^{\mbox{\small\upshape last}}}

\newcommand{\loglogratio}[1]{\frac{\log #1}{\log\log #1}}

\newcommand{\ff}[0]{f}
\newcommand{\nguess}[0]{n\mbox{-}guess}
\newcommand{\Mguess}[0]{M\mbox{-}guess}
\newcommand{\tetration}[1]{2^{2^{2#1}}}

\newcommand{\FRSAn}[0]{\calF^{\mbox{\sc rsa}}_{\nn}}

\newcommand{\FRSAmn}[0]{\calF^{\mbox{\sc rsa}}_{\MM,\nn,p}}
\newcommand{\FRSAQ}[0]{\calF^{\mbox{\sc rsa}}(\calQ)}
\newcommand{\FRSA}[0]{\calF^{\mbox{\sc rsa}}}

\newcommand{\Dlineonp}{\mbox{\sc D-Line}_+^{\mbox{\upshape on}}}

\newcommand{\set}[1]{\left\{ #1 \right\}}

\newcommand{\onalgrsa}[0]{\mbox{\sc onalg}_{\mbox{\sc\small mcd}}}

\newcommand{\diag}[0]{\mbox{\sc diag}}

\begin{document}

\title{
Optimal competitiveness for the Rectilinear Steiner Arborescence problem
}

\author{
Erez Kantor
\thanks{CSAIL, MIT, Cambridge, MA.
Supported in a part by NSF Awards 0939370-CCF, CCF-1217506 and  CCF-AF-0937274 and AFOSR FA9550-13-1-0042.
}\\
{\small\tt erezk@csail.mit.edu}
\and
Shay Kutten
\thanks{Department of
Industrial Engineering and Management,
IE\&M, Technion,
Haifa, Israel.  Supported in part by the
ISF,
Israeli ministry of science
and by the Technion Gordon Center.
} \\
{\small\tt kutten@ie.technion.ac.il}
}

\date{}

\maketitle

\begin{abstract}
We present optimal online algorithms for two related known problems involving Steiner Arborescence, improving both the lower and the upper bounds.
One of them is the well studied continuous problem of the {\em Rectilinear Steiner Arborescence} ($\RSA$).
We improve the lower bound and the upper bound on the competitive ratio for $\RSA$ from $O(\log N)$ and $\Omega(\sqrt{\log N})$ to $\Theta(\frac{\log N}{\log \log N})$,
where $N$ is the number of Steiner points.
This separates the competitive ratios of $\RSA$ and the Symetric-$\RSA$ $(\StRSA)$, two problems for which the bounds of Berman and Coulston is STOC 1997 were identical.
The second problem is one of the Multimedia Content Distribution problems presented by Papadimitriou et al. in several papers and Charikar  et al. SODA 1998.
It can be viewed as the discrete counterparts (or a network counterpart) of $\RSA$.
For this second problem we present tight bounds also in terms of the network size,
in addition to presenting tight bounds in terms of the number of Steiner points (the latter are similar to those we derived for $\RSA$).
\end{abstract}

\paragraph*{\bf Keywords: Online Algorithm, Approximation Algorithm, Video-on-Demand}


\section{Introduction}
\label{sec: Introduction}

Steiner trees, in general, have many applications, see e.g.
\cite{steiner-book} for a rather early survey that already included hundreds of items.
In particular, Steiner Arborescences%
\footnote{A Steiner arborescence is a Steiner tree directed away from the root.
}
are useful for describing the evolution of processes in time.
Intuitively, directed edges represent the passing of time.
Since there is no way to go back in time in such processes, all the directed edges are directed away from the initial state of the problem (the root), resulting in an arborescence. Various examples are given in the literature
such as processes in constructing a
Very Large Scale Integrated electronic circuits (VLSI),
optimization problems computed in iterations (where it was not feasible to return to results of earlier iterations),
dynamic programming, and problems involving DNA, see, e.g. \cite{berman,CDL01,vlsi,KnuthYao09}.
Papadimitriou at al. \cite{papa1,papa3} and Charikar et al. \cite{halperin} presented the discrete version, in the context of Multimedia Content Delivery ($\MCD$)
to model locating and moving caches for titles on a path graph.
The formal definition of (one of the known versions ) of this problem, Directed-$\MCD$,
appears in Section \ref{sec:preliminaries}.

We present new tight lower and upper bounds for two known interrelated problems involving Steiner Arborescences:
{\em Rectilinear Steiner Arborescence ($\RSA$)} and Directed-$\MCD$ (\DMCD).
We also deal indirectly with a third known arborescence problem: the {\em Symmetric-$\RSA$} ($\StRSA$) problem
by separating its competitive ratio from that of $\RSA$. That is, when the competitive ratios of $\RSA$ and $\StRSA$ were discussed originally by Berman and Coulston \cite{berman}, the same lower and upper bounds were  presented for both problems.

\paragraph*{The {\em RSA} problem:}
This is a rather heavily studied problem, described also e.g. in \cite{ptas1,shor-rsa,berman,natansky,presented-rsa}.
A rectilinear line segment in the plane is either horizontal or vertical.
A rectilinear path contains only rectilinear line segments.
This path is also $y$-{\em monotone} (respectively, $x$-{\em monotone}) if during the traversal, the $y$ (resp., $x$) coordinates of the successive points
are never decreasing.
The input is a set of {\em requests}
$\calR=\{r_1=(x_1,y_1),...,r_N=(x_N,y_N)\}$
called Steiner terminals (or points) in the positive quadrant of the plane.
A feasible solution  
to the problem is a set of rectilinear segments connecting all the $N$ terminals to the origin $r_0=(0,0)$,
where the path from the origin to each
terminal is both $x$-monotone and $y$-monotone (rectilinear shortest path).
The goal is to find a feasible solution in which the sum of lengths of all the segments is the minimum possible.
The above mentioned third problem, $\SRSA$
was defined in the same way, except that the above paths were not required to be $x$-monotone (only $y$-monotone).

Directed-$\MCD$ defined in Section \ref{sec:preliminaries} is very related to $\RSA$.
Informally, one difference is that it is discrete (Steiner points arrive only at discrete points) whiling $\RSA$ is continuous.
In addition, in $\DMCD$ each ``$X$ coordinates''  represents a network nodes.
Hence, the number of $X$ coordinates is bounded from above by the network size.
This resemblance turned out to be very useful for us, both for solving $\RSA$ and for solving $\DMCD$.

\paragraph*{The {\em online} version of {\em RSA} \cite{berman}:}
the given requests (terminals) are presented to the algorithm with nondecreasing
$y$-coordinates.
After receiving the $i$'th request $r_i=(x_i,y_i)$ (for $i=1,...,N$),
the on-line $\RSA$ algorithm must extend the existing arborescence solution to incorporate $r_i$.
There are two additional constraints:
(1) a line, once drawn (added to the solution), cannot be deleted,
and
(2) a segment added when handling a request $r_i$, can
only be drawn in the region between $y_{i-1}$ (the $y$-coordinates of the previous request $r_{i-1}$) and upwards
(grater $y$-coordinates).
If an algorithm obeys constraint (1) but not constraint (2), then we term it a {\em pseudo online} algorithm.
Note that quite a few algorithms known as ``online'', or as ``greedy offline'' fit this definition of ``pseudo online''.

\paragraph*{Additional Related works.}
Online algorithms for $\RSA$ and $\StRSA$ were presented by Berman and Coulston \cite{berman}.
The online algorithms in \cite{berman} were $O({\log N})$ competitive (where $N$ was the number of the Steiner points) both for $\RSA$ and $\SRSA$.
Berman and Coulston also presented $\Omega{(\sqrt{\log N})}$ lower bounds for both continuous problems.
Note that the upper bounds for both problems were equal, and were the squares of the lower bounds.
A similar gap for $\MCD$ arose from results of Halperin, Latombe, and  Motwani \cite{halperin1}, who gave a similar  competitive ratio of $O(\log N )$,
while  Charikar, Halperin,  and Motwani \cite{halperin} presented a lower bound of $\Omega(\sqrt{\log n})$  for various variants of $\MCD$,  where $n$ was the size of the network.
Their upper bound was again the square of the lower bound:
$O(\min \{ \log n, \log N \} )$  (translating their parameter $p$ to the parameter $n$ we use).

Berman and Coulston also conjectured that to close these gaps, both the upper bound and the lower bound
for both problems could be improved.
This conjecture was disproved in the cases of  $\SRSA$ and of $\MCD$ on undirected line networks \cite{KK2014}.
The latter paper closed the gap by presenting an optimal competitive ratio of $O(\sqrt{\log N})$ for $\SRSA$ and
$O(\min\{\sqrt{n},\sqrt{\log N}\})$ for $\MCD$ on the undirected line network with $n$ nodes.
They left the conjecture of Berman and Coulston open for $\RSA$ and for $\MCD$ on directed line networks.
In the current paper, we prove this conjecture (for $\RSA$ and for Directed-$\MCD$), thus separating $\RSA$ and $\SRSA$ in terms of their competitive ratios.

Charikar, Halperin, and Motwani \cite{halperin} also studied the  the offline case for $\MCD$, for which they gave   a constant approximation.
The offline version of $\RSA$ is heavily studied.
It was  attributed to    \cite{natansky} who gave an exponential integer programming solution and to \cite{presented-rsa} who gave an exponential time dynamic programming algorithm.
An exact and polynomial algorithm was proposed in \cite{rsa-error}, which seemed surprising, since many Steiner  problems are NP Hard.
Indeed, difficulties in that solution were noted by Rao, Sadayappan, Hwang, and Shor \cite{shor-rsa}, who also presented an approximation algorithm.
Efficient algorithms are claimed in \cite{another-at-poly} for VLSI applications.
However, the problem was proven NP-Hard in \cite{rsa-nph}.
(The rectilinear Steiner tree problem was proven NPH in \cite{garey-johnson}).
Heuristics that are fast ``in practice'' were presented in \cite{cong}.
A PTAS was presented by \cite{ptas1}.
An optimal logarithmic competitive ratio for $\MCD$ on {\em general undirected} networks was presented in \cite{MAicalp12}.
They also present a constant off-line approximation for $\MCD$ on grid networks.

\paragraph*{On the relation between this paper and \cite{KK2014}.}
An additional contribution of the current paper is the further development of the approach of developing (fully) online algorithms in two stages:
(a) develop a pseudo online algorithm; and
(b) convert the pseudo online into an online algorithm.
As opposed to the problem studied in \cite{KK2014} where a pseudo online algorithm was known, here the main technical difficulty was to develop such an algorithm.
From \cite{KK2014} we also borrowed an interesting twist on the rather common idea to translate between instances of a discrete and a continuous problems: we translate in {\em both} directions, the discrete  solutions helps in optimizing the continuous one {\em and vice versa}.


\paragraph*{Our Contributions.}
We improve both the upper and the lower bounds of $\RSA$  to show that the competitive ratio is $\Theta(\frac{\log N}{\log \log N})$.
This proves the conjecture for $\RSA$ of  Berman and Coulston  \cite{berman} and also separates the competitive ratios of $\RSA$ and $\SRSA$.
We also provide tight upper and lower bound for Directed-$\MCD$, the network version of $\RSA$ (both in terms of  $n$ and of $N$).
The main technical innovation is the specific pseudo online algorithm we developed here,
in order to convert it later to an online algorithm.
The previously known offline algorithms for $\RSA$ and for $\DMCD$ where {\em not} pseudo online, so we could not use them.
In addition to the usefulness of the new algorithm in generating the online algorithm, this pseudo online algorithm may be interesting in itself:
It is $O(1)$-competitive for $\DMCD$ and for $\RSA$ (via the transformation)
for a {\em different} (but rather common) online model (where
each request must be served before the next one arrives, but no time passes between requests).

\paragraph*{Paper Structure.}
Definitions are given in Section \ref{Sec:preliminaries}.
The pseudo online algorithm $\Square$ for $\DMCD$ is presented and analyzed in Section \ref{sec:square}.
In Section \ref{subsec: Algorithm Donline}, we transform $\Square$ to a (fully) online algorithm  $\Dlineon$ for $\DMCD$.
Then, Section \ref{subsec: optRSA} describes the transformation of the online $\DMCD$ algorithm  $\Dlineon$ to become an optimal online algorithm for $\RSA$, as well as a transformation back from $\RSA$ to $\DMCD$ to make the $\DMCD$ online algorithm also optimal in terms of $n$ (not just $N$).
These last two transformations are taken from \cite{KK2014}. Finally, a lower bound is given in Section \ref{sec:Lower bound}.
%
The best way to understand the algorithms in this paper may be from a geometric point of view.
Hence, we added multiple drawings
to illustrate both the algorithms and the proofs.

%
\vspace{-0.3cm}
\section{Preliminaries}
\label{Sec:preliminaries}
\label{sec:preliminaries}

\paragraph*{\bf The network$\times$time grid} (Papadimitriou et. al, \cite{papa3}).
A {\em directed line network } $L(\nn)=(V_\nn,E_\nn)$ is a network whose node set is $V_\nn=\{1,...,n\}$ and its edge set is $E_\nn=\{(i,i+1) \mid i=1,...,n-1\}$.
Given a directed line network $L(\nn)=(V_\nn,E_\nn)$, construct ''time-line'' graph $\calL(\nn)=(\calV_\nn,\calE_\nn)$, 
intuitively, by ``layering'' multiple replicas of $L(\nn)$,
one per time unit, where in addition, each node in each replica is connected to the same node in
the next replica
(see Fig. \ref{figure:TimeNet}).
Formally, the node set $\calV_n$ contains a {\em node replica}
(sometimes called just a {\em replica}) $\rep{v}{t}$ of every $v \in V_n$,
coresponding to each time step $t \in \naturals$.
That is, $\calV_n=\{\rep{v}{t} \mid v\in V_n, t\in \naturals \}$.
The set of directed edges $\calE_n=\calH_n\cup\calA_n$ contains \emph{horizontal directed edges}
$\calH_n=\{ (\rep{u}{t},\rep{v}{t}) \mid (u,v)\in E_n, t\in \naturals\}$,
connecting network nodes in every time step (round),
and directed {\em vertical edges}, called {\em arcs},
$\calA_n=\{ (\rep{v}{t},\rep{v}{t+1}) \mid v\in V_n, t\in \naturals\}$,
connecting different copies of $V_n$.
When $n$ is clear from the context, we may write just $X$ rather than $X_n$, for every $X\in\{V,E,\calV,\calH,\calA\}$.
Notice that $\calL(\nn)$ can be viewed geometrically as a
grid of $n$ by $\infty$ whose grid points are the replicas.
Following Fig. \ref{figure:TimeNet}, we consider the time as if it proceeds upward.
We use such geometric presentations also in the text, to help clarifying the description.

\begin{figure}[ht]
\begin{center}
\includegraphics[width=0.4\textwidth]{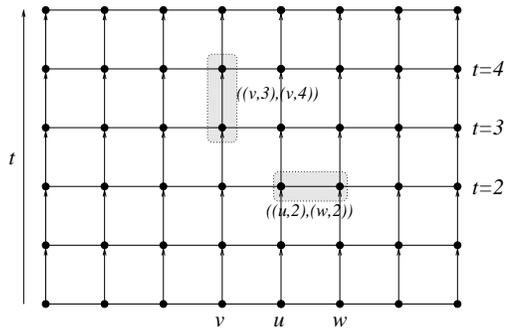}
\end{center}
\caption{\sf \label{figure:TimeNet}
An example of a time-line graph $\calL(n)=(\calV,\calE=\calH\cup\calA)$.
Each node in $\calV$ is represented by a circle;
 each horizontal edge in $\calH$ is represented by a horizontal segment (see, as an example, $(\rep{u}{2},\rep{w}{2})\in\calH$ for an horizontal directed edge in the marked  rectangle on the right);
each arc in $\calA$ is represented by a horizontal arrow (see, as an example,
$(\rep{v}{3},\rep{v}{4})\in\calA$ for an arc in the marked  rectangle on the left).
}
\end{figure}

\paragraph*{\bf The $\DMCD$ problem.}

We are given a directed line network $L(n)$, an \emph{origin} node  $\rot \in V$,
and a set of \emph{requests} $\calR \subseteq \calV$.
A feasible solution
is a subset of directed edges $\calF\subseteq\calE$
such that for every request $r\in\calR$, there exists a path in $\calF$ from the origin $(\rot,0)$ to $r$.
Intuitively a directed horizontal edge $((u,t),(v,t))$ is for delivering a copy of a multimedia title from node $u$ to node $v$ at time $t$.

A directed vertical edge (arc) $((v,t),(v,t+1))$ is for storing a copy of the title at node $v$ from time $t$ to time $t+1$.
For convenience, the endpoints $\calV_\calF$ of edges in $\calF$ are also considered parts of the solution.
For a given algorithm $A$, let $\calF_A$ be the solution of $A$, and let $\cost(A,\calR)$, (the cost of  algorithm $A$), be $|\calF_A|$.
(We assume that each storage cost and each delivery cost is $1$.)
The goal is to find a minimum cost feasible solution.
Let $\opt$ be the set of edges in some optimal solution whose cost is $|\opt|$.

\vspace{-0.3cm}
\paragraph{\bf Online {\bf \em DMCD}.}
\label{app:online}
In the online versions of the problem,  the algorithm receives as input a sequence of events.
One type of events is a request in the (ordered)
set $\calR$ of requests $\calR= \{r_1, r_2, ...,r_N\}$,
where the requests times are in a non-decreasing order, i.e.,  $t_1\leq t_2 \leq ...\leq t_N$ (as in $\RSA$).
A second type of events is a time event (this event does not exists in $\RSA$),
where we assume a clock that tells the algorithm that
no additional requests for time $t$ are about to arrive (or that there are no requests for some time $t$ at all).
The algorithm then still has the opportunity to complete its calculation for time $t$ (e.g., add arcs from some replica $(v,t)$ to $(v,t+1)$).
Then time $t+1$ arrives.

When handling an event ${ev}$,
the algorithm only knows the following:
(a) all the previous requests
$r_1,  ..., r_{i}$;
(b) time $t$; and
(c) the solution arborescence $\calF_{ev}$ it constructed so far (originally containing only the origin).
In each event,
the algorithm may need to
make decisions of two types, before seeing future events: 
\begin{itemize}
\vspace{-0.1cm}

\item [(1.$\DMCD$)] If the event is the arrival of a request $r_i=(v_i,t_i)$, then from which {\em current} (time $t_i$) cache
(a point already in the solution arborescence $\calF_{ev}$ when $r_{i}$ arrives) to serve $r_{i}$
by adding {\em horizontal} directed edges to $\calF_{ev}$.

\item[(2.$\DMCD$)]
If this is the time event for time $t$, then at which nodes to store a copy for time $t+1$, for future use: select some replica (or replicas)
$(v,t)$ already in the solution $\calF_{ev}$ and add to $\calF_{ev}$ an edge directed from $(v,t)$ to  $(v,t+1)$.

\end{itemize}

\noindent Note that at time $t$, the online algorithm cannot add nor delete any edge with an endpoint that corresponds to previous times.
Similarly to e.g. \cite{MAicalp12,papa1,papa2,papa3,halperin},
at least one copy must remain in the network at all times.

\paragraph*{\bf General definitions and notations.%
\commsingle.\commsingleend}

Consider an interval $J=\{v,v+1,...,v+\rho\}\subseteq V$
and two integers $s,t\in\naturals$, s.t. $s\leq t$.
Let $J[s,t]$ (see Fig. \ref{figure:subgraphJ}) be the
{\em ``rectangle subgraph''} of $\calL(\nn)$ corresponding to vertex set $J$ and time interval $[s,t]$.
This rectangle consists of the
replicas and edges of the nodes of $J$ corresponding to every time in the interval $[s,t]$.
For a given subsets $\calV'\subseteq \calV$, $\calH'\subseteq\calH$ and $\calA'\subseteq\calA$,
denote by
(1) $\calV'[s,t]$ replicas of $\calV'$ corresponding to times
 $s,...,t$. Define similarly (2) $\calH'[s,t]$ for horizontal edges of $\calH'$; and (3) $\calA'[s,t]$ arcs of $\calA'$.
(When $s=t$, we may write
$\calX[t]=\calX[s,t]$, for $\calX\in\{J,\calV',\calH'\}$.)
\begin{figure}[http!]
\begin{center}
\includegraphics[width=0.4\textwidth]{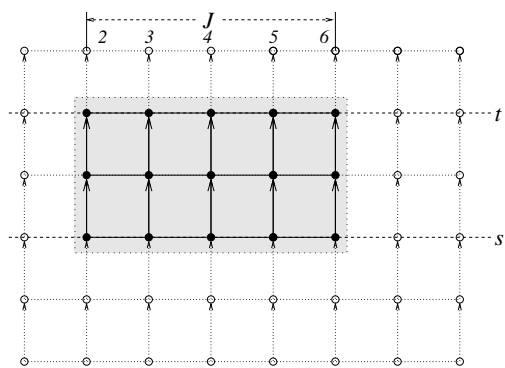}
\end{center}
\caption{\sf \label{figure:subgraphJ}
A subgraph rectangle $J[s,t]$, where $J=\{2,3,4,5,6\}$.}
\end{figure}
Consider also two nodes $v,u\in V$ s.t. $u\leq v$.
Let $\calP_\calH[\rep{u}{t},\rep{v}{t}]$
be the set of horizontal directed edges of the  path from $\rep{u}{t}$ to $\rep{v}{t}$.
Let $\calP_\calA[\rep{v}{s},\rep{v}{t}]$ be the set of arcs
of the path from $\rep{v}{s}$ to $\rep{v}{t}$.
Let $\distinf{\rep{u}{s},\rep{v}{t}}$ be the ``directed'' distance from $\rep{u}{s}$ to $\rep{v}{t}$ in $L_{\infty}$ norm.
Formally, $\distinf{\rep{u}{s},\rep{v}{t}}=\max\{t-s, v-u\}$, if $s\leq t$ and $u\leq v$ and $\distinf{\rep{u}{s},\rep{v}{t}}=\infty$, otherwise.

\vspace{-0.0cm}

\vspace{-0.3cm}
\section{Algorithm $\Square$, a pseudo online algorithm}
\label{sec:square}

This section describes a pseudo online algorithm named $\Square$ for the $\DMCD$ problem.
Developing $\Square$ was the main technical difficulty of this paper.
Consider a requests set $\calR=\{\rr_0=(0,0),\rr_1=(v_1,t_1),...,\rr_\NN=(v_\NN,t_\NN) \}$ such that $0\leq t_1\leq t_2\leq...\leq t_\NN$.
When Algorithm $\Square$ starts, the solution includes just $\rr_0=(0,0)$.
Then, $\Square$ handles, first, request $\rr_1$, then, request $\rr_2$, etc...
In handling a request $\rr_i$, the algorithm
may add some edges to the solution.
(It never deletes any edge from the solution.)
After handling $\rr_i$, the solution is an arborescence rooted at $\rr_0$ that spans the request replicas $\rr_1,...,\rr_i$.
Denote by $\Square(i)$ the solution of $\Square$ after handling the $i$'th request.
For a given replica $\rr=(v,t)\in \calV$ and a positive integer $\rho$,
let
\vspace{-0.1cm}
$$
\seq{r,\rho}=[v-\rho, v]\times [t-\rho, t]
$$
denotes the rectangle subgraph (of the layered graph) whose top right corner is $r$ induced by the set of replicas
that contains every replica $q$ such that
(1) there is a directed path in the layer graph from $q$ to $r$; and
(2) the distance from $q$ to $r$ in $L_\infty$ is at most $\rho$.
For each request $\rr_i\in\calR$, for $i=1,...,\NN$, $\Square$ performs the following
(The pseudo code of $\Square$ is given in Fig. \ref{figure:Pseducode Squarem}
and an example of an execution in Fig. \ref{fig: Square execution SQ3}).

\begin{figure}[htb]
\begin{center}
\includegraphics[scale=0.35]{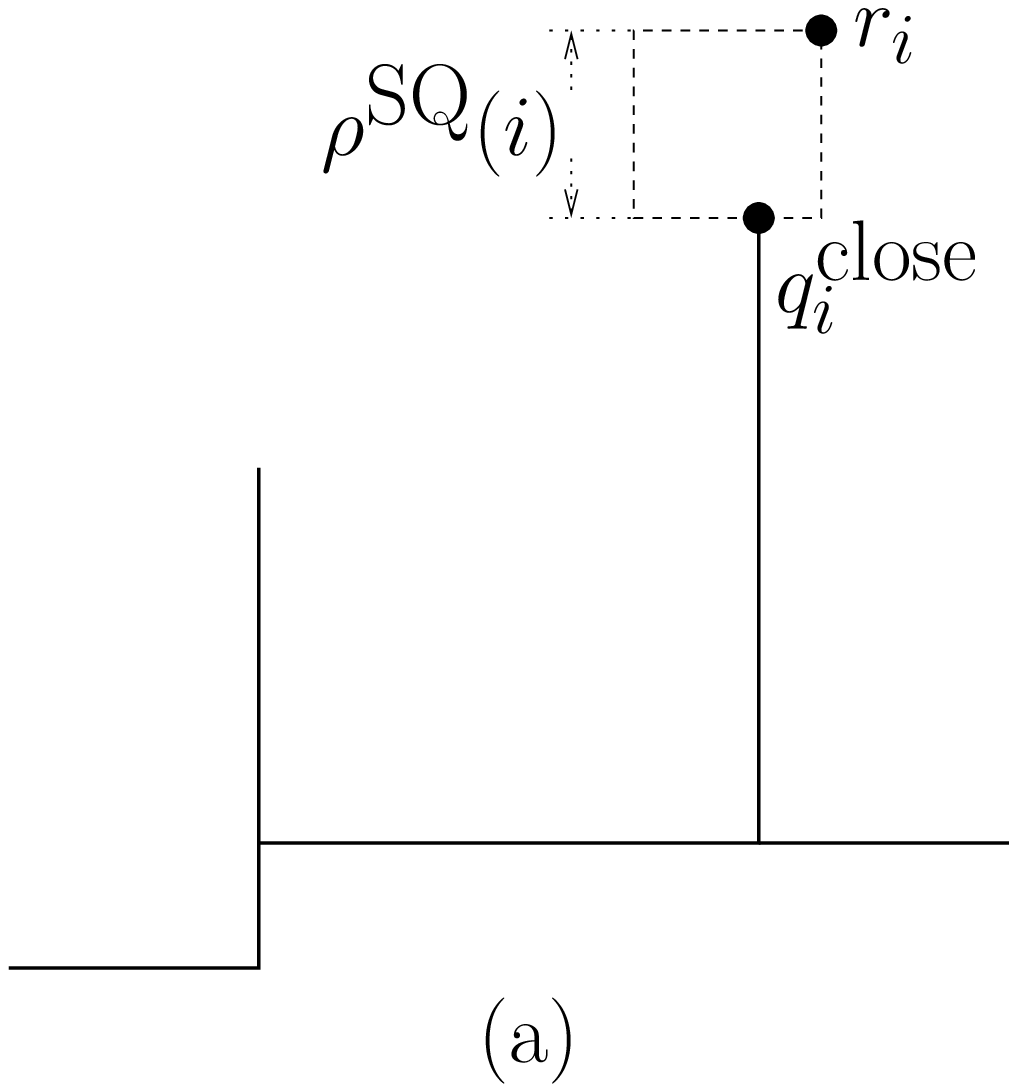}
\hfill
\includegraphics[scale=0.35]{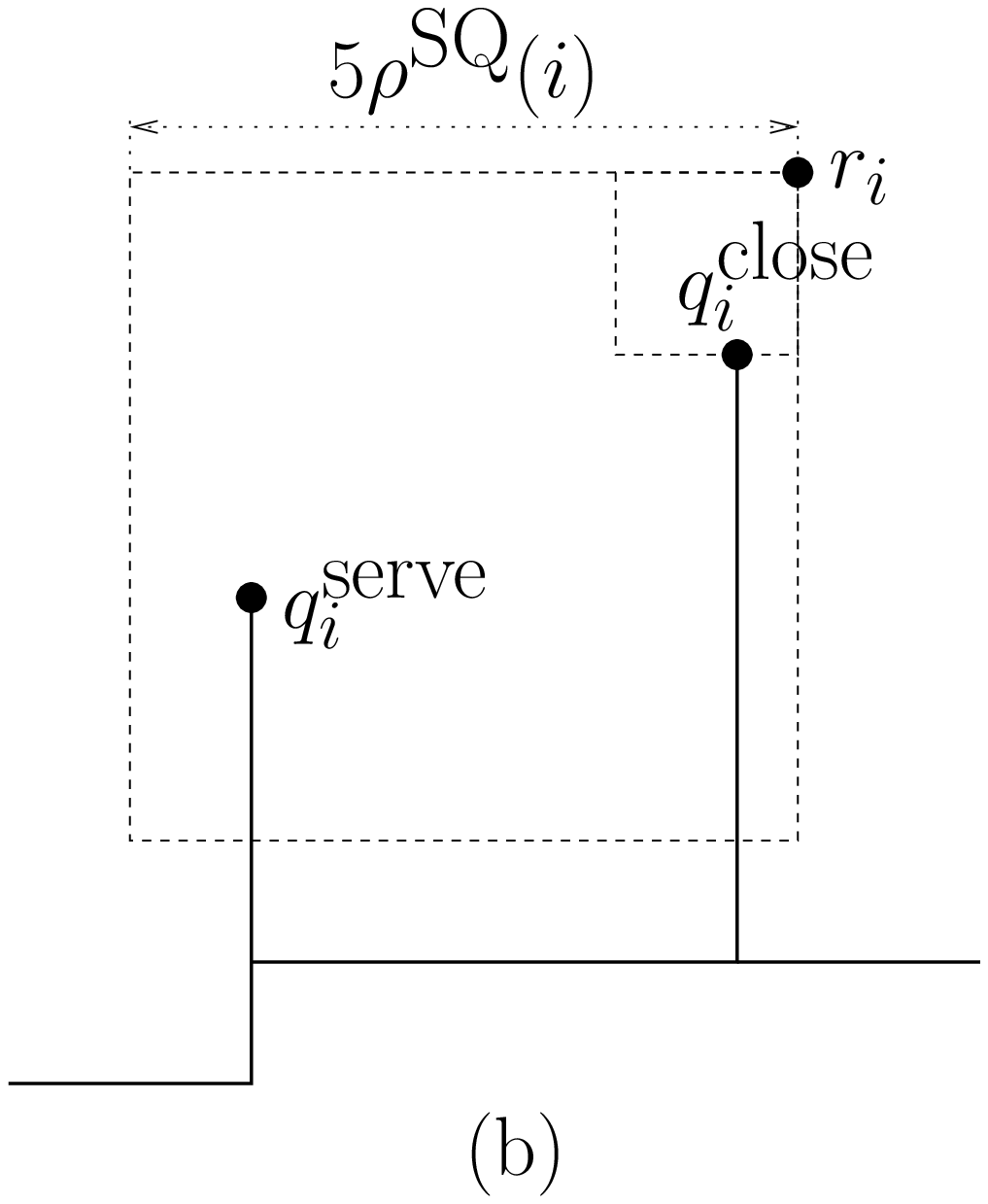}
\hfill
\includegraphics[scale=0.35]{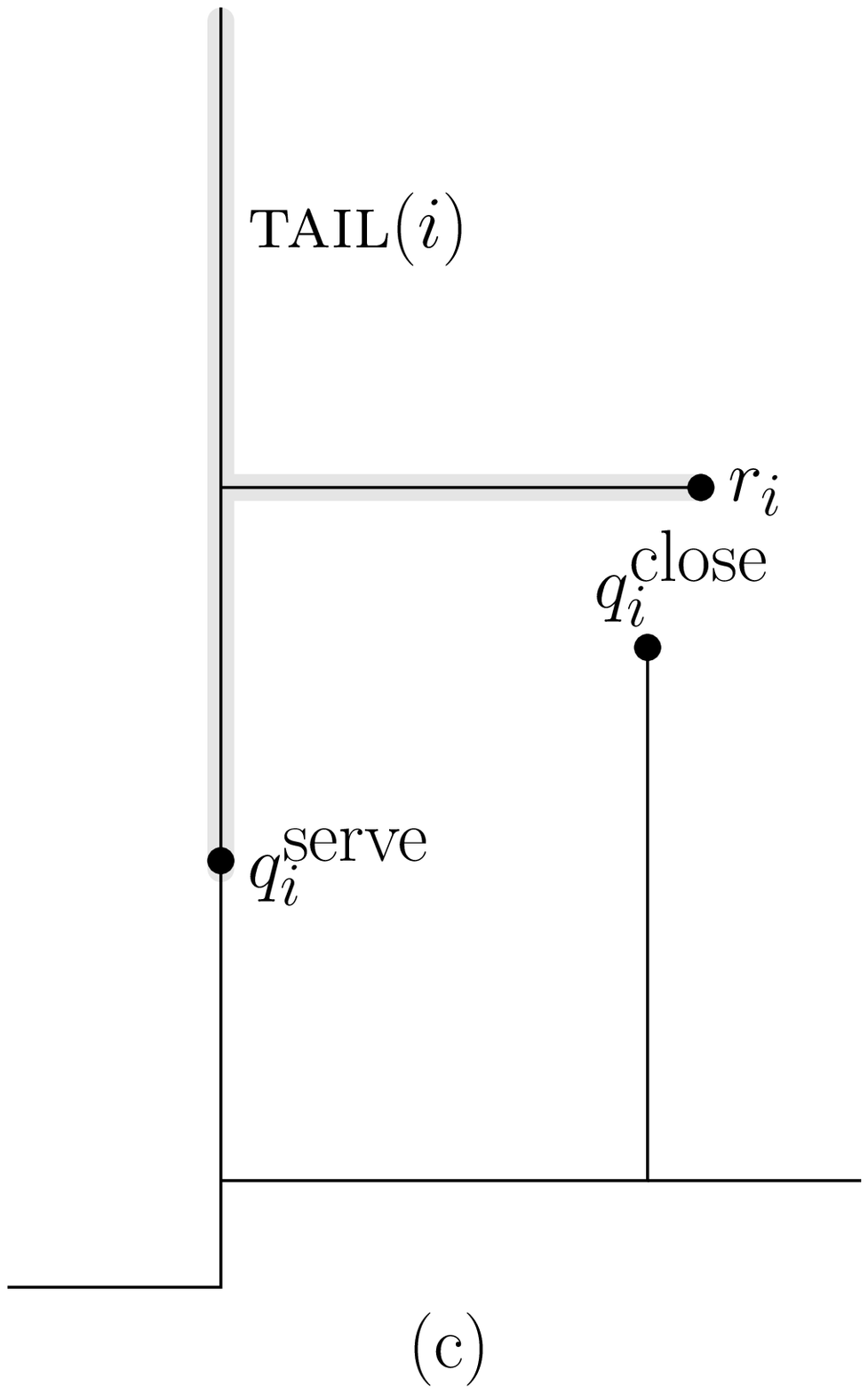}
\end{center}
\caption{\sf $\Square$ execution example.
(a) in case (SQ3), $\Square(i-1)$ ($\Square$'s solution after handling point $r_{i-1}$); $\qclose_i$ defines the radius $\rhoSQ(i)$.
(b) the serving replica $\qSQ_i$ is the leftmost in $\seq{r_i,5\rhoSQ(i)}\cap\Square(i-1)$.
(c) $\Square(i)$.
\label{fig: Square execution SQ3}
}
\end{figure}

\begin{itemize}

\item[(SQ1)] Add the vertical path from $(0,t_{i-1})$ to $(0,t_i)$.

\item[(SQ2)] Let replica $\qclose_i=\rep{\uclose_i}{\sclose_i}$ be such that $\qclose_i$ is already in the solution $\Square(i-1)$ and
    (1) the distance in $L_\infty$ norm from $\qclose_i$ to $\rr_i$ is minimum (over the replicas already in the solution); and
    (2) over those replicas choose the latest,
    that is, $\sclose_i=\max\{t\leq t_i \mid (\uclose_i,t)\in\Square(i-1)\}$.
%
Define the {\em radius} of $\rr_i$ as $\rhoSQ(i)=\distinf{\qclose_i,\rr_i}=\max\{|v_i-\uclose_i|,|t_i-\sclose_i|\}$.
Call $\qclose_i$ the {\em closest} replica of the $i$'th request.

\item [(SQ3)] Choose a replica $\qSQ_i=\rep{\uSQ_i}{\sSQ_i}\in\seq{r_i,5\cdot\rhoSQ(i)}$ such that $\qSQ_i$ is already in the solution $\Square(i-1)$
              and $\uSQ_i$ is the leftmost node (over the nodes corresponding to replicas of $\seq{r_i,5\cdot\rhoSQ(i)}$ that are already in the solution).
              Call $\qSQ_i$ the {\em serving replica} of the $i$'th request.

\item[(SQ4)]
Deliver a copy from $\qSQ_i$ to $\rr_i$ via $(\uSQ_i,t_i)$.
This is done by storing a copy in node $\uSQ_i$
from time $\sSQ_i$ to time $t_i$,
and then delivering a copy from $\rep{\uSQ_i}{t_i}$ to $\rep{v_i}{t_i}$
\footnote{More formally, add the arcs of $\calP_\calA[\rep{\uSQ_i}{\sSQ_i},
\rep{\uSQ_i}{t_i}]$
and the horizontal directed edges of $\calP_\calH[(\uSQ_i,t_i),(v_i,t_i)]$
to the solution.
}.

\item[(SQ5)]
Store a copy in $\uSQ_i$ from time $t_i$ to time $t_i+4\cdot\rhoSQ(i)$
\footnote{
More formally, add the arcs of $\calP_\calA[\rep{\uSQ_i}{t_i},\rep{\uSQ_i}{t_i+4\cdot\rhoSQ(i)}]$
to the solution.
}.

\end{itemize}

\noindent Intuitively, steps SQ1--SQ4 utilize previous replicas in the solution, while
step SQ5 prepares the contribution of $r_i$ to serve later requests.
Note that $\Square$ is not an online algorithm, since in step SQ4,
it may add to the solution some arcs corresponding to previous times.
Such an action cannot be preformed by an online algorithm.
Denote by $\FSQ=\HSQ\cup\ASQ$ the feasible solution $\Square(\NN)$ of $\Square$.
Let $\DBase(i)=\{(u,t_i) \mid  \uSQ_i\leq u\leq v_i\}$ and let $\DBase=\cup_{i=1}^{\NN}\DBase(i)$
(notice that $\DBase\subseteq \FSQ$ because of step SQ4).
Similarly, let $\tail(i)=\{(\uSQ_i,t) \mid t_i\leq t\leq t_i+4\rhoSQ(i)\}$
be the nodes of the path $\calP_\calA[\rep{\uSQ_i}{t_i},\rep{\uSQ_i}{t_i+4\cdot\rhoSQ(i)}]$
(added to the solution in  step SQ5)
and let $\tail=\cup_{i=1}^{\NN}\tail(i)$.
Note that $\FSQ$ is indeed an arborescence rooted at $(0,0)$.

\begin{figure}[ht!]
\fboxsep=0.2cm
\framebox[\textwidth]{
\begin{minipage}{0.95\textwidth}

~$\bullet$ $\FSQ\leftarrow\{(0,0)\}$ is the $\Square$'s solution after handling request $\rr_{i-1}$.

~$\bullet$ When request $\rr_i$ arrives do:
    \begin{enumerate}

    \item $\FSQ\leftarrow\FSQ\cup\calP_\calA[(0,t_{i-1}),(0,t_i)]$.

    \item [2'.]$\rhoSQ(i)\leftarrow\min\{\distinf{q,\rr_i} \mid q\in\FSQ\}$.

    \item Choose a replica $\qclose_i=\rep{\uclose_i}{\sclose_i}$ such that $\qclose_i$ is in $\Square(i-1)$ and\\
     (a) $\distinf{\qclose_i,\rr_i}=\rhoSQ(i)$; and\\
     (b) $\sclose_i=\max\{t\leq t_i \mid (\uclose_i,t)\in\Square(i-1)\}$.

    \item Choose the serving replica $\qSQ_i=(\uSQ_i,\sSQ_i)\in\seq{\rr_i,5\cdot\rhoSQ(i)}\cap\FSQ$ such that
    $\uSQ_i=\min\left\{ u \mid\exists s \mbox{ such that } (u,s)\in\seq{\rr_i,5\cdot\rhoSQ(i)}\cap\FSQ \right\}$.\\
    $\vartriangleright$     $\uSQ_i$ is the leftmost node
    corresponding to the replicas of $\seq{\rr_i,5\cdot\rhoSQ(i)}\cap\FSQ$

    \item $\FSQ\leftarrow\FSQ\cup\calP_\calA[(\uSQ_i,\sSQ_i),(\uSQ_i,t_i)]\cup\calP_\calH[(\uSQ_i,t_i),(v_i,t_i)]$.

    $\vartriangleright$ deliver a copy from $\qSQ_i$ to $\rr_i$ via $(\uSQ_i,t_i)$.

    \item $\FSQ\leftarrow\FSQ\cup\calP_\calA[(\uSQ_i,t_i),(\uSQ_i,t_i+4\rhoSQ(i))]$.

    $\vartriangleright$ leave a copy at $\uSQ_i$ from current time $t_i$ to time $t_i+4\rhoSQ(i)$.

    \end{enumerate}

\end{minipage}
}
\caption{\label{figure:Pseducode Squarem}
Algorithm $\Square$.
}
\end{figure}

\newpage
\subsection{Analysis of $\Square$}
\label{subsec: Analysis Square}

\vspace{-0.2cm}

First, bound the cost of $\Square$  as a function of the radii (defined in SQ2).

\begin{observation}
$\cost(\Square,\calR)\leq 14\sum_{i=1}^{\NN}\rhoSQ(i).$
\label{obser:sqr: cost Square > 14 sum radii}
\end{observation}
\proof
For each request $r_i\in\calR$, Algorithm $\Square$ pay a cost of $10\rhoSQ(i)$ to the path between $r_i$'s serving replica $\qSQ_i$ a $r_i$ itself (see step SQ4) and additional cost of $4\rhoSQ(i)$ for serving a copy to all replicas of $\tail(i)$ (see step SQ5).
\QED

It is left to bound from below the cost of the optimal solution as a function of the radii.

\paragraph*{\bf Quarter balls.}
Our analysis is based on the following notion.
A \emph{quarter-ball}, or a $\Qball$, of \emph{radius}
$\rho \in \naturals$ centered at a replica $q=(v,t)\in\calV$ contains every replica
from which there exists a path of length $\rho$ to $q$
\footnote{
This is, actually, the definition of the geometric place ``ball''.
We term them ``quarter ball'' to emphasize that we deal with directed edges.
That is, it is not possible to reach $(v,t)$ from above nor from the right.
}
.
For every request $r_i\in\calR$, denote by $\SQball(\rr_i,\rhoSQ(i))$
\footnote{
Note that $\SQball(\rr_i,\rhoSQ(i))$ is different from $\seq{\rr_i,\rhoSQ(i)}$, since the first ball considers distances in $L_2$ norm and the last considers distances in $L_\infty$ norm.
}
(also $\SQball(i)$ for short) the quarter-ball centered at $r_i$ with radius $\rhoSQ(i)$.

Intuitively, for every request $r_i\in\calR'$ (where $\calR'$ obey the observation's condition below), $\opt$'s solution starts outside of $\SQball(i)$,
and must reach $r_i$ with a cost of $\rhoSQ(i)$ at least.

\begin{observation}
Consider some subset $\calR'\subseteq\calR$ of requests.
If the $Q$-balls, $\SQball(i)$ and
$\SQball(j)$, of every two requests $\rr_i,\rr_j\in\calR'$  are edges disjoint, then
$|\opt| \geq \sum_{r_i\in\calR'} \rhoSQ(i)$.
\label{obser:sqr:opt geq sum rho_i}
\end{observation}
Intuitively, for every request $r_i\in\calR'$ (where $\calR'$ obey the observation's condition), $\opt$'s solution starts outside of $\SQball(i)$,
and must reach $r_i$ with a cost of $\rhoSQ(i)$ at least.
\proof
Consider some request $r_i \in \calR'$.
Any directed path from $(v_0,0)$ to $r_i$ must enter the quarter ball $\SQball(i)$ of radius $\rhoSQ(i)$ to reach $r_i$.
The length of this path inside the $\SQball(i)$ is $\rhoSQ(i)$.
All the $\Qball$s of $\calR'$ are disjoint, which implies the observation.
\QED

\subsubsection{Covered and uncovered requests}
Consider some request $\rr_i=(v_i,t_i)$ and its serving replica $\qSQ_i=(\uSQ_i,\sSQ_i)$
(see step SQ3).
We say that $\rr_i$  is {\em covered}, if $v_i-\uSQ_i \geq \rhoSQ(i)$ (see SQ2 and SQ3).
Intuitively, this means the solution $\FSQ$ is augmented by the whole top of the square $\Square[r_i,\rhoSQ(i)]$; see Figure \ref{fig: Square execution SQ3} (a) and (b).
Otherwise, we say that $\rr_i$ is {\em uncovered}.
Let $\cover=\{i \mid r_i \mbox{ is a covered request}\}$ and let $\uncover=\{i \mid r_i \mbox{ is an {\bf un}covered request}\}$.
Given Observation \ref{obser:sqr:opt geq sum rho_i},
the following lemma implies that
\begin{eqnarray}
|\opt|\geq \sum_{i\in\cover}\rhoSQ(i).
\label{ineq:sqr: opt geq sum cover radii}
\end{eqnarray}

\begin{lem}
Consider two {\bf covered} requests $\rr_i$ and $\rr_j$.\\
Then, the  quarter balls
$\SQball(i)$ and $\SQball(j)$
are edge disjoint.
\label{lem:Square: covered requests are edges disjoint}
\end{lem}

\begin{proof}
Assume without loss of generality that, $i>j$.
Thus, $\calP_\calH[(\uSQ_j,t_j),(v_j,t_j)]$ (see SQ4) is already in the solution when handling request $i$.
Also, $\rhoSQ(j) \geq v_j-\uSQ_j$, since $r_j$ is covered.
Consider three cases.

\begin{itemize}
\item[{\bf Case 1.}]
  $v_j\leq v_i$, see Figure \ref{fig:Square: covered requests case 1}.
Since, $r_i$ is covered, $\rhoSQ(i)\leq\distinf{r_j,r_i}=\max\{v_i-v_j ~,~ t_i-t_j\}$.
If $v_j\leq v_i-\rhoSQ(i)$, then these two $\Qball$s are edges disjoint.
Otherwise, $v_i-v_j < \rhoSQ(i)$.
Then $\rhoSQ(i)\leq t_i-t_j$ which implies that these two $\Qball$s are edges disjoint.

\item[{\bf Case 2.}] $v_j-\rhoSQ(j)\leq v_i\leq v_j$, see Figure \ref{fig:Square: covered requests case 2}.
Then, also $v_i\geq \uSQ_j$, since $r_j$ is covered.
Thus, in particular, $(v_i,t_j)\in\Square(i-1)$.
Hence, $\rhoSQ(i)\leq\distinf{(v_i,t_j),r_i=(v_i,t_i)} = t_i-t_j$, which implies that these two $\Qball$s are edges disjoint.

\item[{\bf Case 3.}] $v_i<v_j-\rhoSQ(j)$.
The $Q$-ball of $r_j$ is on the right of any possible (radius) $Q$-ball with $r_i$ as a center.
Thus, these $Q$-balls are edges disjoint.

\end{itemize}
\QED
\end{proof}

\begin{figure}[ht]
\begin{center}
\includegraphics[scale=0.37]{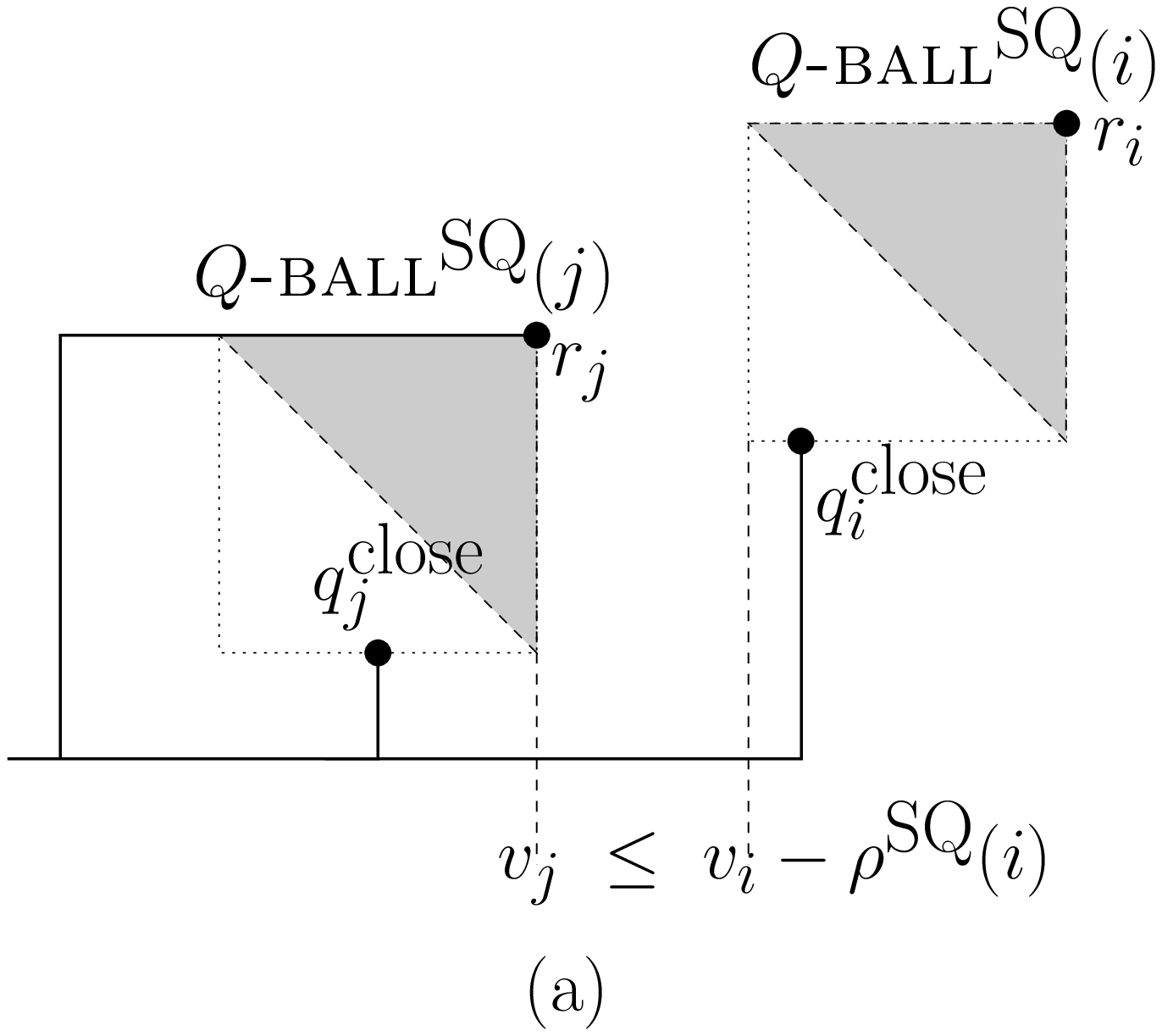}
\hfill
\includegraphics[scale=0.37]{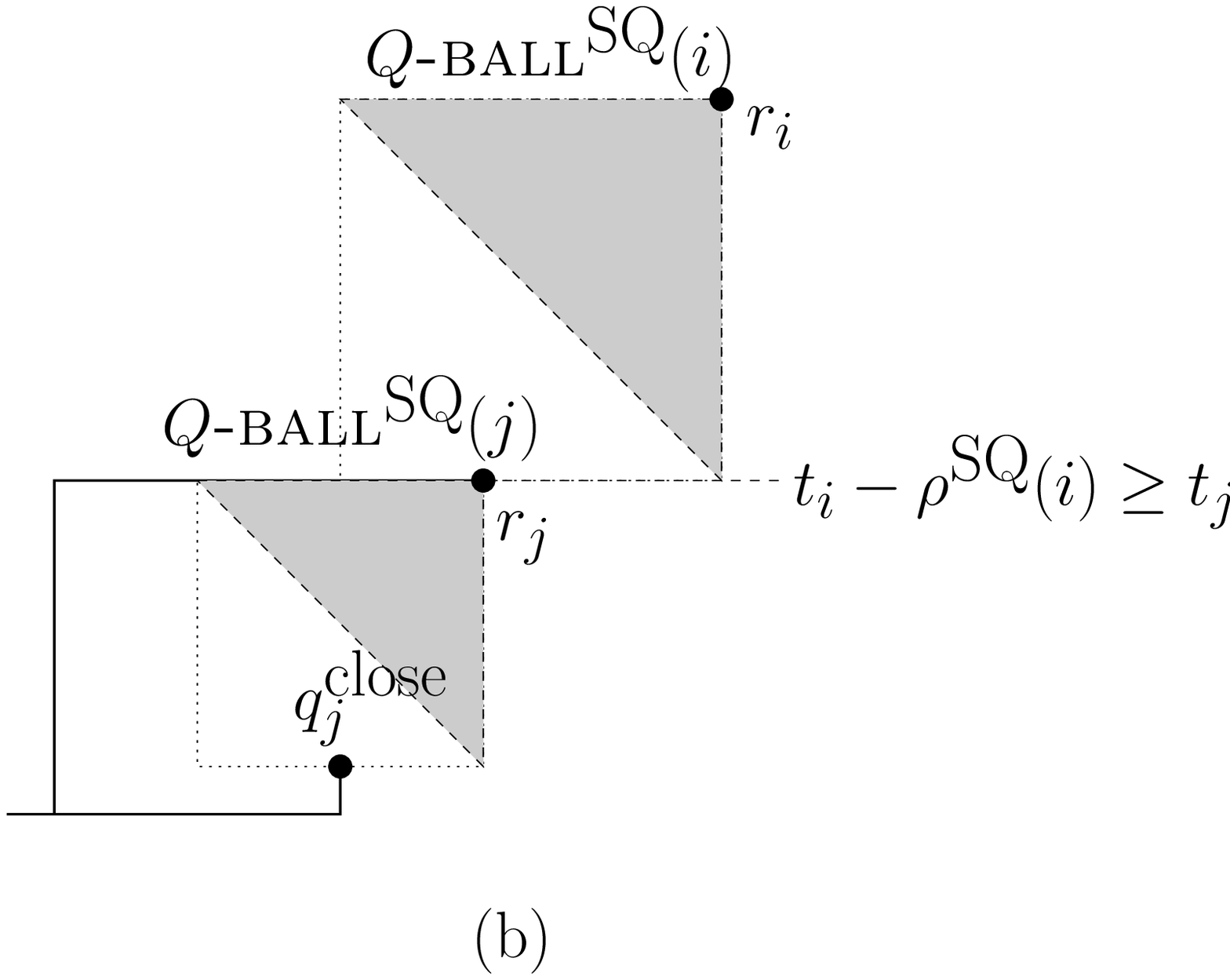}
\end{center}
\caption{\sf Two covered requests are edge disjoint, case 1; (a) $\SQball(i)$ is on the right of $\SQball(j)$, since $v_j\leq v_i-\rhoSQ(i)$;
(b) $v_i-\rhoSQ(i)\leq v_j \leq v_i$ implying that the whole $\SQball(i)$ is above $\SQball(j)$.
\label{fig:Square: covered requests case 1}
}
\end{figure}
\begin{figure}[ht]
\begin{center}
\includegraphics[scale=0.37]{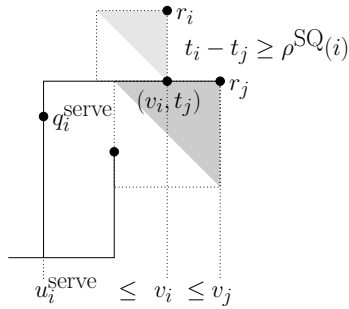}
\end{center}
\caption{\sf
Two covered requests are edge disjoint, case 2.
\label{fig:Square: covered requests case 2}
}
\end{figure}

By the above lemma and observations \ref{obser:sqr: cost Square > 14 sum radii}, \ref{obser:sqr:opt geq sum rho_i},
and Inequality (\ref{ineq:sqr: opt geq sum cover radii}), we have:

\begin{observation}
\label{obs:covered}
$\Square$'s cost for covered requests is no more than $14\cdot\opt$.
\end{observation}
\vspace{0.3cm}
\noindent It is left to bound the cost of $\Square$ for the uncovered requests.

\subsubsection{Overview of the analysis of the cost of uncovered requests}

Unfortunately, unlike the case of covered requests, balls of two {\em uncovered} requests may not be disjoint.
Still, we managed to have a somewhat similar argument that we now sketch.
(The formal analysis appears later in Subsection \ref{subsec: Formal Ana Square}.)
Below, we partition the balls of uncovered requests into disjoint subsets.
Each has a representative request, a {\em root}.
We show that the $\Qball$ of roots {\em are} edge disjoint.
This implies by Observation \ref{obser:sqr: cost Square > 14 sum radii} and Observation \ref{obser:sqr:opt geq sum rho_i}
that the cost $\Square$ pays for the roots is smaller than 14 times the total cost of an optimal solution.
Finally, we show that the cost of $\Square$ for all the requests in each subset is at most twice the cost of $\Square$  for the root of the subset.
Hence, the total cost of $\Square$ for the uncovered requests is also just a constant times the total cost of the optimum.

To construct the above partition, we define the following relation:
ball $\SQball(j)$ is the {\em child} of $\SQball(i)$ (for two {\em uncovered} requests $r_i$ and $r_j$) intuitively,
if the $\SQball(i)$ is the first ball (of a request later then $r_j$) such that $\SQball(i)$ and $\SQball(j)$ are {\em not} edge disjoint.
Clearly, this parent-child relation induces a forest on the $\Qball$s of uncovered requests.
The following observation follows immediately from the definition of a root.

\begin{observation}
\label{obs:roots-distjoint}
The quarter balls of every two root requests are edge disjoint.
\end{observation}
\begin{proof}
Consider two root requests $r_i$ and $r_j$. Assume W.O.L.G that $j<i$.
Also assume, toward contradiction that $\SQball(i)$ and $\SQball(j)$ are {\bf not} edge disjoint.
By the definition of the child parent relationship, either $r_j$ is child of $r_i$, or $r_j$ is a child of some other request $r_{\ell}$ for some $j<\ell<i$.
In both cases, $r_j$ has a parent, hence $r_j$ is not a root request which contradict to choice of $r_j$ as a root request.
The observation follow.
\QED
\end{proof}

\noindent The above observation together with Observation
\ref{obser:sqr:opt geq sum rho_i}, implies the following.

\vspace{0.2cm}
\begin{observation}
\label{obs:roots-cost}
The cost of $\Square$ for the roots is $14\cdot|\opt|$ at most.
\end{observation}
\vspace{0.2cm}

It is left to bound the cost that $\Square$ pays for the balls in each tree (in the forest of $\Qball$s) as a constant function of the cost it pays for the tree root.
Specifically, we show that the sum of the radii of the $\Qball$s in the tree (including that of the root) is at most twice the radius of the root.
This implies the claim for the costs by Observation \ref{obser:sqr: cost Square > 14 sum radii} and Observation \ref{obser:sqr:opt geq sum rho_i}.
To show that, given any non leaf ball $\SQball(i)$ (not just a root), we first analyze
only $\SQball(i)$'s ``latest child'' $\SQball(j)$.
That is,
$
j=\max_k\{\SQball(k) \mbox{ is a child of } \SQball(i)\}.
$
We show
that the radius of the latest child is, at most, a quarter of the radius of $\SQball(i)$.
(See Property (P1) of Lemma \ref{lemma:square:uncoverd req not edegs disjoint} in Subsection \ref{subsec: Formal Ana Square}.)
Second, we show that the {\em sum} of the radii of the rest of the children (all but the latest child) is, at most, a quarter of the radius of $\SQball(i)$ too.
Hence, the radius of a parent ball is at least twice as the sum of its children radii.
This implies that the sum of the radii of all the $\Qball$s in a tree is at most twice the radius of the root.

The hardest technical part here is in the following lemma that,
intuitively, states that ``a lot of time'' (proportional to the request's radius) passes between the time one child ball ends and the time the next child ball starts, see Fig. \ref{fig:sqr:ParentAndChildrenAA}.

\begin{figure*}
\begin{center}
\includegraphics[scale=0.35]{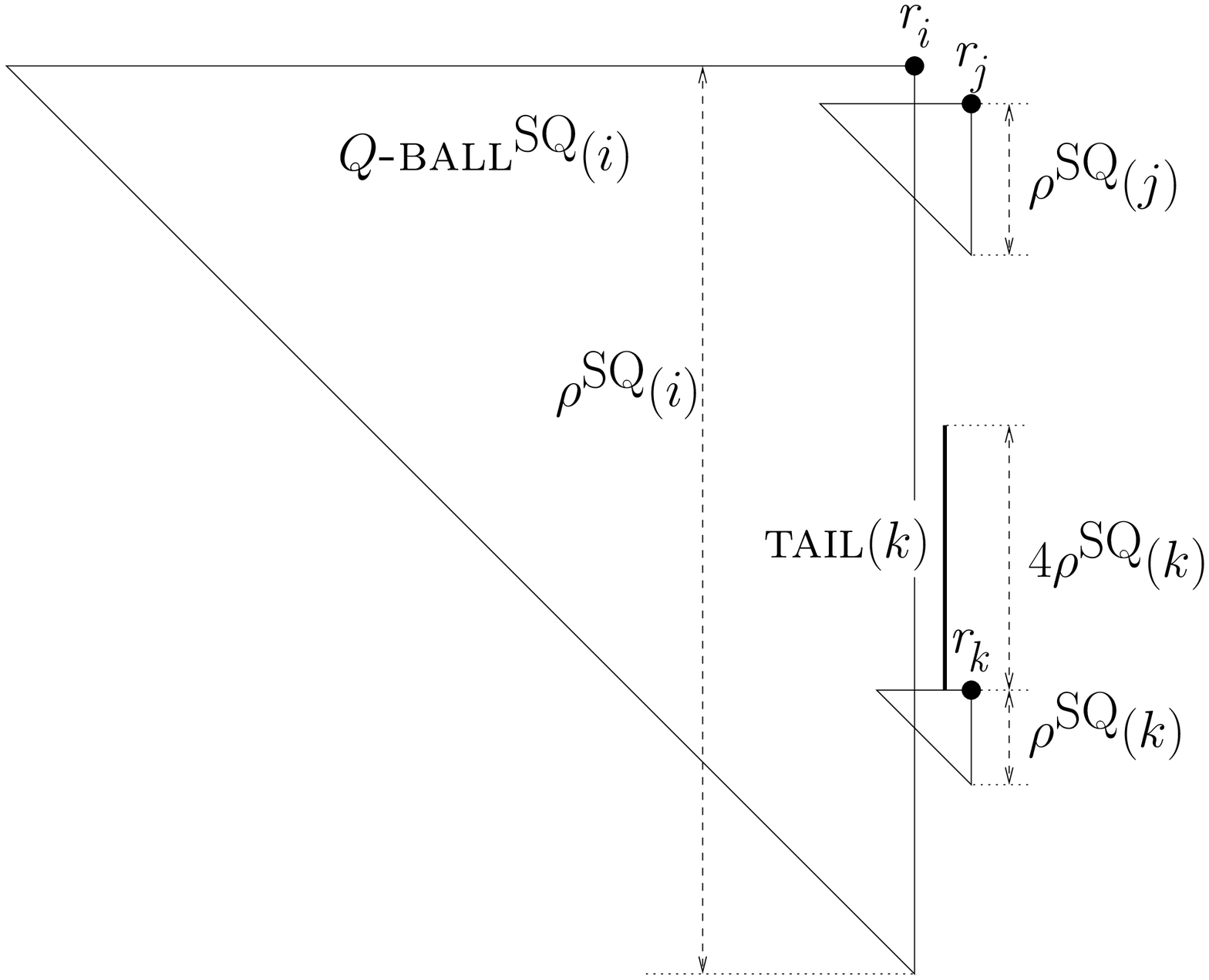}
\end{center}
\caption{\sf Geometric look on a parent $\SQball(i)$ (note that a $\Qball$ is a triangle) and its children $\SQball(j)$ and $\SQball(k)$.
\label{fig:sqr:ParentAndChildrenAA}
}
\end{figure*}

\begin{lem}
Consider some uncovered request $\rr_i$ which has at least two children.
Let $\SQball(j)$, $\SQball(k)$ some two children of $\SQball(i)$, such that $k<j$.
Then,
$t_j-\rhoSQ(j)  \geq t_k + 4 \rhoSQ(k)$. 
\label{lemma:sqr: tj-tk leq 4 radius k}
\label{app:lemma:sqr: tj-tk leq 4 radius k}
\end{lem}
Intuitively, the radius of a parent $\Qball$ is covered by the radii of its children $\Qball$s, plus the tails (see step SQ5) between them.
%
Restating the lemma,
the time of the earliest replica in $\SQball(j)$ is not before the time of the latest replica in
$\tail(k)$.
Intuitively,
%
recall that the tail length of a request is much grater than the radius of the request's $\Qball$.
%
%
%
Hence,  the fact that the
radius of a latest child is at most a quarter of the radius of its parent, together with Lemma \ref{lemma:sqr: tj-tk leq 4 radius k},
imply that the sum of the children’s radii is less than half of the radius of the
parent $\Qball$.

%

%
%
%
%
%
%

%
The full proof of Lemma \ref{lemma:sqr: tj-tk leq 4 radius k} (appears in Subsection \ref{subsec: Formal Ana Square}) uses geometric considerations.
%
%
%
%
%
Outlining the proof, we first establish an additional lemma.
%
%
Given any two requests $r_j$ and $r_\ell$ such that $j>\ell$, the following lemma formalizes the following:
Suppose that the node $v_j$ of request $r_j$ is ``close in space (or in the network)'' to the node $v_\ell$ of another request $r_\ell$.
Then, the whole $\Qball$ of $r_j$ is ``far in time'' (and later) from $r_j$.
%
%
%

\begin{lem}
\label{lemma:suppose two con hold}
Suppose that, $j>\ell$ and $v_j-\rhoSQ(j)+1\leq \uSQ_\ell\leq v_j$.
Then, the time of the earliest replica in $\SQball(j)$ is not before the time of the latest replica in $\tail(\ell)$,
i.e., $t_j-\rhoSQ(j)\geq t_\ell+4\rhoSQ(\ell)$.
\end{lem}
Intuitively, Lemma \ref{lemma:suppose two con hold} follows thanks to the tail left in step SQ5 of $\Square$,
as well as to the action taken in SQ3 for moving $\uSQ$ further left of $\uclose$.
\noindent In the proof of Lemma \ref{lemma:sqr: tj-tk leq 4 radius k}, we show that in the case that two requests $r_k$ and $r_j$ are siblings,
either
{\bf (1)} they satisfy the conditions of Lemma \ref{lemma:suppose two con hold}, or
{\bf (2)} there exists some request $r_\ell$ such that $k<\ell<j$ such that  $r_\ell$ and $r_j$ satisfy the conditions of  Lemma \ref{lemma:suppose two con hold}.
Moreover, the time of the last replica in $\tail(\ell)$ is even later then the time of the last replica in $\tail(k)$.
In both cases, we apply Lemma \ref{lemma:suppose two con hold} to show that
the time of the earliest replica in $\SQball(j)$ is not before the time of the latest replica in $\tail(k)$
as needed for the lemma.
%


To summarize, we show
(1) For {\em covered} requests the cost of $\Square$ is $O(1)$ of $|\opt|$; see Observation \ref{obs:covered}.
(2) For {\em uncovered} requests, we prove
two facts:
(2.a) the $\Qball$s of the root requests are edges disjoint, and hence by
Observation \ref{obs:roots-cost}, the sum of their radii is $O(1)$ of $|\opt|$ too.
(2.b) On the other hand, the sum of root's radii is at least half of the sum of the radii of all the uncovered requests.
This establishes Theorem \ref{thm: square is O(1)-approx} (which prove appears in Subsection \ref{subsec: Formal Ana Square}).

\begin{theorem}
Algorithm $\Square$ is $O(1)$-competitive for
$\DMCD$ under the pseudo online model.
\label{thm: square is O(1)-approx}
\end{theorem}
\def\AppSquareThm{
The ratio for covered request follows Inequality (\ref{ineq:sqr: opt geq sum cover radii}).
For uncovered requests it follows from
Observation \ref{obs:roots-distjoint} and Observation \ref{obser:sqr:opt geq sum rho_i} that
$
|\opt| \geq \sum_{i\in\roots}\rhoSQ(i).
$
Combining this with Lemma
\ref{lema:sqr:root radi geq sum of its children}, we have,
$
2|\opt| \geq \sum_{i\in\uncover}\rhoSQ(i).
$
Thus, also,
$
3|\opt| \geq \sum_{r_i\in\calR}\rhoSQ(i).
$
The Theorem follows from Observation
\ref{obser:sqr: cost Square > 14 sum radii}.
\QED
} 

\subsection{Formal analysis of the cost of uncovered requests}
\label{subsec: Formal Ana Square}

We start with a formal definitions of the forest of parent-child relationships.
\paragraph*{Forest of balls.}

For any uncovered request $\rr_i$, define the following notations.
\begin{enumerate}

\item Let $\parent(i)\eqdf j$ be the minimal index grater than $i$ such that $\SQball(i)$  and $\SQball(j)$ are not edges disjoint, if such exists, otherwise $\parent(i)\eqdf\perp$.

\item $\suns(i) \eqdf \{j \mid \parent(j)=i\}$.

\item $\Tree(i) \eqdf \bigcup_{j\in\suns(i)}\Tree(j)$, if $\suns(i)\not=\emptyset$, otherwise $\Tree(i) \eqdf \{i\}$.

\item $\roots=\{ i \mid \parent(i)=\perp\}$.

\end{enumerate}
(We also abuse the definition and say that request $r_j$ is child of request $r_i$; and $j$ is child of $i$, if $\SQball(j)$ is child of $\SQball(i)$.)
%
%

We now, state four observations about uncovered requests.
The main lemmas use these observations heavily.
Recall that, $\qclose_i=(\uclose_i,\sclose_i)$ is the closest replica of $r_i$ (see SQ2).

\begin{observation}
The radius of an uncovered request is determined by the {\em time difference} from its closest replica.
That is, if a request $\rr_i$ is uncovered, then $\rhoSQ(i)=t_i-\sclose_i$. 
\label{obser:sqr: if ri is uncovered then rad(i)= ti-s'i}
\end{observation}
\proof
If $\rr_i$ is uncovered, then $v_i-\uSQ_i < \rhoSQ(i)$.
In addition, $v_i-\uclose_i \leq v_i-\uSQ_i$, since $\uSQ_i\leq \uclose_i$ (see SQ2).
Thus, $v_i-\uclose_i <\rhoSQ(i)$ too.
Therefore, $\rhoSQ(i)=t_i-\sclose_i$, since $\rhoSQ(i)=\distinf{\qclose_i,r_i}=\max\{v_i-\uclose_i, t_i-\sclose_i\}$ (see SQ2).
\QED

\begin{observation}
The replicas of the ``rectangle-graph'' $[v_i-5\rhoSQ(i),\uSQ_i-1]\times[t_i-5\rhoSQ(i),t_i]$
are not in $\Square(i)$.
\label{obser:sqr: no replicas is in left rectangle}
\end{observation}
\proof
Assume by the way of contradiction that some replica $q=(w,t)\in[v_i-5\rhoSQ(i),\uSQ_i-1]\times[t_i-5\rhoSQ(i),t_i]$ is in $\Square(i)$.
This implies that $v_i-5\rhoSQ(i)\leq w<\uSQ_i$,
contradicting the choice (in step SQ3) of node $\uSQ_i$ of the serving replica $\qSQ_i=(\uSQ_i,\sSQ_i)$ as the leftmost node over all replicas that are in the solution and in $\seq{r_i,5\cdot\rhoSQ(i)}$.
\QED

\begin{observation}
Consider some request $r_i$.
Assume that its closest replica $\qclose_i)$ is added to $\Square$'s solution when handling request $r_j$ (for some $j<i$).
Then, $\sclose_i\geq t_j$ (the time of the $i$'th closest replica is not before the time $t_j$ of $r_j$).
\label{obser:sqr: closest replica is of time later than}
\end{observation}
\proof
The replica $\qclose_i=(\uclose_i,\sclose_i)$ is added to the solution in step SQ4 or in step SQ5 while $\Square$ is handling request $r_j$.
If $\qclose_i$ is added to the solution in step SQ4, then the replica $(\uclose_i,t_j)$ is added to the solution in that step too;
otherwise, $\qclose_i$ is added in step SQ5, and then a replica of $\uclose_i$ at time $t$ (for some time $t>t_j$) is added to the solution.
This implies that $\sclose_i\geq t_j$, see step SQ2 for the selection of $\qclose_i$.
\QED

\begin{observation}
If there exists a replica $(w,t_i)$ in the solution of $\Square(i-1)$ such that $0\leq v_i-w\leq 5\rhoSQ(i)$,
then $r_i$ is a covered request.
\label{obser:sqr: if soff=ti then ri is covered}
\end{observation}
\proof
By the definition,
$
\rhoSQ(i)\leq\distinf{(w,t_i),r_i},
$
since the distance from $w$ to $v_i$ is a candidate for $\rhoSQ(i)$.
The observation now follows from the definition of a covered request.
\QED



\subsubsection{Parent ball in tree larger then its child}

As promised (in the overview), Property (P1) of Lemma \ref{lemma:square:uncoverd req not edegs disjoint}
below implies that a parent ball in tree is at least four times larger than its ``last child''.
%
%
%
%
In fact, the lemma is more general (Property (P2) is used in the proof of other lemmas)%
\footnote{
Actually, this lemma shows that property for any other child too,
but for the other children this is not helpful, since there may be too many of them.
}.

\begin{lem}
Consider two {\bf uncovered} requests $\rr_i$ and $\rr_j$ such that $i>j$.
If $\SQball(i)$ and $\SQball(j)$
are not edges disjoint, then
the following properties hold.\\
(P1) $\rhoSQ(i)\geq 4\cdot\rhoSQ(j)$; and\\
(P2) $v_j-\rhoSQ(j) \leq v_i<\uSQ_j\leq v_j$ (
the leftmost replicas of $\SQball(j)$ are on left of $\rr_i$; $\rr_i$ is on the left of $\rr_j$ and even on the left of the $j$'th serving replica).
\label{lemma:square:uncoverd req not edegs disjoint}
\end{lem}
%
%
%
\proof
We first prove Property (P2).
Since $\SQball(i)$ and $\SQball(j)$ are {\bf not} edges disjoint,
\begin{eqnarray}
 v_j-\rhoSQ(j)<v_i ~~\mbox{ and }~~   v_i-\rhoSQ(i) < v_j~.
 \label{ineq:claim:uncovered req: vi geq vj -radius(j)}
\end{eqnarray}
Since also $i>j$ (see Figure \ref{fig:UnCovercase0}),
\begin{eqnarray}
\rhoSQ(i)=t_i-\sclose_i > t_i-t_j,
 \label{ineq:claim:uncovered radius(i)> ti -tj }
\end{eqnarray}
where the equality below follows from Observation \ref{obser:sqr: if ri is uncovered then rad(i)= ti-s'i},
since $r_i$ is uncovered; and
the inequality holds since, on the one hand,  $\SQball(i)$ and $\SQball(j)$ are not edge disjoint, hence has a common edge;
on the other hand, (1) for every edge in $\SQball(i)$, at least one of its corresponding replicas corresponds to time strictly grater than $t_i-\rhoSQ(i)$;
however, non of the edges of $\SQball(j)$ corresponding to replicas of time strictly grater than $t_j$~.

The left inequality of Property (P2) holds by the left inequality of (\ref{ineq:claim:uncovered req: vi geq vj -radius(j)}).
The right inequality of Property (P2) holds trivially, see step SQ3.
Assume by the way of contradiction that the remaining inequality does not holds, i.e., $v_i\geq \uSQ_j$.
Consider two cases.\\

\begin{itemize}

\item [{\bf Case 1.}] $\uSQ_j \leq v_i\leq v_j$, see Figure \ref{fig:UnCovercase1}.
Then, $(v_i,t_j)$ is in $\Square(j)$ ($\Square$'s solution after handling request $\rr_j$).
Thus, $\rhoSQ(i)\leq\distinf{(v_i,t_j),r_i=(v_i,t_i)} = t_i-t_j$, contradicting Inequality (\ref{ineq:claim:uncovered radius(i)> ti -tj }).

\item [{\bf Case 2.}]  $v_j\leq v_i$, see Figure \ref{fig:UnCovercase2}.
In this case, $\rhoSQ(i)\leq\distinf{r_j,r_i}=\max\{v_i-v_j ~,~ t_i-t_j\}$.
By the second inequality of (\ref{ineq:claim:uncovered req: vi geq vj -radius(j)}), $v_i-v_j < \rhoSQ(i)$.
Hence, $\rhoSQ(i)\leq t_i-t_j$.
Again, this contradicts Inequality (\ref{ineq:claim:uncovered radius(i)> ti -tj }).

\end{itemize}

\noindent
These two cases shows that Property (P2) holds.
We next show that Property (P1) holds too.
For that, consider two cases.

\begin{itemize}

\item [{\bf Case A:}] $\uclose_i < v_j-5\rhoSQ(j)$.
In other words, the closest replica $\qclose_i=(\uclose_i,\sclose_i)$ to $r_i$ is on the left of $\seq{r_j,5\rhoSQ(j)}$
(see Figure \ref{fig:UnCovercase3}).
Recall that the closest replica $\qclose_i$ defines the radius $\rhoSQ(i)$ (see SQ2), i.e.,
$\rhoSQ(i)=\max\{v_i-\uclose_i, t_i-\sclose_i\}$.
We have
(the second inequality bellow follows by substituting $v_i$ using the first inequality of (\ref{ineq:claim:uncovered req: vi geq vj -radius(j)})),
\begin{eqnarray*}
\rhoSQ(i) \geq v_i-\uclose_i\geq (v_j-\rhoSQ(j)) - (v_j-5\rhoSQ(j))=4\rhoSQ(j)
\end{eqnarray*}
as needed for Property (P1).
(Intuitively, since $\SQball(i)$ intersect $\SQball(j)$ on the left, $v_i$'s is at most $\rhoSQ(j)$ left of $v_j$;
however, we assumed that $\uclose_i$ is at least $5\rhoSQ(j)$ left of $v_j$;
hence, $\rhoSQ(i)\geq 4\rhoSQ(j)$.)

\item [{\bf Case B:}] $v_j-5\rhoSQ(j) \leq \uclose_i$.
We have that (see Figure \ref{fig:UnCovercase4}),
\begin{eqnarray}
v_j-5\rhoSQ(j) \leq \uclose_i < \uSQ_j~,
\label{ineq:sqr: vj- 5rhoj leq u'i leq uSQi}
\end{eqnarray}
where the inequality on the right holds by Property (P2) of this lemma.
Assume that $\qclose_i$ is added to the solution when $\Square$ handles some request $r_k$ (for some $k<i$).
By Observation \ref{obser:sqr: closest replica is of time later than}, $\sclose_i \geq t_k$.
If $k\geq j$, then $t_k \geq t_j$, which means that $\sclose_i \geq   t_j$
contradicting Inequality (\ref{ineq:claim:uncovered radius(i)> ti -tj }).
Thus $k<j$.
Therefore, by Observation \ref{obser:sqr: no replicas is in left rectangle},
$\qclose_i\not\in[v_j-5\rhoSQ(j),\uSQ_{j}-1]\times[t_j-5\rhoSQ(j),t_j]$.
However, by Inequality (\ref{ineq:sqr: vj- 5rhoj leq u'i leq uSQi}),
$\uSQ_i\in[v_j-5\rhoSQ(j),\uSQ_j-1]$.
Also by Inequality (\ref{ineq:claim:uncovered radius(i)> ti -tj }), $\sclose_i<t_j$.
Hence, $\sclose_i<t_j-5\rhoSQ(j)$,
implying $\rhoSQ(i)>5\rhoSQ(j)$. Property (P1) follows.

\end{itemize}
As showed above Property (P1) and Property (P2) hold, the lemma follows too.
\QED

\begin{figure}
\begin{center}
\includegraphics[scale=0.35]{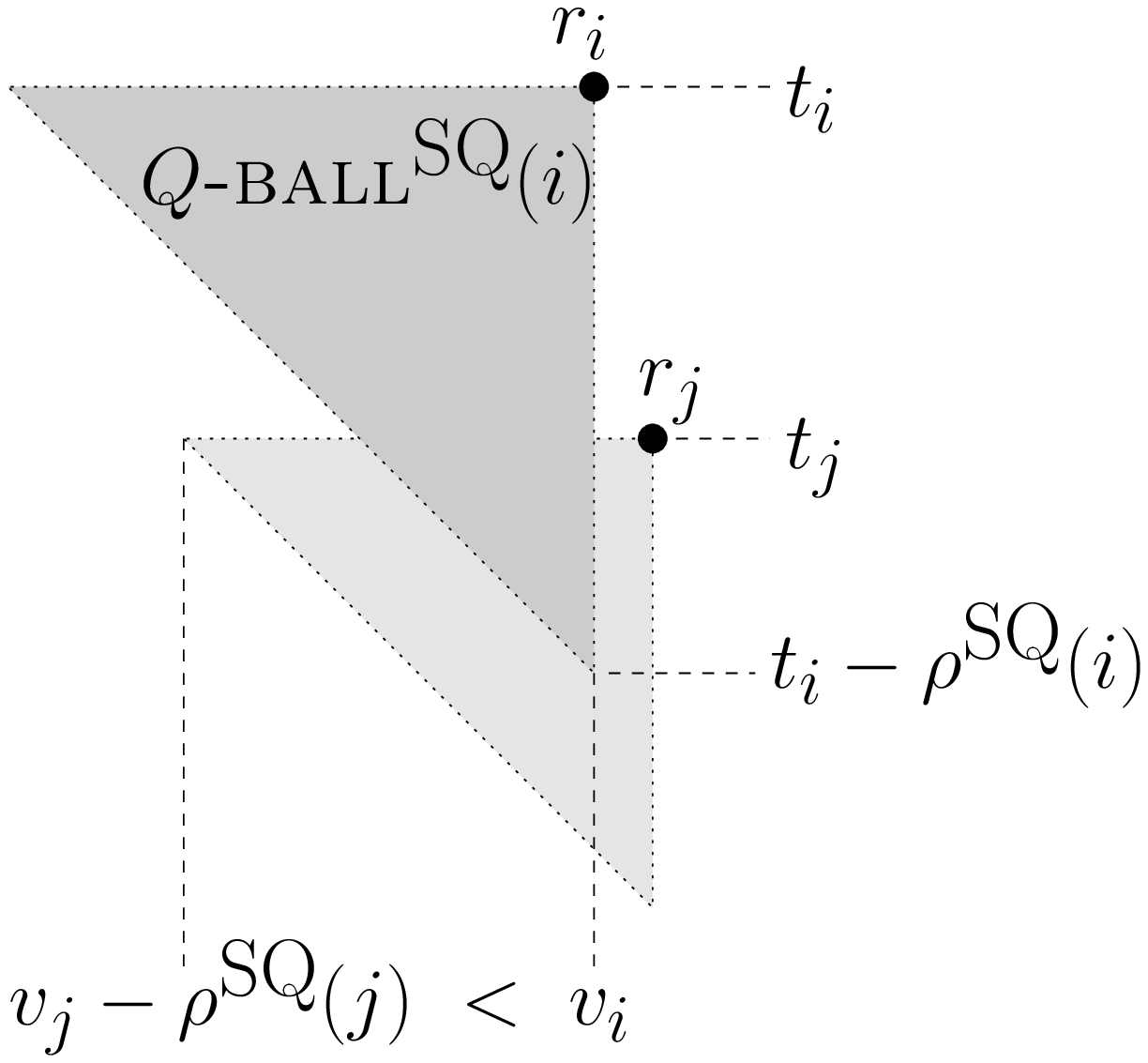}
\end{center}
\caption{\sf
$\SQball(i)$ and $\SQball(j)$ are {\em not} edges disjoint implying inequalities
(\ref{ineq:claim:uncovered req: vi geq vj -radius(j)}) and (\ref{ineq:claim:uncovered radius(i)> ti -tj }).
\label{fig:UnCovercase0}
}
\end{figure}
\begin{figure}[ht]
\begin{center}
\includegraphics[scale=0.35]{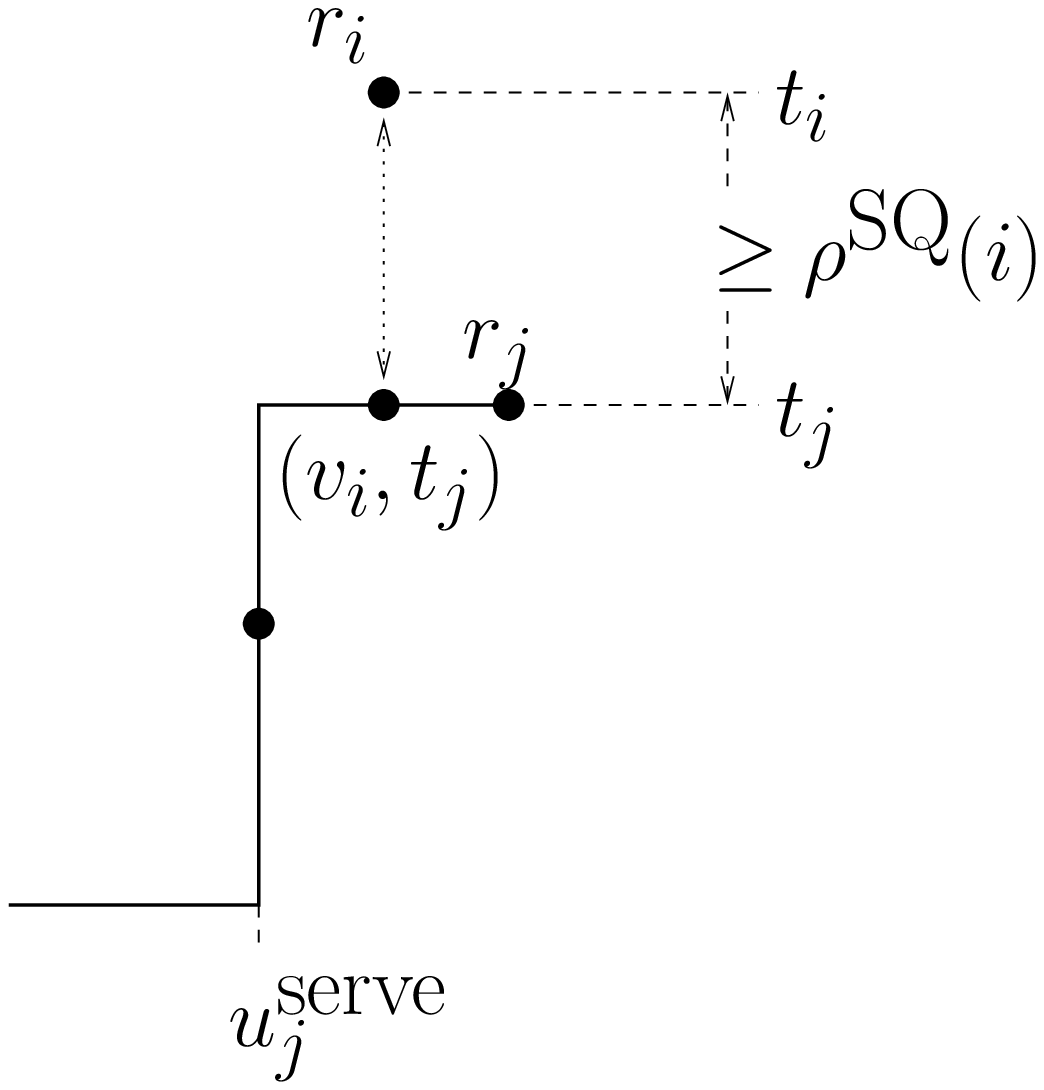}
\end{center}
\caption{\sf
\label{fig:UnCovercase1}
$r_i$ is on the right of $\qSQ_j$ and on the left of $r_j$ (case 1, $\uSQ_j \leq v_i\leq v_j$).
}
\end{figure}
\begin{figure}[ht]
\begin{center}
\includegraphics[scale=0.35]{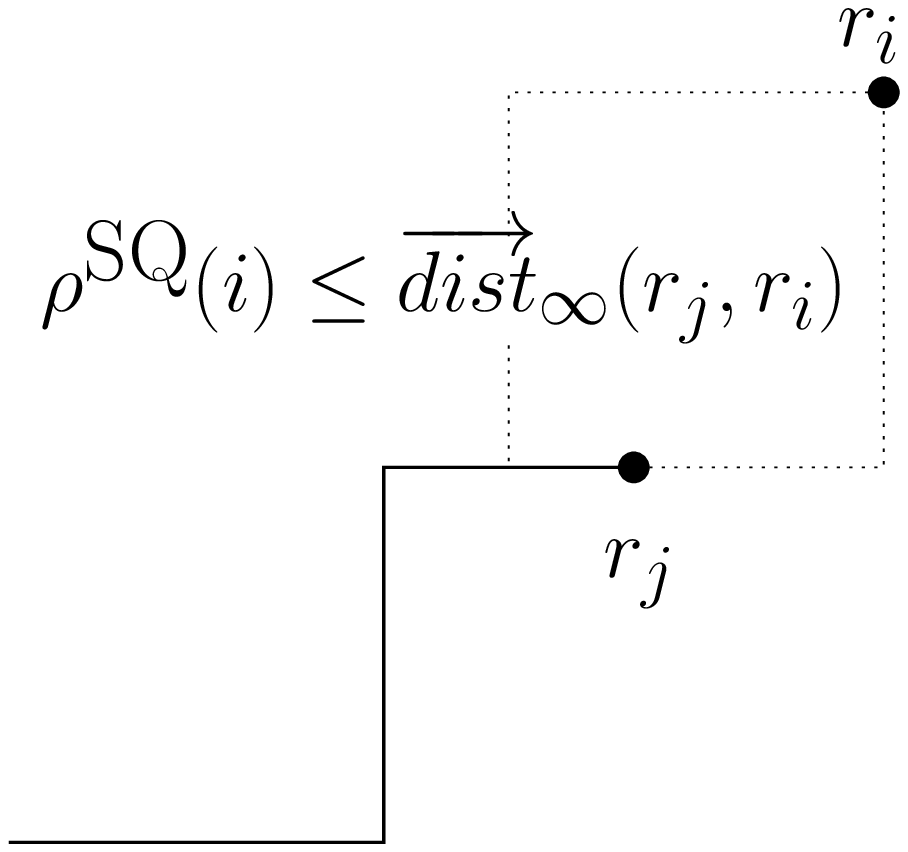}
\end{center}
\caption{\sf
\label{fig:UnCovercase2}
$r_i$ is on the right of $r_j$ (case 2, $v_j\leq v_i$).
}
\end{figure}
\begin{figure}
\begin{center}
\includegraphics[scale=0.35]{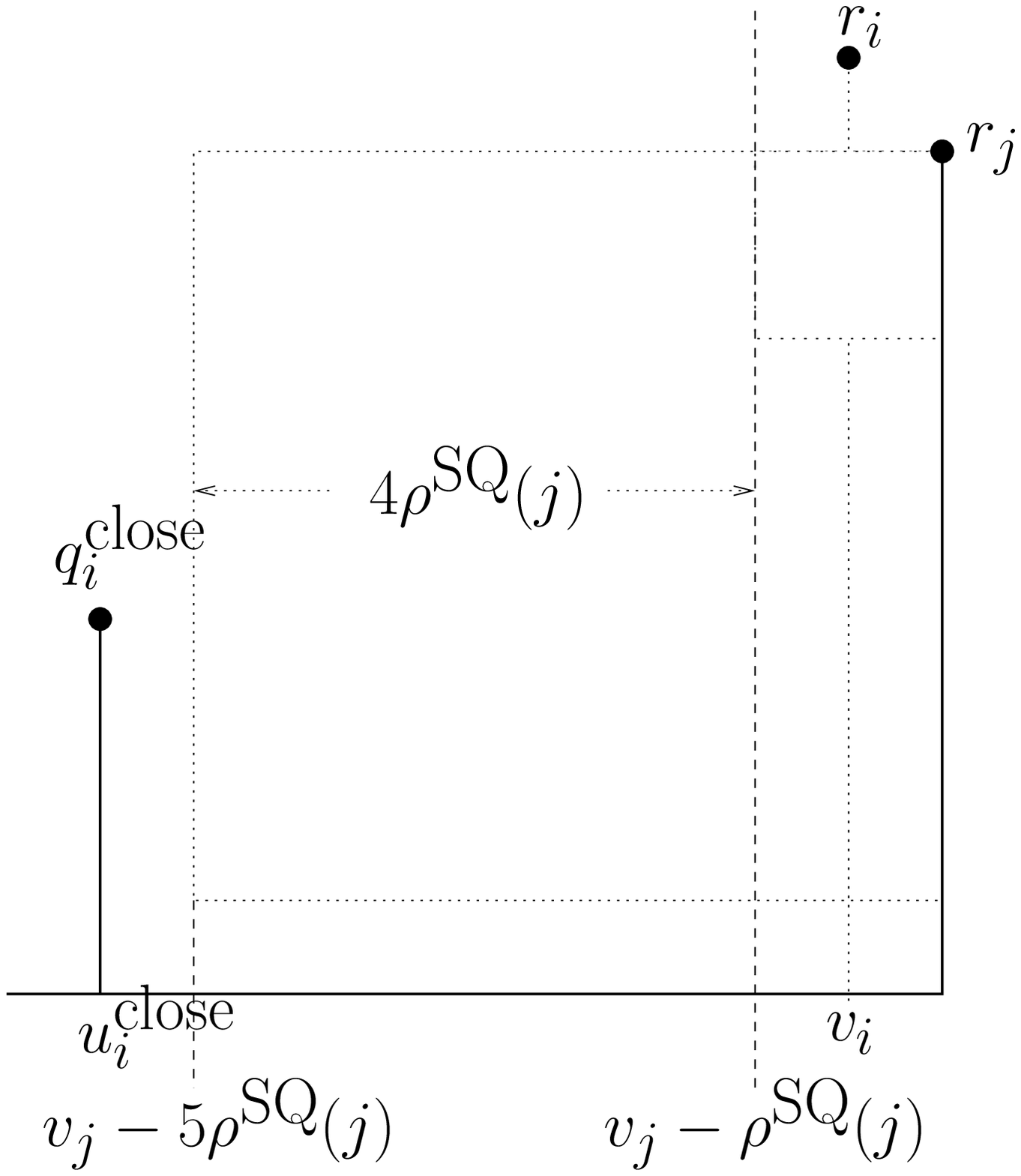}
\end{center}
\caption{\sf
\label{fig:UnCovercase3}
$\qclose_i$ is on the left of $\seq{r_j,5\rhoSQ(j)}$.
}
\end{figure}
\begin{figure}[ht]
\begin{center}
\includegraphics[scale=0.35]{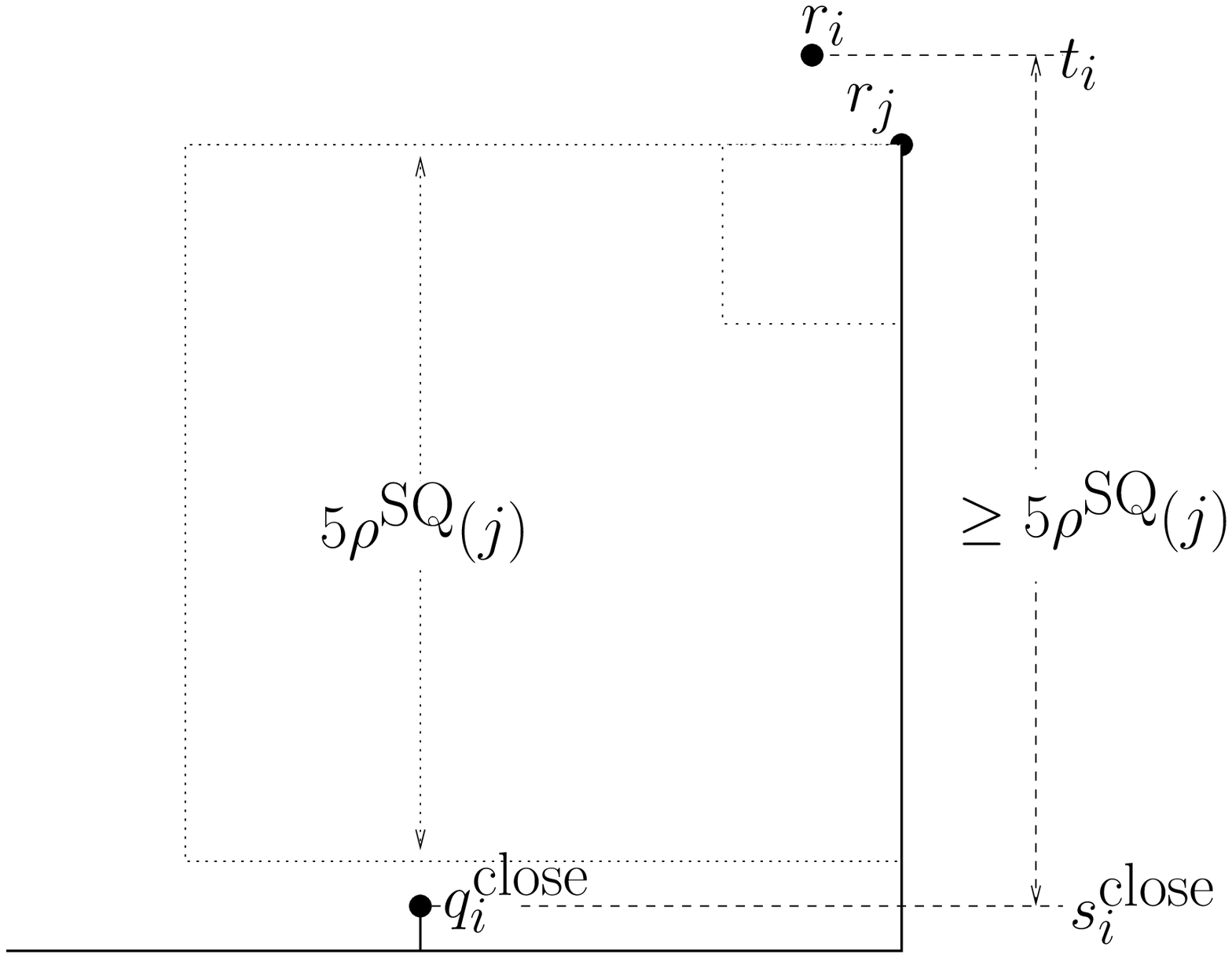}
\end{center}
\caption{\sf
\label{fig:UnCovercase4}
$\qclose_i$ is below $\seq{r_j,5\rhoSQ(j)}$.
}
\end{figure}

\def\TestA{
\begin{figure*}
\centering
\begin{subfigure}[b]{0.25\textwidth}
                \centering
                \includegraphics[width=\textwidth]{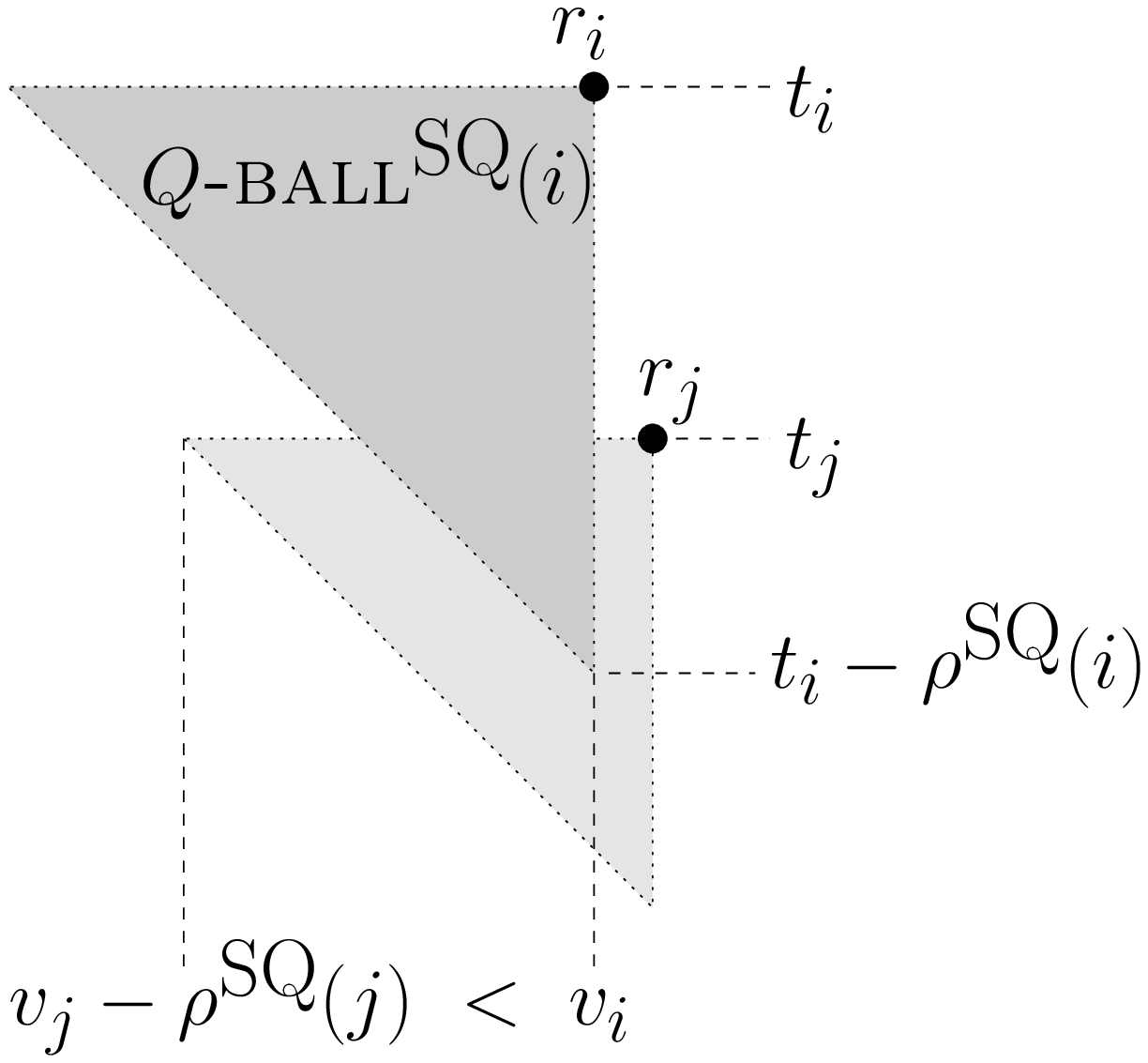}
                \caption{}
                \label{fig:UnCovercase0}
        \end{subfigure}%
\begin{subfigure}[b]{0.25\textwidth}
                \centering
                \includegraphics[width=\textwidth]{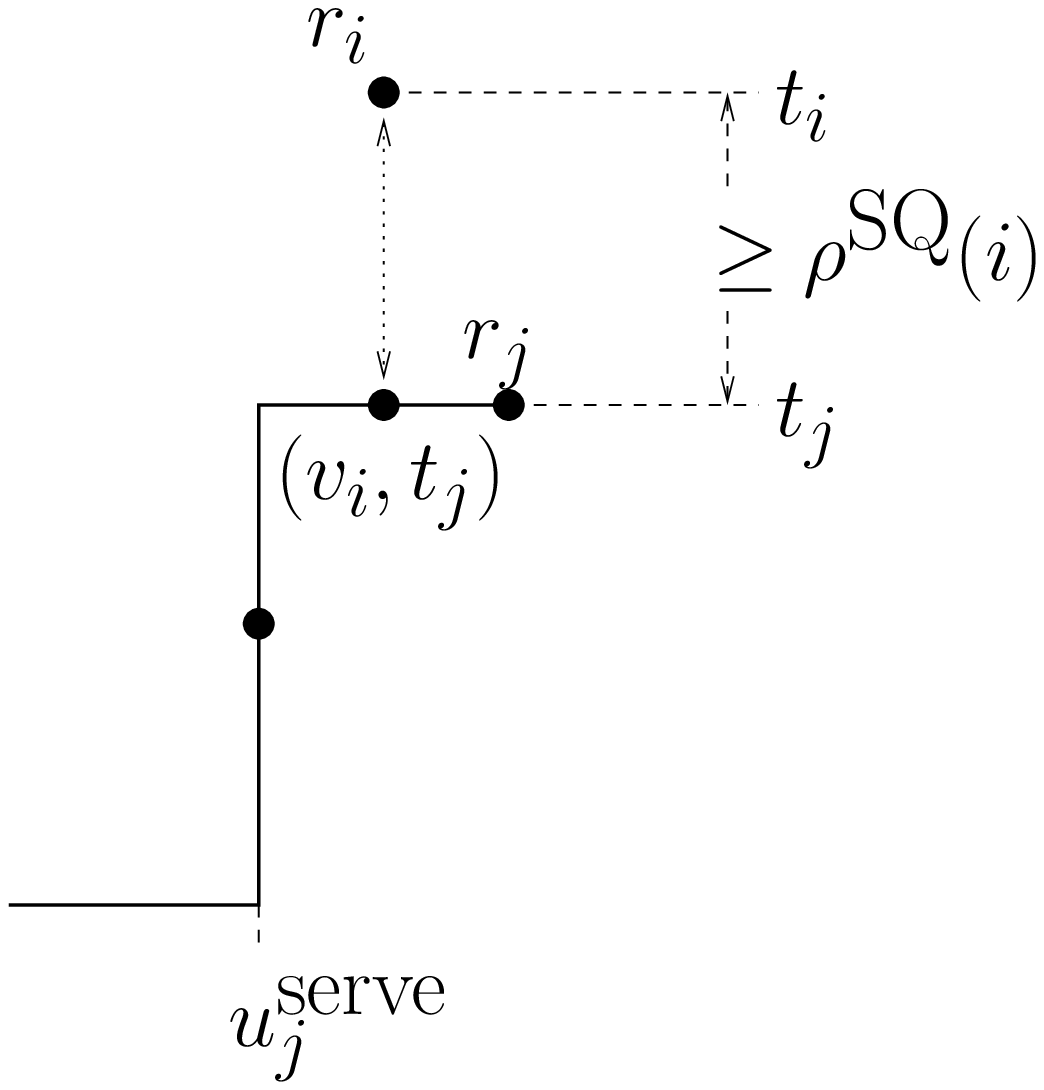}
               \caption{}
                \label{fig:UnCovercase1}
        \end{subfigure}
\begin{subfigure}[b]{0.3\textwidth}
                \centering
                \includegraphics[width=\textwidth]{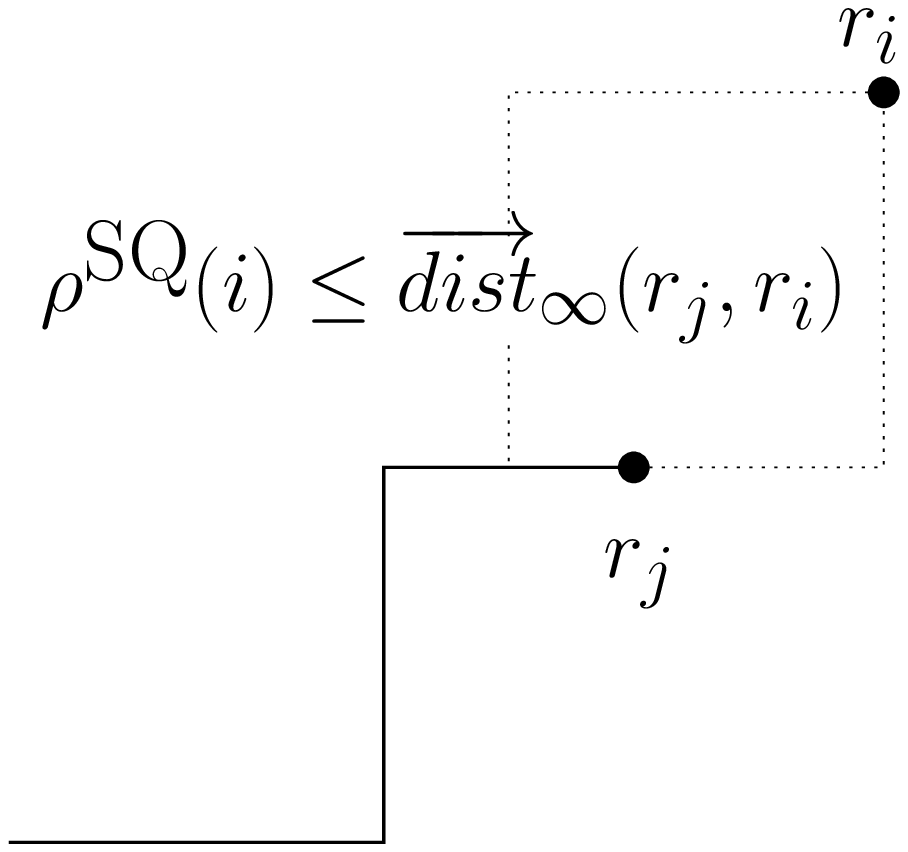}
                \caption{}
                \label{fig:UnCovercase2}
        \end{subfigure}
\begin{subfigure}[b]{0.3\textwidth}
                \centering
                \includegraphics[width=\textwidth]{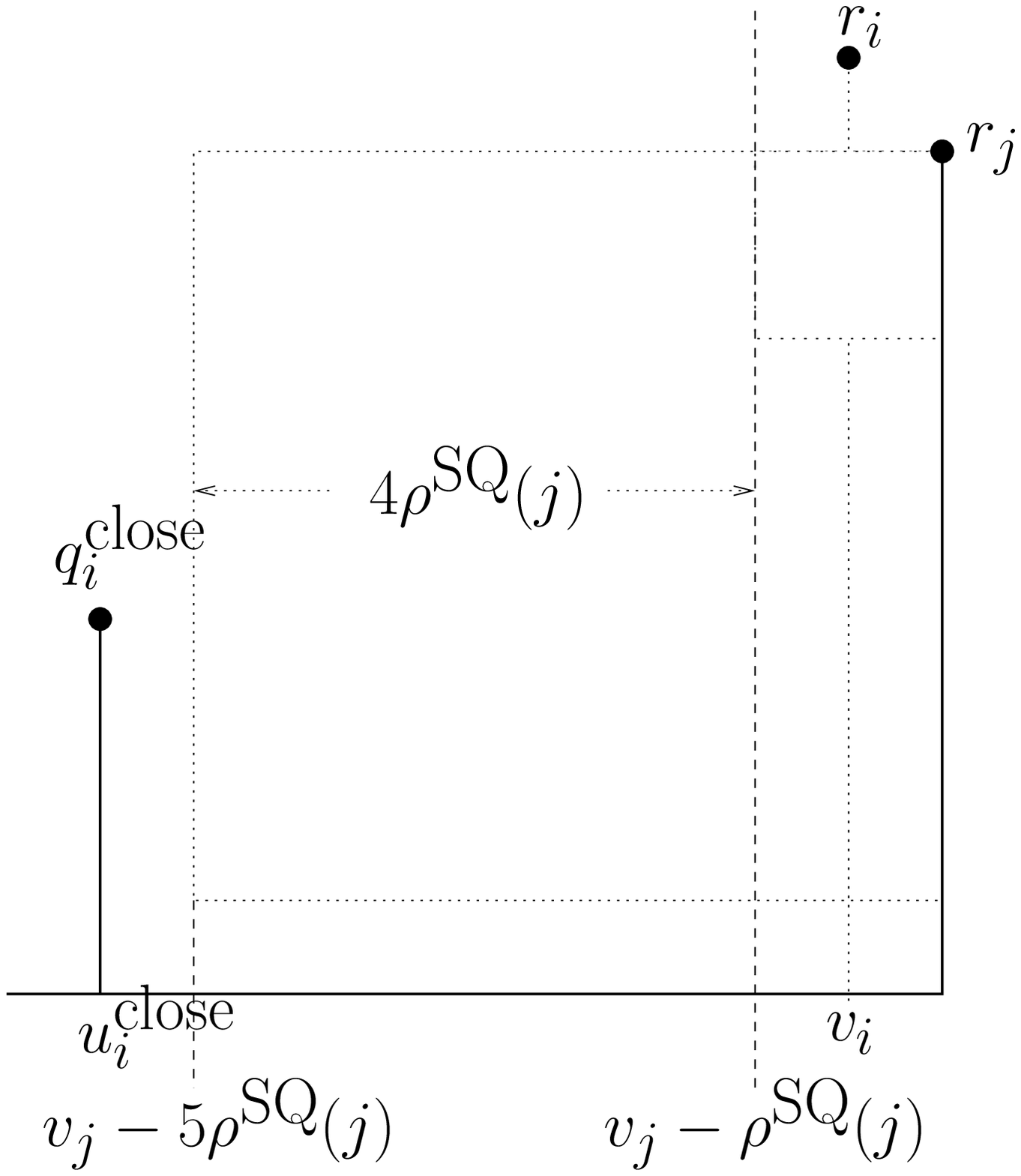}
                \caption{}
                \label{fig:UnCovercase3}
        \end{subfigure}
\begin{subfigure}[b]{0.3\textwidth}
                \centering
                \includegraphics[width=\textwidth]{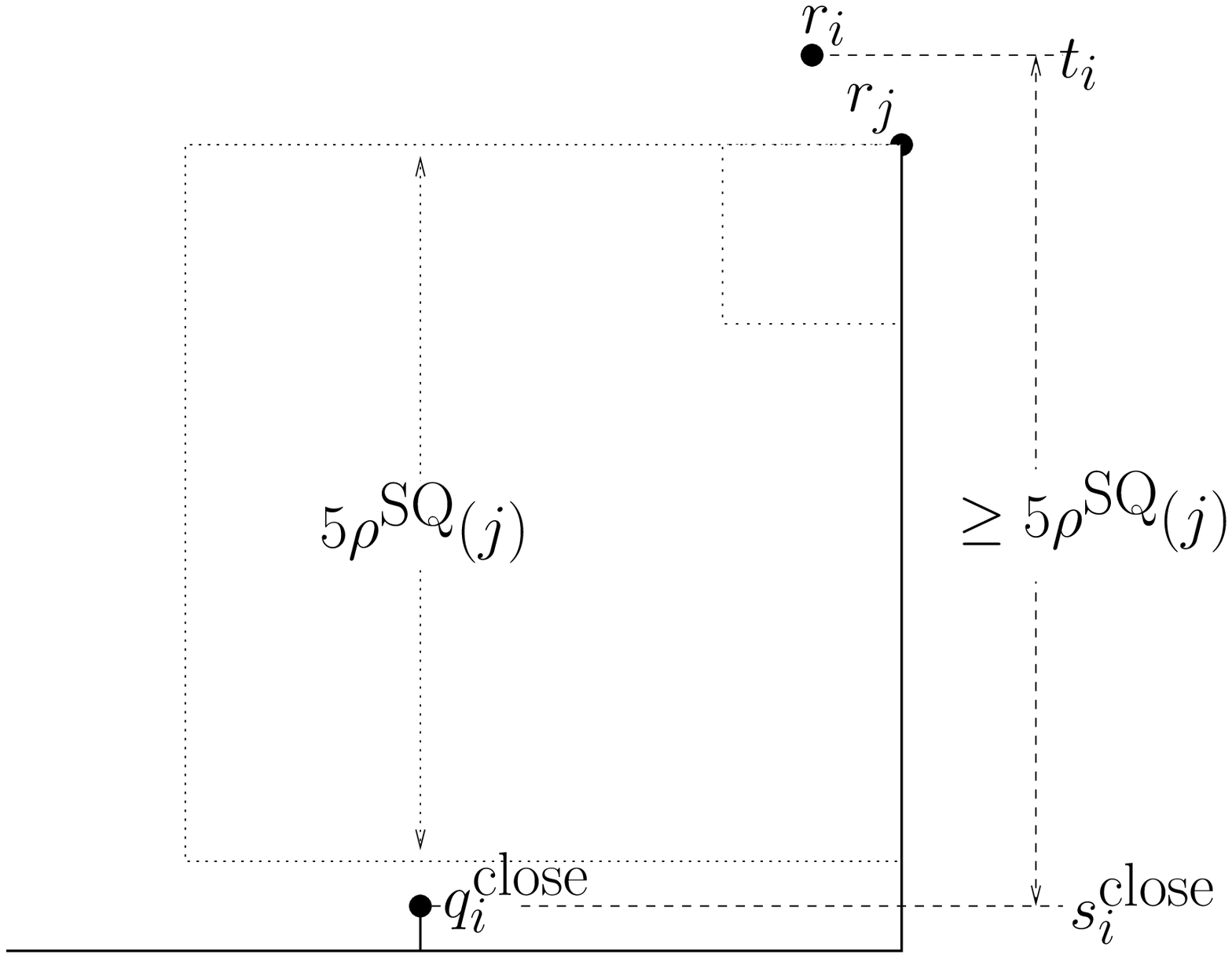}
                \caption{}
                \label{fig:UnCovercase4}
        \end{subfigure}
\caption{\sf
(a) $\SQball(i)$ and $\SQball(j)$ are {\em not} edges disjoint implying inequalities
(\ref{ineq:claim:uncovered req: vi geq vj -radius(j)}) and (\ref{ineq:claim:uncovered radius(i)> ti -tj });
(b) $r_i$ is on the right of $\qSQ_j$ and on the left of $r_j$ (case 1, $\uSQ_j \leq v_i\leq v_j$);
(c) $r_i$ is on the right of $r_j$ (case 2, $v_j\leq v_i$);
(d) $\qclose_i$ is on the left of $\seq{r_j,5\rhoSQ(j)}$; and
(e) $\qclose_i$ is below $\seq{r_j,5\rhoSQ(j)}$.
\label{fig: Square uncovered a}
}
\end{figure*}
%

} 

\newpage

\subsubsection{Uncovered request has at least two children}

The previous lemma suffices for the case that an uncovered request has only one child.
We now consider the case where an uncovered request has at least two children.
We first establish Lemma \ref{lemma:suppose two con hold} (which state in the proof overview)
that deals with the case that the quarter ball of request $r_j$ is later than the tail of some previous request $r_\ell$ (for some $\ell<j$).
Before representing the proof of this lemma, let us make two ``geometric'' definitions.
Consider two given requests $\rr_j$ and $\rr_i$ such that $i>j$.
Intuitively, $\SQball(i)$ is {\em later} than $\tail(j)$, if the time of earliest replica of $\SQball(i)$ is not before the time of the last replica of $\tail(j)$.
Formally, $\SQball(i)$ is later than $\tail(j)$, if $t_i-\rhoSQ(i)\geq t_j+4\rhoSQ(j)$.
In addition, we say that $\tail(j)$ (which contains only replicas of $\uSQ_j$) is in the {\em range} of $\SQball(i)$
(which contains replicas of the nodes of $\{v_i-\rhoSQ(i),...,v_i\}$),
if $v_i-\rhoSQ(i)< \uSQ_j\leq v_i$
(in other words, $\uSQ_j\not= v_i-\rhoSQ(i)$ and there exists a replica of $\uSQ_j$ in $\SQball(i)$).

\newpage

Before presenting the proof of Lemma \ref{lemma:suppose two con hold}, let us remaind and a bit restate this lemma (using formal notations).


\noindent {\bf Lemma \ref{lemma:suppose two con hold}.}
{\em
Consider two requests $\rr_\ell$ and $\rr_j$ such that $j>\ell$.
Suppose that,
$\tail(\ell)$ is in the range of $\SQball(j)$, i.e., $v_j-\rhoSQ(j)<\uSQ_\ell\leq v_j$.
Then, $\SQball(j)$ is later than $\tail(\ell)$. That is,
$t_j-\rhoSQ(j) \geq t_\ell+4\rhoSQ(\ell).$
}

\noindent {\bf Proof of Lemma \ref{lemma:suppose two con hold}:}
Consider two requests $\rr_j$ and $\rr_\ell$ that satisfy the conditions of the lemma.
We begin by showing a slightly weaker assertion, that $r_j$ itself is later than $\tail(\ell)$.
That is, $t_j>t_\ell+4\rhoSQ(\ell)$.
Assume the contrary, that $t_j \leq t_\ell+4\rhoSQ(\ell)$.
Note that
the replicas of $\tail(\ell)$ of time no later than $t_j$ (if such do exists)
are ``candidates'' for the closest and the serving replicas of the $j$'th request
(since they belong to the solution $\Square(j-1)$).
Thus,
$$
\rhoSQ(j)\leq \distinf{(\uSQ_\ell,t_j), r_j}= v_j-\uSQ_\ell~.
$$
That is, the inequality holds since $(\uSQ_\ell,t_j)\in\Square(j-1)$ (see step SQ2);
the equality holds since $r_j=(v_j,t_j)$ and $\uSQ_\ell \leq v_j$.
This means that the complete $j$'th quarter-ball is on the right of the $\ell$'th serving replica $\qSQ_\ell$, i.e.,
$$\uSQ_\ell \leq v_j - \rhoSQ(j).$$
This contradicts the condition of the lemma, hence $t_j>t_\ell+4\rhoSQ(\ell)$ as promised.

We now prove the lemma's assertion that $t_j-\rhoSQ(j) \geq t_\ell+4\rhoSQ(\ell)$.
Denote by $\qlast_{\ell}$ the latest replica in $\tail(\ell)$, i.e., $\qlast_{\ell}=(\uSQ_\ell,t_\ell+4\rhoSQ(\ell))$.
Note that $\qlast$ is a candidate for the closest replica of the $j$'th request, since $\qlast_\ell\in\Square(j-1)$
and the time of $\qlast$ is earlier than the time of $r_j$ (i.e., $t_j>t_\ell+4\rhoSQ(\ell)$).
Thus, the radius $\rhoSQ(j)$ of the $j$'th request is at most as the distance between $\qlast_\ell$ to $r_j$, see step SQ2.
That is,
\begin{eqnarray}
\label{ineq: radi j ;eq dist top tailk}
\rhoSQ(j)\leq \distinf{\qlast_\ell~,~r_j }.
\end{eqnarray}
In addition, by the condition of the lemma,
\begin{eqnarray}
\label{ineq: vj - serv k leq radi(j)}
v_j-\uSQ_\ell < \rhoSQ(j).
\end{eqnarray}
Thus, by Inequalities (\ref{ineq: radi j ;eq dist top tailk}) and (\ref{ineq: vj - serv k leq radi(j)})
\begin{eqnarray}
\label{ineq: vj -uell < distinf qlast, rj}
v_j-\uSQ_\ell  < \distinf{\qlast_\ell~,~r_j }.
\end{eqnarray}
Recall that, $ \distinf{\qlast_\ell~,~r_j } = \max\{ v_j-\uSQ_\ell~,~t_j-(t_\ell+4\rhoSQ(\ell))\}$.
Hence, by Ineq. (\ref{ineq: vj -uell < distinf qlast, rj}),
$$ \distinf{\qlast_\ell~,~r_j } = t_j-(t_\ell+4\rhoSQ(\ell)).$$
Combining this with Ineq. (\ref{ineq: radi j ;eq dist top tailk}), we get also that $\rhoSQ(j)\leq t_j-(t_\ell+4\rhoSQ(\ell))$.
%
The Lemma follows.
\QED

Now, we are ready to show the main lemma (Lemma \ref{lemma:sqr: tj-tk leq 4 radius k}), which
%
intuitively, shows that ``a lot of time'' (proportional to the request's radius)
passes between the time the one child ball ends and the time the next child ball starts.

We begins by remaining this lemma and restate it a bit (using formal notations).

\noindent {\bf Lemma \ref{lemma:sqr: tj-tk leq 4 radius k}.}
{\em
Consider some uncovered request $\rr_i$ such that $|\suns(i)|\geq 2$.
Let $j,k\in\suns(i)$ such that $k<j$.
Then, $\SQball(j)$ is later than $\tail(k)$.
That is, $t_j-\rhoSQ(j)  \geq t_k + 4 \rhoSQ(k)$. 
}





\vspace{0.2cm}
\noindent{\bf Proof of Lemma \ref{lemma:sqr: tj-tk leq 4 radius k}:}
We consider two cases regarding the relation between the serving replica $\uSQ_k$ of the $k$'th request and the node $v_j$ of the $j$'th replica.

\vspace{0.5cm}
\noindent{\bf Case 1: } $\uSQ_k\leq v_j$. This is the simpler case.
Apply Lemma \ref{lemma:suppose two con hold} with the requests $j$ and $\ell=k$.
First, note that $j>k$ as required to apply Lemma  \ref{lemma:suppose two con hold}.
To use this lemma, it is also required to show that
\begin{eqnarray}
\label{ineq: vj-rhoSQj < vk leq vj}
v_j-\rhoSQ(j)<\uSQ_k\leq v_j.
\end{eqnarray}
The right inequality holds by the assumption of this case.
The left inequality holds
since
$$
v_j-\rhoSQ(j)<v_i\leq \uSQ_k,
$$
where the first inequality holds by Lemma \ref{lemma:square:uncoverd req not edegs disjoint} Property (P2) with $i$ and $j$;
the second inequality holds by Lemma \ref{lemma:square:uncoverd req not edegs disjoint} Property (P2) with $i$ and $k$.
Thus, in this case, the lemma follows by Lemma \ref{lemma:suppose two con hold}.

\vspace{0.5cm}

\noindent{\bf Case 2:} $v_j<\uSQ_k$ (that is, $v_j$ is on the left of the $k$'th serving replica $\uSQ_k$).
Note that, unlike the previous case, $\tail(k)$ is {\em not} in the range of $\SQball(j)$.
Thus, the condition of Lemma \ref{lemma:suppose two con hold} does not holds,
and we cannot apply Lemma \ref{lemma:suppose two con hold} with $j$ and $\ell=k$.
Fortunately, we show that in this case, we can use another request for which Lemma \ref{lemma:suppose two con hold} can be applied.
That is, we claim that
in this case, there exists  a request $\rr_\ell$ that has the following three properties.
\begin{enumerate}
\item [(P1)] $k<\ell<j$;

\item [(P2)] $\tail(\ell)$ is in the range of $\SQball(j)$ (it satisfies the condition of Lemma  \ref{lemma:suppose two con hold}); and

\item [(P3)] $\rhoSQ(\ell)\geq\rhoSQ(k)$.

\end{enumerate}
Note that if indeed such a request $\rr_\ell$ (that has the above three properties) does exists,
then
applying Lemma \ref{lemma:suppose two con hold},
we will get that
$$
t_j-\rhoSQ(j)\geq t_\ell+4\rhoSQ(\ell)\geq t_k+4\rhoSQ(k).
$$
The last inequality follows from Property (P3) and since $t_j\geq t_\ell$ (since $j>\ell$).
This will imply the lemma.
%
It is left to show that such a request $\rr_\ell$ {\bf must} exist.
Let
$$\reck=[v_i,v_j]\times[t_k-4\rhoSQ(k), t_i],$$
and let
$$\ell^*= \min_l \Square(l)\cap\reck\not=\emptyset,$$
the index of the first request in which the solution $\Square(\ell^*)$ contains some replicas in $\reck$.
Note that $\reck$ is well defined since $v_i>v_j$ by Lemma \ref{lemma:square:uncoverd req not edegs disjoint}, Property (P2).
We completes the proof by showing
that  $\ell^*$ exists and has  properties (P1)-(P3).
Hence, we can choose $r_\ell=r_{\ell^*}$ and the lemma will follow.

\begin{enumerate}

\item {\bf $\ell^*$ has Property (P1), i.e., $k<\ell^*<j$.}\\
%
%
%
We first show that $\Square(k)$ does not contain  any replica from the rectangle graph $\reck$.
That is,
\begin{eqnarray}
\label{ineq: Square(k) cap reck = emptyset}
\Square(k)\cap\reck=\emptyset.
\end{eqnarray}
Then, we show that $\Square(j-1)$ does contain some replicas from the rectangle graph $\reck$.
That is,
\begin{eqnarray}
\label{ineq: Square(j-1) cap reck not= emptyset}
\Square(j-1)\cap\reck\not=\emptyset.
\end{eqnarray}
Once we prove the above two inequalities, they will imply that $\ell^*$ does exist, and in particular, $k<\ell^*<j$ as needed.

\begin{enumerate}
\item[] {\bf Proving Ineq (\ref{ineq: Square(k) cap reck = emptyset}):}
Note that when Algorithm $\Square$ handles $r_k$, it does not add any replica in the above rectangle,
since it only adds replicas on the right hand side of $\uSQ_k$.
(Recall that, $v_i<v_j$ and we are now analysing case (2) where $v_j<\uSQ_k$, i.e., $[v_i,v_j]\subseteq[v_i,\uSQ_k-1]$.)

\vspace{1cm}
It is left to prove that $\Square(k-1)$ does not include a replica in $\reck$.
By Observation \ref{obser:sqr: no replicas is in left rectangle},
it follows that $\Square(k)$ and $\Square(k-1)$ do not contain  any replica from the ``bottom part'' of $\reck$,
since
$$
\{(v,t)\in\reck \mid t\leq t_k\}
\subseteq
[v_k-4\rhoSQ(k),\uSQ_k-1]\times[t_k-5\rhoSQ(k),t_k]
,$$
where the inequality holds since
$v_k-4\rhoSQ(k) \leq v_i$ (by Lemma \ref{lemma:square:uncoverd req not edegs disjoint}, Property (P2)); and
$v_j < \uSQ_k$ (the assumption of case (2)).

It is left to prove that $\Square(k-1)$ does not contain any replica from the ``top part'' of $\reck$.

\vspace{0.2cm}

Assume by the way of contradiction that there exists a replica $q=(u,s)\in\reck\cap\Square(k-1)$ such that $s > t_k$.
Let $r_l$ be the request in which $\Square$ added $q$ to the solution (that is, when $\Square$ was handling $r_l$, it added $q$ to the solution).
The assumption that $q\in\Square(k-1)$ implies that such a request $r_l$ does exist, and in particular, $l\leq k-1$.
Thus, $t_k\geq t_l$, and hence,
$s > t_l~.$
This implies that $q$ is added to the solution in step SQ5 and
$q\in\tail(l) = \calP_\calA[\rep{\uSQ_{l}}{t_{l}},\rep{\uSQ_{l}}{t_{l}+4\cdot\rhoSQ(l)}]$.
Therefore, also, $(\uSQ_{l},t_k)\in\tail(l)$ (since $t_{l}\leq t_k$ and $t_k\leq s\leq t_{l}+4\cdot\rhoSQ(l)$), and in particular,
$$(\uSQ_{l},t_k)\in \Square(k-1).$$

\vspace{0.5cm}

In addition, $\uSQ_{l}\in [v_i,v_j]$, since $q\in\reck$, and also
$$v_k-\rhoSQ(k)\leq v_i\leq \uSQ_{l}\leq v_j <\uSQ_k \leq v_k~,$$
where the first and the last inequalities hold by Lemma \ref{lemma:square:uncoverd req not edegs disjoint} Property (P2) with $i$ and $k$;
the second and the third inequalities hold since $\uSQ_{l}\in [v_i,v_j]$; and
the fourth inequality holds by the assumption of case (2).

\vspace{0.5cm}
Therefore, in particular, $0\leq v_k-\uSQ_l\leq \rhoSQ(k)$.
Thus, by Observation \ref{obser:sqr: if soff=ti then ri is covered}, $r_k$ is a covered request,
contradicting the assumption that $k$ is child of $i$ (covered requests have no parents).
Therefore, $\Square(k-1)\cap\reck=\emptyset$ and (as mentioned) also $\Square(k)\cap\reck=\emptyset$.
Hence, Ineq. (\ref{ineq: Square(k) cap reck = emptyset}) holds.

\vspace{0.5cm}
\item[]{\bf Proving Ineq. (\ref{ineq: Square(j-1) cap reck not= emptyset}):}
%
Recall that the $j$'th closest replica $\qclose_j=(\uclose_j,\sclose_j)\in \Square(j-1)$, see step SQ2.
Thus, to show that Ineq. (\ref{ineq: Square(j-1) cap reck not= emptyset}) holds, it is sufficient to show that $\qclose_j\in \reck$.

The assumption that $\SQball(j)$ and $\SQball(i)$ are not edge disjoint implies that the $j$'th serving and closest replicas are on the right of $r_i$.
That is,
\begin{eqnarray*}
v_j-\rhoSQ(j)<v_i<\uSQ_j\leq \uclose_j\leq v_j<\uSQ_k,
\end{eqnarray*}
where the first and the second inequalities hold by Lemma \ref{lemma:square:uncoverd req not edegs disjoint}, Property (P2);
the third and the forth inequalities hold by steps SQ2 and SQ3;
and the fifth inequality is the assumption in the current case (2).

This implies, in particular, that
\begin{eqnarray}
v_j-\rhoSQ(j)<v_i<\uclose_j\leq v_j~.
\label{ineq: uclose in the range of (vi,vj]}
\end{eqnarray}
In addition,
by Observation \ref{obser:sqr: if ri is uncovered then rad(i)= ti-s'i},
the radius of an uncovered request is the {\em time} difference between the request and its closest replica, that is,
$\rhoSQ(j)=t_j-\sclose_j$, and equivalently
\begin{eqnarray}
\label{ineq: rhosq j = tj -s'j}
 \sclose_j=t_j-\rhoSQ(j).
\end{eqnarray}
Recall that $k$ and $j$ are children of $i$, thus $\SQball(j)$ and $\SQball(k)$ are edges disjoint.
This, together with inequalities (\ref{ineq: uclose in the range of (vi,vj]}) and (\ref{ineq: rhosq j = tj -s'j})
imply that
\begin{eqnarray}
\sclose_j \geq t_k~.
\label{ineq: sclosej geq tk}
\end{eqnarray}
Hence, $\qclose_j\in\reck$ by inequalities (\ref{ineq: uclose in the range of (vi,vj]}) and (\ref{ineq: sclosej geq tk}) and since $\sclose_j\leq t_j\leq t_i$.
Thus, Ineq. (\ref{ineq: Square(j-1) cap reck not= emptyset}) holds as promised.

\end{enumerate}

We have shown that inequalities (\ref{ineq: Square(k) cap reck = emptyset}) and (\ref{ineq: Square(j-1) cap reck not= emptyset}) hold
as we argued above
this implies that $\rr_{\ell^*}$ has Property (P1).

\item {\bf $\ell^*$ has Property (P2), i.e., $\tail(\ell^*)$ is in the range of $\SQball(j)$.}
Recall that $i>j$; and $\SQball(i)$ and $\SQball(j)$ are not edge disjoint, thus by Lemma \ref{lemma:square:uncoverd req not edegs disjoint}, Part 2,
\begin{eqnarray}
\label{ineq:vj-rho j < vi < vj}
v_j-\rhoSQ(j) < v_i~ < v_j~.
\end{eqnarray}
We show
%
%
that
\begin{eqnarray}
\label{ineq:vi < uSQell vi < vj}
v_i~ < \uSQ_{\ell^*} \leq v_j~,
\end{eqnarray}
which implies together with Ineq. (\ref{ineq:vj-rho j < vi < vj}) that
$
v_j-\rhoSQ(j) < \uSQ_{\ell^*}~ < v_j
$ 
as needed (for showing that $\tail(\ell^*)$ is in the range of $\SQball(j)$).

It remains to show that Ineq. (\ref{ineq:vi < uSQell vi < vj}) holds.
%
%
Note that, on the one hand,
the choice of $\rr_{\ell^*}$ (as the first request which the solution $\Square(\ell^*)$ contains a replica in $\reck$) implies that
some replica $q'=(u',t')\in\reck$ is added to the solution when $\Square$ handles $r_{\ell^*}$.
On the other hands, when Algorithm $\Square$ handles $r_{\ell^*}$,
it only adds replicas (in steps SQ4 and SQ5) to the right of $\uSQ_{\ell^*}$ and to the left of $v_{\ell^*}$.
Thus, on the one hand, $v_i\leq u' \leq v_j$, and on the other hand, $\uSQ_{\ell^*}~ \leq u' \leq v_{\ell^*}$.
Hence, also
\begin{eqnarray}
\label{ineq: two inequalities}
v_i~ \leq v_{\ell^*} \mbox{~~~~and~~~~} \uSQ_{\ell^*} \leq v_j~.
\end{eqnarray}

This already establish the right inequality of (\ref{ineq:vi < uSQell vi < vj}).
To show that its left inequality holds too,
assume toward contradiction that $\uSQ_{\ell^*} \leq v_i$.
Combining this with the left inequality of (\ref{ineq: two inequalities}), we have
$$
\uSQ_{\ell^*} \leq v_i \leq v_{\ell^*}~.
$$
This implies that, when Algorithm $\Square$ handles $\rr_{\ell^*}$~, it added the replica $(v_i,t_{\ell^*})$, in step SQ4 to the solution.
Hence, $(v_i,t_{\ell^*})\in\Square(\ell^*)$,
and is a candidates for the $i$'th close replica (see step SQ2).
Thus,
$$
\rhoSQ(i) \leq \distinf{(v_i,t_{\ell^*},\rr_i)} = t_i - t_{\ell^*}~.
$$
Hence, the time of each of $\SQball(i)$'s replicas is at least $t_{\ell^*}$.
Recall that, $t_{\ell^*} \geq t_k$ (since $\ell^*>k$);
and that in each edge $e$ of $\SQball(i)$ at least one of $e$'s endpoints is corresponds to time later  than $v_i-\rhoSQ(i)$.
Therefore, $\SQball(i)$ and $\SQball(k)$ are edge disjoint,
which contradicts the choice of $k$ as a child of $i$.
Hence, $v_i<\uSQ_{\ell^*}$,  Ineq. (\ref{ineq:vi < uSQell vi < vj}) holds and $\ell^*$ maintains Property (P2) as promised.

\item {\bf $\ell^*$ has Property (P3), i.e., $\rhoSQ(\ell^*)\geq \rhoSQ(k)$.}\\

We first show that the time $\sSQ_{\ell^*}$ of the serving replica $\qSQ_{\ell^*}$ of the $\ell^*$'th request is before $t_k-5\rhoSQ(k)$.
That is,
\begin{eqnarray}
\label{ineq; sSQ ell* < tk -5rhotk}
\sSQ_{\ell^*} \leq t_k-5\rhoSQ(k).
\end{eqnarray}
The choice of $\ell^*$ implies that $\Square(\ell^*-1)\cap\reck=\emptyset$.
On the other hand, the serving replica $\qSQ_{\ell^*}$ does belong to $\Square(\ell^*-1)$ (see step SQ3).
This implies, in particular, that
$$
\qSQ_{\ell^*}=(\uSQ_{\ell^*},\sSQ_{\ell^*})\not\in \reck~=~[v_i,v_j]\times[t_k-5\rhoSQ(k), t_i]~.
$$
Recall that $\uSQ_{\ell^*}\in [v_i,v_j]$ by Ineq. (\ref{ineq:vi < uSQell vi < vj}),
hence $\sSQ_{\ell^*}\not\in [t_k-5\rhoSQ(k), t_i]$.
Inequality (\ref{ineq; sSQ ell* < tk -5rhotk}) holds, since $\sSQ_{\ell^*}\leq t_{\ell^*}\leq t_i$.

Summarizing what we know so far,  $t_{\ell^*}\geq t_k$ and $\sSQ_{\ell^*} \leq t_k-5\rhoSQ(k)$.
Thus, on  one hand,
\begin{eqnarray}
\label{ineq: distinf qSQell* geq 5 rho(k)}
\distinf{\qSQ_{\ell^*},\rr_{\ell^*}} \geq t_{\ell^*}-\sSQ_{\ell^*}\geq 5\rhoSQ(k).
\end{eqnarray}
On the other hand,
$ 
\qSQ_{\ell^*} \in \seq{r_{\ell^*},5\cdot\rhoSQ(\ell^*)}
$ 
(see step SQ3),
which implies that
\begin{eqnarray}
\label{ineq: distinf qSQell* leq 5 rho(l)}
5\rhoSQ(\ell^*) \geq \distinf{\qSQ_{\ell^*},\rr_{\ell^*}}~.
\end{eqnarray}

Inequalities (\ref{ineq: distinf qSQell* geq 5 rho(k)}-\ref{ineq: distinf qSQell* leq 5 rho(l)}), imply that $\rhoSQ(\ell^*)\geq \rhoSQ(k)$ as needed.
Hence, $\ell^*$ maintains Property (P3).

\end{enumerate}

\noindent We have shown that $\rr_{\ell^*}$ maintains the three properties, implying the lemma for case (2) too.
%
[Lemma \ref{lemma:sqr: tj-tk leq 4 radius k}]
\QED


The previous lemma shows that a lot of time pass between the time of the last replica in the quarter ball of a child and the time of first replica in the quarter ball of the next child.
The next, lemma use this property to show that the radius of a root is at least half of the sum of the radii of its children in its tree.

\begin{lem}
Consider some root request $\rr_{i^*}\in\roots$.
Then,
$$
2\rhoSQ(i^*)\geq \sum_{i\in\Tree(i^*)}\rhoSQ(i).
$$
\label{lema:sqr:root radi geq sum of its children}
\end{lem}
\proof
We begin by showing that the radius of each ball $\SQball(i)$ in the tree is at least, twice the sum of the radii of its children.
Consider some non leaf request $\rr_i\in\Tree(i^*)$ (that is, $\suns(i)\not=\emptyset$).
Let us first, show that
\begin{eqnarray}
\rhoSQ(i)\geq 2\sum_{j\in\suns(i)}\rhoSQ(j).
\label{ineq:sqr: rhoSQ(i) geq 2sum rhoSQ(j)}
\end{eqnarray}
If $\suns(i)=\{j\}$ ($i$ has exactly one child), then (\ref{ineq:sqr: rhoSQ(i) geq 2sum rhoSQ(j)}) follows from
property (P1) of Lemma \ref{lemma:square:uncoverd req not edegs disjoint}.
Otherwise, $\suns(i)=\{j_1,j_2...,j_\nu\}$, where $\nu\geq 2$ and $j_1\leq j_2\leq...\leq j_\nu$.
(For simplicity, to avoid double subscripts, we may write $t(l)$ instead of $t_l$.)
By Lemma \ref{lemma:sqr: tj-tk leq 4 radius k} with $k=j_l$ and $j=j_{l+1}$, it follows that
\begin{eqnarray}
t(j_{l})+4\rhoSQ(j_l)+\rhoSQ(j_{l+1})\leq t(j_{l+1}),
\label{ineq:sqr: t(jl) + 4 rad(jl) + rad}
\end{eqnarray}
for every $l=1,...,\nu-1$.
Now, (see Figure \ref{fig:sqr:ParentAndChildren})
\begin{eqnarray}
\rhoSQ(i)\geq t_i-t(j_1) \geq 4\sum_{l=1}^{\nu-1}\rhoSQ(j_l),
\label{ineq:sqr:rad(i) geq ti-tj1 geq sum}
\end{eqnarray}
where the first inequality holds since the $\SQball(i)$ and $\SQball(j_1)$ are {\em not} edges disjoint;
the second inequality holds by Inequality (\ref{ineq:sqr: t(jl) + 4 rad(jl) + rad}), since $t_i\geq t(j_\nu)$.
\begin{figure*}
\begin{center}
\includegraphics[scale=0.35]{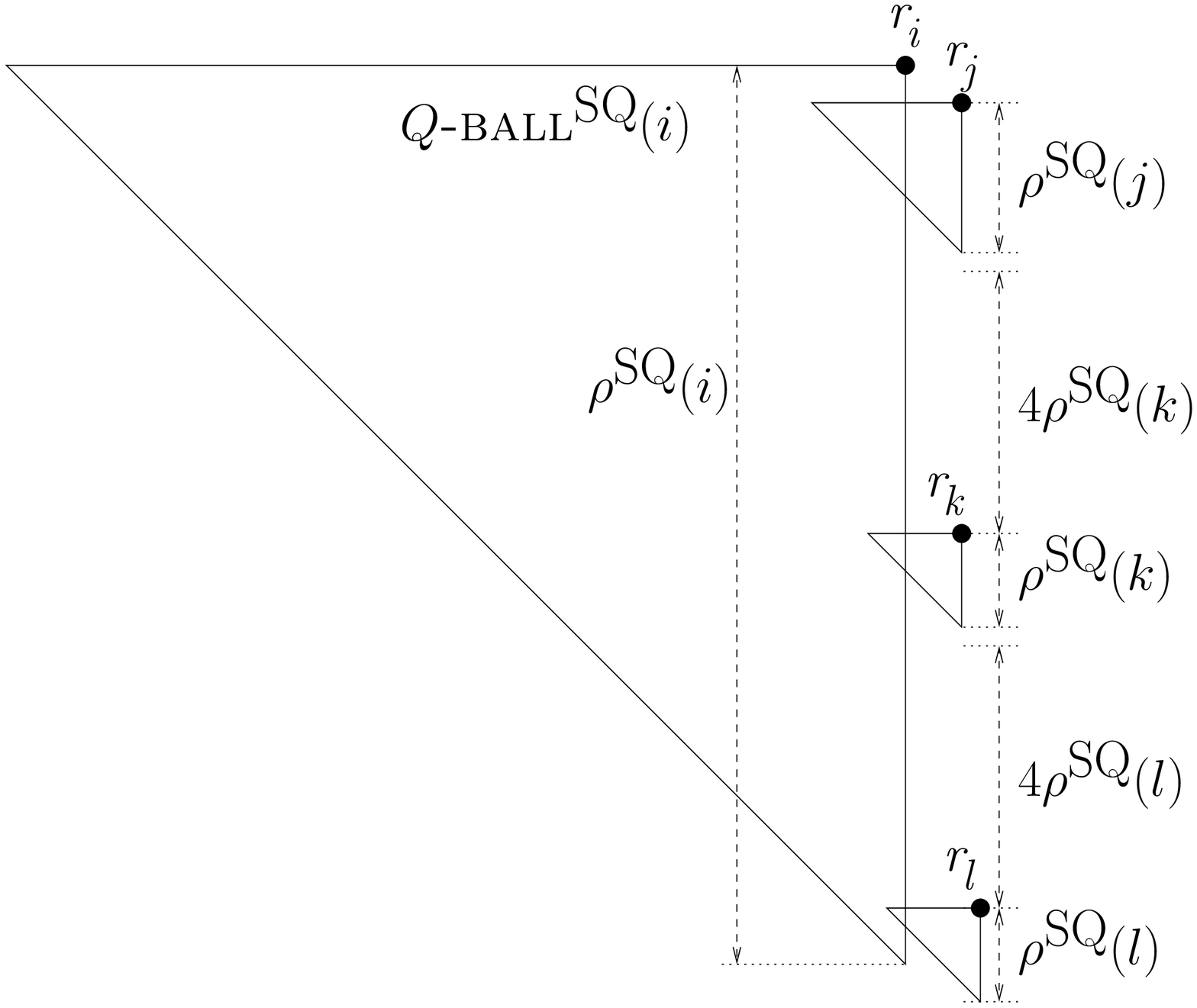}
\end{center}
\caption{\sf Geometric vision on a parent and its children relationships.
\label{fig:sqr:ParentAndChildren}
}
\end{figure*}
In addition, by Property (P1) of Lemma
\ref{lemma:square:uncoverd req not edegs disjoint},
\begin{eqnarray}
\rhoSQ(i)\geq 4\rhoSQ(j_\nu),
\end{eqnarray}
which implies Inequality (\ref{ineq:sqr: rhoSQ(i) geq 2sum rhoSQ(j)})
that implies the lemma.
\QED

So far, we have shown that
(1) the quarter-ball of the covered requests are edges disjoint;
(2) the quarter-ball of the root requests are edges disjoint, and hence by Observation \ref{obser:sqr: cost Square > 14 sum radii} and Observation \ref{obser:sqr:opt geq sum rho_i}, the sum of their radii of the covered request and the root requests is no more than 28 times the cost of $\opt$.
On the other hand,
the sum of root's radii is at least half of the sum of the radii of the uncovered requests.
This, in fact, establishes Theorem \ref{thm: square is O(1)-approx}.

\noindent {\bf Proof of Theorem \ref{thm: square is O(1)-approx}:}
\AppSquareThm


\section{Algorithm $\Dlineon$ - the ``real'' online algorithm}
\label{subsec: Algorithm Donline}

In this section, we transform the pseudo online algorithm $\Square$ of Section \ref{sec:square} into a (fully) online algorithm $\Dlineon$  for $\DMCD$%
\footnote{
We comment that it bears similarities to the transformation of the pseudo online algorithm Triangle to a (full) online algorithm for {\em undirected} $\MCD$ in \cite{KK2014}. The transformation here is harder, since there the algorithm sometimes delivered a copy to a node $v$ from some node on $v$'s right, which we had to avoid here (since the network is directed to the right).
}.
Let us first give some intuition here.

The reason Algorithm $\Square$ is {\em not} online, is one of the the actions it takes at step SQ4.
There, it stores a copy at the serving replica  $\uSQ_i$ for request $r_i$ from time $\sSQ_i$ to time $t_i$.
This requires ``going back in time'' in the case that the time ${\sSQ_i}< t_i$.
A (full) online algorithm cannot perform such an action.
Intuitively, Algorithm $\Dlineon$ ``simulates'' the impossible action by
(1) storing additional copies (beyond those stored by $\Square$); and
(2) shifting the delivery to request $r_i$ (step SQ4 of $\Square$) from an early time to time $t_i$ of $r_i$.
It may happen that the serving node $\uSQ_i$ of $r_i$ does not have a copy (in $\Square$) at $t_i$.
In that case, Algorithm $\Dlineon$ also (3)
delivers first a copy to $\rep{\uSQ_i}{t_i}$ from some node $w$ on the left of $\uSQ_i$.
Simulation step (1) above (that we term the storage phase) is the one responsible for ensuring that such a node $w$ exists, and is ``not too far'' from $\uSQ_i$.

For the storage phase, Algorithm $\Dlineon$ covers the network by ``intervals'' of various lengthes (pathes that are subgraphs of the network graph).
There are overlaps in this cover, so that each node is covered by intervals of various lengthes.
Let the length of some interval $I$ be $length(I)$.
Intuitively, given an interval $I$ and a time $t$,
if $\Square$ kept a copy in a node of interval $I$ ``recently'' (``recent'' is proportional to $length(I)$),
then $\Dlineon$ makes sure that a copy is kept at the left most node of this interval, or ``nearby''
(in some node in the interval just left to $I$).

Now, we (formally) illustrated Algorithm $\Dlineon$. We begins by giving some definitions.



\paragraph*{\bf Partitions of $[1,n]$ into intervals%
\commsingle.\commsingleend}
Consider some positive integer $\delta$ to be chosen later.
For convenience, we assume that $n$ is a power of $\delta$.
(It is trivial to generalize it to other values of $n$.)
Define $\log_\delta n+1$ {\em levels} of partitions of the interval $[1,n]$.
In level $l$, partition $[1,n]$ into $n/\delta^l$ intervals,
$\Jnter{\delta}{1}{l}$, $\Jnter{\delta}{2}{l}$,...,$\Jnter{\delta}{n/\delta^l}{l}$, each of size $\delta^l$.
%
That is,
$\Jnter{\delta}{j}{l}=\{(j-1)\cdot \delta^l+k\mid k=1,...,\delta^l\}$,
for every $1\leq j \leq n/\delta^{l}$ and every $0\leq l \leq \log_\delta n$.
Let $\calI\langle\delta\rangle$ be the set of all such intervals.
When it is clear from the context, we may omit $\langle\delta\rangle$ from $\calI\langle\delta\rangle$ and $\Jnter{\delta}{j}{l}$ and write $\calI$ and $\Inter{l}{j}$, respectively.
Let $\lfun{I}$ be the {\em level} of an interval $I\in\calI$, i.e., $\lfun{\Inter{l}{j}}=l$.
For a given interval $\Inter{l}{j}\in\calI$,
denote by $\DNLI{\Inter{l}{j}}$, for $1< j\leq n/\delta^l$ the {\em neighbor} interval of level $l$ that is on the left of $\Inter{l}{j}$.
That is, $\DNLI{\Inter{l}{j}}=\Inter{l}{j-1}$. 
Define that $\DNLI{\Inter{i}{1}}=\{0\}$. 
Let
$\NDI{\Inter{}{}}=
\DNLI{\Inter{}{}}\cup \Inter{}{}.
$
We say that $\NDI{I}$ is the {\em neighborhood } of $I$.

Denote by $\Interk{v}{l}$ (for every node $v\in V$ and every level $l=0,..., \log_\delta n$) the  interval in level $l$ that contains $v$.
That is,
$\Interk{v}{l}=\Inter{l}{k}, \mbox{ where } k=\left\lfloor\frac{v}{\delta^l}\right\rfloor+1$.
%
The {\em neighborhood} $\NID{l}{v}$ of a node $v$ contains all those nodes in the neighborhood $\NDI{\Interk{v}{l}}$
(of the interval of level $l$ of $v$) that are left of $v$.
That is, $\NID{l}{v}=\{u\in\NDI{I^l(v)}\mid u\leq v\}$.

\vspace{-0.3cm}
\paragraph*{\bf Active node%
\commsingle.\commsingleend}
\label{subsection: def delta active intervals}
Consider some node $v\in V$, some level $0\leq l\leq \log_\delta n$.
Node $v$ is called {\em $\langle l,\delta\rangle$-active} at time $t$,
if
$(\DBase\cup\tail)\cap v[t-\delta^l,t]\not=\emptyset$.
Intuitively,  Algorithm $\Square$
kept a movie copy in $v$, at least once, and ``not to long'' before time $t$.
We say that
$v$ is {\em $\langle l,\delta\rangle$-$\stayactive$}, intuitively, if $v$ is {\bf not} ``just about to stop being $\langle l,\delta\rangle$-active'',
that is, if 
$(\DBase\cup\tail)\cap v[t-\delta^l+1,t]\not=\emptyset$.


Let us now construct $\calC_{t+1}$, the set of replicas corresponding to the nodes that store copies from time $t$ to time $t+1$ in a $\Dlineon$ execution.
Let $\calC_0=\{(v_0=0,0)\}$.
(The algorithm will also leave a copy in $v_0=0$ always.)
To help us later in the analysis, we also added an auxiliary set
$\DCOMMIT\subseteq\{\langle I,t\rangle\mid I\in \calI\langle\delta\rangle \mbox{ and } t\in\mathbb{N}\}$.
Initially, $\DCOMMIT\leftarrow\emptyset$.
For each time $t=0,1,2,...$, consider first the case that there exists at least one request corresponding to time $t$, i.e., $\calR[t]=\{\rr_{j},...,\rr_k\}\not=\emptyset$.
Then, for each request $\rr_i\in\calR[t]$,
$\Dlineon$ simulates $\Square$ to find the radius $\rhoSQ(i)$ and the serving node $\uSQ_i$ of the serving replica $\qSQ_i=(\uSQ_i,\sSQ_i)$ of $\rr_i$.
Unfortunately, we may not be able to deliver, at time $t$, a copy from $\qSQ_i$ may be $t>\sSQ_i$.
Hence, $\Dlineon$ delivers a copy to $\rr_i$ via $(\uSQ_i,t)$ (this is called the {\em ``delivery phase''}).
That is, for each $i=j,...,k$ do:
\begin{itemize}
\item[(D1)] choose a closest (to $(\uSQ_i,t)$) replica $\Dqon_i=(\Duon_i,t)$
on the left of $\uSQ_i$  of time $t=t_i$ already in the solution;

\item[(D2)] add the path $\DPHon(i)=\calP_\calH[\Dqon_i,\rr_i]$
to the solution.

\end{itemize}
Let $\DPVon(i)=\{\rr \mid (r,q)\in\DPHon(i)\}$.
(Note that $\rr_j$ is served from $\calC_t$, after that, the path $\DPHon(j)$ is added; and $\rr_{j+1}$ is served from $\calC_t\cup\DPVon(j)$, etc.)

Recall that before the delivery phase, the replicas of $\calC_{t}$ have copies.
It is clear, that the delivery phase of time $t$ ensures that the replicas of $\DBase[t]\cup\tail[t]$ have copies too.
That is, at the end of the delivery phase of time $t$, at least the replicas of $\calC_{t}\cup\DBase[t]\cup\tail[t]$ have copies.
It is left to decide which of the above copies to leave for time $t+1$.
That is (the {\em ``storage phase''}),
$\Dlineon$ chooses the set
$\calC_{t+1}\subseteq \calC_{t}\cup\DBase[t]\cup\tail[t]$.
%
Initially, $C_{t+1}\leftarrow\{\rep{\rot}{t+1}\}\cup\{(u,t+1)\mid (u,t)\in\tail\}$ (as we choose to leave copy at the replicas of the tails and to leave a copy at $\rot$ always).
Then, for each level $l=0,...,\log_\delta n$, in an {\em increasing} order, the algorithm goes over and
each node $v=1,...,\nn$, in an {\em increasing} order, selects as follows.
\begin{itemize}
\item[(S1)] Choose a node $v$ such that
(1) $v$ is level $\langle l,\delta\rangle$-$\stayactive$ at $t$; but
(2) no replica has been selected in level $l$ $v$'s neighborhood ($\calC_{t+1}\cap\NID{l}{v}[t+1]=\emptyset$).
If such a node $v$ does exist, then perform steps (S1.1--S1.3) below.

\item[(S1.1)] Add the tuple $\langle I^l(v),t\rangle$ to the auxiliary set $\DCOMMIT$;
we say that the interval $I^l(v)$ {\em commits} at level $l$ at time $t$.

\item[(S1.2)]
Select a node $u\in\NID{l}{v}$ such that a replica of $u$ at time $t$ is in
$\DBase[t]\cup \calC_{t}$ (by Observation \ref{obser:Dlinoon: well defined} below, such a replica does exist,
recall that all these replicas have copies at this time).
\item[(S1.3)] Add  $\rep{u}{t+1}$ to $\calC_{t+1}$ and add the arc $(\rep{u}{t},\rep{u}{t+1})$ to the solution.

\end{itemize}

The solution constructed by $\Dlineon$ is denoted $\FDon=\HDon\cup\ADon$, where
$\HDon=\cup_{i=1}^{\NN}\DPHon(i)$
represents the horizontal edges added in the delivery phases and
$\ADon=\{(\rep{v}{t},\rep{v}{t+1})\mid \rep{v}{t+1}\in\calC_{t+1} \mbox{ and } t=0,...,t_\NN\}$
represents the arcs added in the storage phase.

\begin{observation}
{\sc (``Well defined'').}
If a node $v\in V$ is level $\langle l,\delta\rangle$-$\stayactive$ at time $t$,
then there exists a replica $\rep{u}{t}\in \calC_{t}\cup\DBase[t]\cup\tail[t]$ such that $(v,t)\in\NID{l}{v}$. 
\label{obser:Dlinoon: well defined}
\end{observation}
\def\AppDlineonObserWellDefined{
\proof
Consider some node $v\in V$ and a time $t$. If $\langle l,\delta\rangle$-$\stayactive$ at time $t$,
then either
$(v,t)\in\DBase\cup\tail$ 
or $(v,t)\not\in\DBase[t]\cup\tail[t]$ and
$v$ is also $\langle l,\delta\rangle$-$\stayactive$ at time $t-1$ (and
$\calC_{t}\cap
\NID{l}{v}[t]\not=\emptyset$);
hence, $(\DBase[t]\cup\tail[t]\cup \calC_{t})\cap \NID{l}{v}[t]\not=\emptyset$.
The observation follows.
%
\QED
} 
\AppDlineonObserWellDefined

Moreover, a $\stayactive$ node $v$ has a copy in its neighborhood longer (for an additional round).

\begin{observation}
{\sc (``A $\langle l,\delta\rangle$-active node has a near by copy'').}
If a node $v$ is $\langle l,\delta\rangle$-$\Active$ at time $t$, then, either
(1) $(\DBase\cup\tail)\cap\NID{l}{v}[t]\not=\emptyset$,
or
(2) $\NID{l}{v}[t]\cap \calC_{t}\not=\emptyset$.
\label{obser:Dlineon: a l active has a near by copy}
\end{observation}
\def\AppObserDlineonActiveNearByCopy{
\proof
Consider a node $v\in V$ that is $\langle l,\delta\rangle$-$\Active$  at time $t$.
If $(\DBase\cup\tail)\cap\NID{l}{v}[t]\not=\emptyset$, then the observation follows.
Assume that $(\DBase\cup\tail)\cap\NID{l}{v}[t]=\emptyset$.
Then, the fact that $v$ is $\langle l,\delta\rangle$-$\Active$ at $t$,
but $(v,t)\not\in\DBase\cup\tail$, implies also, that $v$ is $\langle l,\delta\rangle$-$\stayactive$ at time $t-1$.
Thus,  either (1) $I^l(v)$ commit at $t-1$ (at step (S1.1)) which ``cause'' adding an additional replica to $\calC_t$ from $\NID{l}{v}$ (at step (S1.2));
or (2) $I^l(v)$ does not commit at $t-1$, since $\calC_t$ has, already, a replica from $\NID{l}{v}$.
%
\QED
} 
\AppObserDlineonActiveNearByCopy

\begin{observation}
{\sc (``Bound from above on $|\ADon|$'').}
$|\ADon\setminus \calP_\calA[\rep{\rot}{0},\rep{\rot}{t_{\NN}}]|\leq |\DCOMMIT|$.
\label{obser:Dlineon: |Commit|=|Aon-0|}
\end{observation}
\def\AppObserDlineonBoundFromAboveonCommit{
\proof
Let
$\ADon_{-\rot}=\ADon\setminus\calP_\calA[\rep{\rot}{0},\rep{\rot}{t_{\NN}}]$.
%
Now we prove that $|\ADon_{-\rot}|=|\DCOMMIT|$.
Every arc in $\ADon_{-\rot}$ (that add at step (S1.3)) corresponds to exactly one tuple $\langle I,t\rangle$ of an interval $I$ that commits at time $t$ (in step (S1.1));
and every interval commits at most once in each time $t$ that corresponds to exactly one additional arc in $\calA_{-\rot}$.
Thus, $|\ADon_{-\rot}|=|\DCOMMIT|$.
The observation follows.
\QED
} 
\AppObserDlineonBoundFromAboveonCommit

\vspace{-0.5cm}
\paragraph*{\bf Analysis of $\bfDlineon$%
\commsingle.\commsingleend}
We, actually, compare the cost of Algorithm $\Dlineon$ to that of the pseudo online Algorithm $\Square$.
The desired competitive ratio for $\Dlineon$ will follow, since we have shown that $\Square$ approximates the optimum (Theorem \ref{thm: square is O(1)-approx}).
A similar usage of a (very different) pseudo online algorithm utilized in \cite{KK2014}.
\commlong
\begin{eqnarray*}
\frac{\cost(\Dlineon,\calR)}{\cost(\Square,\calR)}=O(\frac{\log n}{\log \log n}).
\end{eqnarray*}
\commlongend
\commshort
$
\frac{\cost(\Dlineon,\calR)}{\cost(\Square,\calR)}=O(\frac{\log n}{\log \log n}).
$
\commshortend
%
This implies the desired competitive ratio of $O(\frac{\log n}{\log \log n})$ by
Theorem \ref{thm: square is O(1)-approx}.
%
%
%
We first show,
that the number of horizontal edges in $\HDon$ ({\em``delivery cost''}) is $O\left(\delta\cdot\cost(\Square,\calR)\right)$.
Then, we show,
that the the number of arcs in $\ADon$ ({\em``storage cost''}) is $O\left(\log_\delta n\cdot\cost(\Square,\calR)\right)$.
Optimizing $\delta$, we get a competitiveness of $O(\frac{\log n}{\log \log n})$.

\commsingle
\commsingleend
\vspace{-0.2cm}
\paragraph{\bf Delivery cost analysis.}

For each request $\rr_i\in\calR$, the delivery phase (step (D2)) adds $\DPHon(i)=\calP_\calH[\Dqon_i,\rr_i]$
to the solution.
Define the {\em online} radius of $\rr_i$ as $\Dron{i}=d(\Dqon_i,\rr_i)$.
We have,
%
\vspace{-0.3cm}
\begin{eqnarray}
|\HDon|\leq \sum_{i=1}^{\NN} \Dron{i}.
\label{Ineq: |Hon| leq sum Dron(i)}
\end{eqnarray}
It remains to bound $\Dron{i}$ as a function of $\rhoSQ(i)$ from above.
%
%
%
Restating Observation \ref{obser:Dlineon: a l active has a near by copy} somewhat differently
we can use the distance
$v_i-\uSQ_i \leq 5\rhoSQ(i)$ (see (SQ3)) and the time difference $t_i-\sSQ_i\leq 5\rhoSQ(i)$ for bounding $\Dron{i}$.
That is, we show that $\Dlineon$ has a copy at time $t_i$ (of $\rr_i$) at a distance at most $10\delta\rhoSQ(i)$ from $\uSQ_i$ (of $\qSQ_i$ of $\Square$).
Since, $v_i-\uSQ_i\leq 5\rhoSQ(i)$, $\Dlineon$ has a copy at distance at most $(10\delta+5)\rhoSQ(i)$ from $v_i$ (of $\rr_i$).

\vspace{0.2cm}
\begin{lem}
$\Dron{i}\leq (10\delta+5)\cdot\rhoSQ(i)$.
\label{lemma:Dlineon: delivery cost}
\end{lem}
\def\AppLemmaDlineonDevCost{
\proof
The following claim
restating Observation \ref{obser:Dlineon: a l active has a near by copy} somewhat differently and
help us to prove that the serving replica has a ``near by'' copy.

\vspace{0.2cm}
\begin{claim}
Consider some base replica $\rep{v}{t}\in\DBase\cup\tail$ and some $\rho>0$, such that, $t+\rho\leq t_{\NN}$.
Then, there exists a replica $(w,t+\rho)\in\calC_{t+\rho}$ such that $v-w\leq 2\delta\rho$. 
\label{claim: dist(v, C t+rho) leq 2 delta rho}
\end{claim}
\begin{proof}
Assume that $\rep{v}{t}\in\DBase\tail$.
Consider an integer $\rho>0$.
Let $l=\lceil \log_\delta \rho\rceil$.
Node $v$ is $\langle l,\delta\rangle$-$\Active$ at time $t+\rho$.
Thus, by Observation \ref{obser:Dlineon: a l active has a near by copy},
there exists some node $w\in\NID{l}{v}$
that keep a copy for time $t+\rho$.
That is, a replica $(w,t+\rho)\in\NID{l}{v}[t+\rho]\cap \calC_{t+\rho}$ does exists.
The fact that $w\in\NID{l}{v}$ implies that $v-w\leq 2\delta^{l}$.
The claim follows, since $\rho>\delta^{l-1}$.
\end{proof}
\QED

Recall that $\Square$ serves request $\rr_i=(v_i,t_i)$ from some base replica $\qSQ_i=(\uSQ_i,\sSQ_i)$ already include in the solution.
That $\qSQ_i$ may correspond to some earlier time. That is, $\sSQ_i\leq t_i$.
In the case that $\sSQ_i=t_i$, $\Dlineon$ can serve $\rr_i$ from $\qSQ_i$.
Hence, $\Dron{i}\leq 5\rhoSQ(i)$.
In the more interesting case%
, $\sSQ_i<t_i$.
%
By Claim \ref{claim: dist(v, C t+rho) leq 2 delta rho}
(substituting $v=\uSQ_i$, $t=\sSQ_i$, and $\rho= t_i-\sSQ_i\leq 5\rhoSQ(i)$),
there exists a replica $(w,t_i)\in\calC_{t_i}$
such that $\uSQ_i-w\leq 10\delta\rhoSQ(i)$.
Recall that $v_i-\uSQ\leq 5\rhoSQ(i)$ (see (SQ3)).
Thus, 
$v_i-w\leq (10\delta+5)\rhoSQ(i)$.
Hence, $\Dron{i}\leq(10\delta+5)\rhoSQ(i)$ as well.
\QED
} 
\AppLemmaDlineonDevCost

\noindent The following corollary holds, by combining together the above lemma  with Inequality (\ref{Ineq: |Hon| leq sum Dron(i)}).

\begin{corollary}
$|\HDon|\leq (10\delta+5)\cdot \cost(\Square,\calR)$.
\label{corollary:Dlineon: delivery cost}
\end{corollary}

\vspace{-0.3cm}
\paragraph*{\bf Analysis of the storage cost%
\commsingle.\commsingleend}

By Observation \ref{obser:Dlineon: |Commit|=|Aon-0|},
it remains to bound the size of $|\DCOMMIT|$ from above.
Let $\Dcommit(I,t)=1$ if $\langle I,t\rangle\in\DCOMMIT$ (otherwise 0).
Hence,
$|\DCOMMIT|=
\sum_{I\in \calI}\sum_{t=0}^{t_\NN}\Dcommit(I,t)
$.
We begin by bounding the number of commitments in $\Dlineon$ made by nodes for level $l=0$.
Observation \ref{obser:Dlineon Base acounts level l=0 committments} below follows directly from the definitions of $\commit$ and $\stayactive$.

\vspace{0.2cm}
\begin{observation}
$\sum_{I\in \calI:\lfun{I}=0}\sum_{t=0}^{t_\NN} \Dcommit(I,t)\leq\big|\FSQ\big|.$
\label{obser:Dlineon Base acounts level l=0 committments}
\end{observation}
\proof
Consider some commitment $\langle I,t\rangle\in\DCOMMIT$, where interval $I$ is of level $\lfun{I}=0$.
Interval $I$ commit at time $t$ only if there exists a node $v\in I$ such that $v$ is $\langle l=0,\delta\rangle$-$\stayactive$ at $t$ (see step (S1) in $\Dlineon$).
This $\stayactive$ status at time $t$ occur only if $(v,t)\in\DBase\cup\tail$.
Hence, each base replica causes at most one commitment at $t$ of one interval of level $l=0$.
%
\QED

The following lemma is not really new.
The main innovation of the paper is the special pseudo online algorithm we developed here.
The technique for simulating the pseudo online algorithm by a ``true'' online one, as well as the following
analysis of the simulation, are not really new.
For completeness we still present a (rather detailed) proof sketch for Lemma \ref{lem:Dlineon: storage cost < Hoff + Aoff}.
Its more formal analysis is deferred to the full paper
(and a formal proof of a very similar lemma for very similar mapping of undirected $\MCD$)
can be found in Lemma 3.8 of \cite{KK2014TR}.

\vspace{0.2cm}
\begin{lem}
$|\DCOMMIT|
\leq (1+4\log_\delta n)\big|\FSQ\big|
$.
\label{lem:Dlineon: storage cost < Hoff + Aoff}
\end{lem}
%
\noindent{\bf Proof sketch:}
The $1$ term in the statement of the lemma follows from Observation \ref{obser:Dlineon Base acounts level l=0 committments} for commitments of nodes for level $l=0$.
The rest of the proof deals with commitments of nodes for level $l>0$.

Let us group the commitments of each such interval (of level $l>0$) into {\em ``bins''}.
Later, we shall ``charge'' the commitments in each bin on certain costs of the pseudo online algorithm $\Square$.
Consider some level $l>0$ interval $I\in\calI\langle\delta\rangle$ an input $\calR$.
We say that $I$ is a {\em committed-interval} if $I$ commits at least once in the execution of $\Dlineon$ on $\calR$.
For each committed-interval $I$ (of level $\lfun{I}>0$),
we define (almost) non-overlapping {\em``sessions''}
(one session may end at the same time the next session starts;
hence, two consecutive sessions may overlap on their boundaries).
The first session of
$I$ does {\em not} contain any commitments (and is termed an {\em uncommitted-session}); it begins at time $0$ and ends at the first time that $I$ contains some base replica.
Every other session (of $I$) contains at least one commitment (and is termed a {\em committed-session}).

Each commitment (in $\Dlineon$) of $I$ belongs to some committed session.
Denote by $\pivot(I)$ the leftmost node in $I$, i.e., $\pivot(I)=\min\{v \mid v\in I\}$.
Given a commitment $\langle I,t\rangle\in\DCOMMIT$ that $I$ makes at time $t$,
let us identify $\langle I,t\rangle$'s session. 
Let $t^{-}< t$ be the last time (before $t$) there was a base replica in $\pivot(I)$.
Similarly, let $\tplus{}> t$ be the next time (after $t$) there will be a base replica in $\pivot(I)$
(if such a time does exist; otherwise, $\tplus{}=\infty$).
The session of commitment $\langle I,t\rangle$ starts at $\tminus{}$ and ends at $\tplus{}$.
Similarly, when talking about the $i$'s session of interval $I$, we say that the session starts at $\tminus{i}(I)$
and ends at $\tplus{i}(I)$.
When $I$ is clear from the context, we may omit $(I)$%
and write $\tminus{i}$, $\tplus{i}$.
A bin is a couple $(I,i)$ of a committed-interval and the $i$th commitment-session of $I$.
Clearly, we assigned all the commitments (of level $l>0$ intervals) into bins.

Before proceeding, we claim that the bins indeed do not overlap (except, perhaps, on their boundaries).
This is because the boundaries of the sessions are times when $\pivot(I)$ has a $\Base$ replicas.
At such a times $t^*$, $I$ does not commit.
This is because the pivot of $I$ is $\langle l=0,\delta\rangle$-$\stayactive$ at $t^*$ and hence keeps a copy.
On the other hand, $I$ is of higher level
(we are dealing with the case of $l>0$);
hence, it is treated later by the algorithm (see step (S1)).
Hence, $I$ indeed does not commit at $t^*$.
Therefore, there is no overlap between the sessions, except the ending and the starting times.
That is, $\tminus{0}\leq\tplus{0}\leq\tminus{1}<\tplus{1}\leq,...,\leq\tminus{i'}<\tplus{i'}$,
where $i'$ is the number of bins that $I$ has.

Let us now point at costs of algorithm $\Square$ on which we ``charge'' the set of commitments $\DCOMMIT(I,i)$ in bin $(I,i)$ for the $i$th session of $I$.
We now consider only a bin $(I,i)$ whose  committed session  is not the last.
Note that the bin corresponds to a rectangle of $|I|$ by $t_i^+ - t_i^-$ replicas.
Expand the bin by $|I|$ replicas left, if such exist.
This yields the {\em payer} of bin $(I,i)$; that is the payer is a rectangle subgraph of $|\DNLI{I}\cup{I}|$
by $t_i^+ - t_i^-$ replicas.
We point at specific costs $\Square$ had in this payer.

Recall that every non last session of $I$ ends with a {\em base} replica in $\pivot(I)$, i.e., $(\pivot(I),\tplus{i})\in\DBase\cup\tail$.
The solution of $\Square$ contains a route ($\Square$ route) that starts at the root and reaches
$(\pivot(I),\tplus{i})$
by the definition of a base replica.
For the charging, we use {\em some} (detailed below) of the edges in the intersection of that $\Square$ route and the payer rectangle.

The easiest case is that the above $\Square$ route enters the payer at the payer's bottom ($t_i^-$) and stays in the payer until $t_i^+$%
.
In this case ({\bf EB}, for Entrance from Below), each time ($t_i^-  < t < t_i^+$) there is a commitment in the bin,
there is also an arc $a_t$ in the $\Square$ route (from time $t$ to time $t+1$).
We charge that commitment on that arc $a_t$.
%
%
%
%
The remaining case ({\bf SE}, for Side Entrance) is that
the $\Square$ route enters the payer from the left side of the payer.
(That is, $\Square$ delivers a copy to $\pivot(I)$ from some other node $u$ outside $I$'s neighborhood,
rather than stores copies at $\pivot(I)$'s neighborhood from some earlier time%
).
Therefore, the route must ``cross'' the left neighbor interval of $I$ in that payer.
Thus, there exists at least $|I|=\delta^{\lfun{I}}$ horizontal edges in the intersection between the payer ($\payer(I,i)$), of $(I,i)$  and
the $\Square$ route.

Unfortunately, the number of commitments in bin $(I,i)$ can be much grater than $\delta^{\lfun{I}}$.
However, consider some replica $(v,t^*)\in(\DBase\cup\tail)\cap I[t^*]$, where $t^*$ is the last time there was a base replica in $I$ at its $i$'th session.
The number of commitments in bin $(I,i)$ corresponding to the times {\em after} $t^*$ is $\delta^{\lfun{I}}$ at most.
(To commit, an interval must have an active node; to be active, that node needs a base replica in the last $\delta^{\lfun{I}}$ times.)
%
%
The commitments of times $t^*$ to $\tplus{i}$ are charged on the horizontal edges in the intersection between
$\payer(I,i)$
and $\Square$'s route that reach $(\pivot(I),\tplus{i})$.
Recall that, on the one hand, there are $\delta^{\lfun{I}}$ commitments at most in bin $(I,i)$ corresponding to times $t^*\leq t\leq \tplus{}$.
On the other hand, there exists at least $\delta^{\lfun{I}}$ horizontal edges in the intersection between $\Square$ route and $\payer(I)$.
%
%
%

We charge the commitments of times $\tminus{i}$ to $t^*-1$ on the arcs
in the intersection between the payer ($\payer(I,i)$), of $(I,i)$  and
the $\Square$'s route that reaches $(v,t^*)$.
(The route of $\Square$ that reach $(v,t^*)$ must contain an arc $a_t=((u,t),(u,t+1))$
in $\payer(I,i)$ for every time $t\in[\tminus{i},t^*-1]$;
this implies that in each time ($t_i^-  < t < t^*$) there is a commitment in the bin,
there is also an arc $a_t$ in $\Square$ solution (from time $t$ to time $t+1$);
we charge that commitment on that arc $a_t$.)

For each interval $I$, it is left to account for commitments in $I$'s last session.
That is, we now handle the bin $(I,i')$ where $I$ has $i'$ commitment-sessions.
This session may not end with a base replica in the pivot of $I$, so we cannot apply the argument above
(that $\Square$ must have a route reaching the pivot of $I$ at $\tplus{i'}$).
On the other hand, the first session of $I$ (the uncommitted-session) does end with a base replica in $\pivot(I)$, but has no commitments.
Intuitively, we use the payer of the first session of $I$ to pay for the commitments of the last session of $I$.
Specifically, in the first session, the $\Square$ route must enter the neighborhood of $I$ from the left side;
Hence, we apply the argument of case SE above.

To summarize,
{\bf (1)} each edge that belongs to $\Square$'s solution may be charged at most once to each payer that it belongs too.
{\bf (2)} each edge belongs to $4\log_\delta n$ payers at most (there are $\log_\delta n$ levels;
the payer rectangle of each level is two times wider than the bins; two consecutive sessions may intersect only at their boundaries)%
\footnote{
Note that, unlike the analysis of $\lineon$ for undirected line network \cite{KK2014,KK2014TR},
we don't claim that each arc is charged just for constant number of times.
}.
This leads to the term $4\log_\delta \nn$ before the $|\FSQ|$ in the statement of the lemma.
\QED

We now optimize a tradeoff between the storage coast and the delivery cost of $\Dlineon$.
On the one hand, Lemma \ref{lem:Dlineon: storage cost < Hoff + Aoff} shows that a large $\delta$
reduces the number of commitments.
By Observation \ref{obser:Dlineon: |Commit|=|Aon-0|}, this means a large $\delta$ reduces the storage cost of $\Dlineon$.
On the other hand, corollary \ref{corollary:Dlineon: delivery cost} shows that a {\em small} $\delta$ reduces the delivery cost.
To optimize the tradeoff (in an order of magnetite), fix $\delta=\lceil\frac{\log n}{\log \log n}\rceil$.
Thus, $\log_\delta n=\Theta(\frac{\log n}{\log \log n})$.
Corollary \ref{corollary:Dlineon: delivery cost},
Lemma \ref{lem:Dlineon: storage cost < Hoff + Aoff} and Observation \ref{obser:Dlineon: |Commit|=|Aon-0|}
imply that $\cost(\Dlineon,\calR)=O( \frac{\cost(\Square,\calR)\log n}{\log \log n})$.
Thus, by Theorem \ref{thm: square is O(1)-approx}, we
have the proof of the following theorem.

\begin{theorem}
Algorithm
$\Dlineon$ is $O(\frac{\log n}{\log \log n})$-competitive for
$\DMCD$ problem.
\label{thm: Dlineon is frac(log n)(log log n) competitive}
\end{theorem}

\vspace{-0.3cm}
\section{Optimal algorithm for $\RSA$ and for $\DMCD$}
\label{subsec: optRSA}

\vspace{-0.2cm}

\noindent Algorithm $\Dlineon$ in Section \ref{subsec: Algorithm Donline} solves $\DMCD$.
To solve also $\RSA$, we transform Algorithm $\Dlineon$ to an algorithm $\onRSA$ that solves $\RSA$.
First, let us view the reasons why the  solution for $\DMCD$ (Section \ref{subsec: Algorithm Donline}) does not yet solve $\RSA$.
In $\DMCD$,  the $X$ coordinate of every request (in the set $\calR$) is taken from a known set of size $n$ (the network nodes $\{1,2, ... , n\}$).
On the other hand, in $\RSA$, the $X$ coordinate of a {\em point} is arbitrary.
(A lesser obstacle is that the $Y$ coordinate is a real number, rather than an integer.)
The main idea is to make successive guesses of the number of Steinr points and of the largest $X$ coordinate and solve under is proven wrong
(e.g. a point with a larger $X$ coordinate arrives) then readjust the guess for future request.
Fortunately, the transformation is exactly the same as the one used in \cite{KK2014TR,KK2014} to transform the algorithm for undirected $\MCD$ to solve $SRSA$.
For completeness, we nevertheless present the transformation here.

\subsection{Proof Outline}
\label{App;subsec:Proof Outline}
The following outline is taken (almost) word for word from \cite{KK2014}.
(We made minor changes, e.g. replacing the word $SRSA$ by the word $\RSA$).

First, let us view the reasons why the  solution for $\DMCD$ (Section \ref{subsec: Algorithm Donline}) does not yet solve $\RSA$.
In $\DMCD$,  the $X$ coordinate of every request (in the set $\calR$) is taken from a known set of size $n$ (the network nodes $\{1,2, ... , n\}$).
On the other hand, in $\RSA$, the $X$ coordinate of a {\em point} is arbitrary.
(A lesser obstacle is that the $Y$ coordinate is a real number, rather than an integer.)
The main idea is to make successive guesses of the number of Steinr points and of the largest $X$ coordinate and solve under is proven wrong
(e.g. a point with a larger $X$ coordinate arrives) then readjust the guess for future request.
Let us now transform, in three conceptual stages, $\Dlineon$ into an optimal algorithm for the online problem of $\RSA$:
\vspace{-0.2cm}
\begin{enumerate}
     \item Given an instance of $\RSA$, assume temporarily (and remove the assumption later) that the number $N$ of points is known, as well as $M$, the maximum $X$ coordinate any request may have. Then, simulate a network where $n\geq \NN$ and $\sqrt{\log n}=O(\sqrt{\log \NN})$, and the $n$ nodes are spaced evenly on the interval between $0$ and $M$. Transform each $\RSA$ request to the nearest grid point. Solve the resulting $\DMCD$ problem.

     \item Translate  these results to results of the original
           $\RSA$ instance.

\item  Get rid of the assumptions.
\end{enumerate}
\vspace{-0.2cm}
The first stage is, of course, easy. It turns out that ``getting rid of the assumptions'' is also relatively easy.
To simulate the assumption that $M$ is known, guess that $M$ is some $M_j$.
Whenever a guess fails, (a request $r_i=(x_i,t_i)$ arrives, where $x_i>M_j$), continue with an increased guess $M_{j+1}$.
A similar trick is used for guessing $N$. In implementing this idea, our algorithm turned out paying a cost of $\Sigma M_j$.
(This is $M_j$ per failed guess, since each application of $\Square$ to a new instance, for a new guess, starts with delivering a copy to every node in the simulated network; see the description of Algorithm $\Square$.) On the other hand, an (optimal) algorithm that knew $M$ could have paid $M$ only once.
IF $M_{j+1}$ is ``sufficiently'' larger than $M_j$, then $\Sigma M_j=O(M)$.

The second stage above (translate the results) proved to be somewhat more difficult, even in the case that $N$ and $M$ are known (and even if they are equal).
Intuitively, following the first stage, each request $r_i=(x_i,t_i)$ is inside a grid square. The solution of $\DMCD$ passes via a corner of the grid square. To augment this into a solution of $\RSA$, we need to connect the corner of the grid square to
 $r_i$. This is easy in an offline algorithm. However, an online algorithm is not allowed to connect a point at the top of the grid square (representing some time t) to a point somewhere inside the grid square (representing some earlier time $t-\epsilon$).

Somewhat more specifically,
following the first stage,
each request $r_i=(x_i,t_i)$ is in some grid square, where the corners of the square are points of the simulated $\DMCD$ problem. If we normalize $M$ to be $N$, then the left bottom left corner of that square is $(\lfloor x_i \rfloor , \lfloor t_i \rfloor )  )$.
Had we wanted an  {\em offline} algorithm, we could have solved an instance of $\DMCD$, where the points are  $(\lfloor x_1 \rfloor , \lfloor t_1 \rfloor ), (\lfloor x_2 \rfloor , \lfloor t_2 \rfloor ), (\lfloor x_3 \rfloor, \lfloor t_3 \rfloor ), ...$.
Then, translating the results of $\DMCD$ would have meant just augmenting with segments connecting each $(\lfloor x_i \rfloor , \lfloor t_i \rfloor )$ to $(x_i,t_i)$.
Unfortunately, this is not possible in an {\em online} algorithm, since  $(x_i,t_i)$ is not yet known at $(\lfloor t_i \rfloor )$.
Similarly, we cannot use the upper left corner of the square (for example) that way, since at time $\lceil t_i \rceil$, the algorithm may no longer be allowed to add segments reaching the earlier time $t_i$.

\subsection{Informal description of the transformed {\bf\em RSA} algorithm assuming  $\nn/2\leq\xmaxQ\leq n$ and $\sqrt[4]{\nn}\leq\NN\leq \nn$ and $\nn$ is known}
\label{subsec: onRSAn}

The algorithm under the assumptions above appears in Figure \ref{figure: onRSAn}.
Below, let us explain the algorithm and its motivation informally.

When describing the solution of $\DMCD$, it was convenient for us to assume that the network node were $\{1,...,\nn\}$.
In this section (when dealing with $\RSA$),
it is more convenient for us to assume that $\Dlineon$ solves $\DMCD$ with the set of network nodes being $\{0,...,\nn-1\}$.
Clearly, it is trivial (though cumbersome) to change $\Dlineon$ to satisfy this assumption.

Assume we are given a set of points $\calQ=\{p_1=(x_1,y_1,...,(x_N,y_N))\}$ for $\RSA$.
We now translate $\RSA$ points to $\DMCD$ requests (Fig. \ref{fig: pointsTogrid}).
That is, each point $p_i=(x_i,y_i)$ that is not already on a grid node, is located inside some square whose corners are the grid vertices.
We move point $p_i=(x_i,y_i)$ to the grid vertex (replica) $r_i=(v_i,t_i)$ on the left top corner of this square.
That is, we move $p_i$ (if needed) somewhat later in time, and somewhat left on the $X$ axis.
We apply $\Dlineon$ to solve the resulting $\DMCD$.
This serves $r_i=(v_i,t_i)$ from some other replica $(u,t_i)$, where $t_i$ may be slightly later than the time $y_i$ we must serve $p_i$.
After $\Square$ solves the $\DMCD$ instance, we modify the $\DMCD$ solution to move the whole horizontal route
$\DPHon(i)$ of request $r_i$
(route from $\qon_i=(\uon_i,t_i)$ to $r_i=(v_i,t_i)$
somewhat earlier in time (from time $t_i$ to time $y_i$).
This now serves a point $(v_i,y_i)$, where $v_i$ may be slightly left of $x_i$.
Hence,
we extend the above horizontal route
by the segment from $(v_i,y_i)$ to $p_i=(x_i,y_i)$.
%
In addition,
the transformed algorithm leaves extra copies in every network node along the route $\DPHon(i)$, until time $t_i$ (see Fig. \ref{fig: onRSAmn execution}(d));
a little more formally, the algorithm adds to the solution of $\RSA$ the
vertical line segment
$\linev{(k, {y_i}),(k,t_i)}$
(a vertical segment between the points $(k, {y_i})$ and $(k,t_i)$),
for every $k$ such that $(k,t_i)\in\Von(i)$.

\def\FigpointTogrid{
\begin{figure}[ht!]
\begin{center}
\includegraphics[scale=0.45]{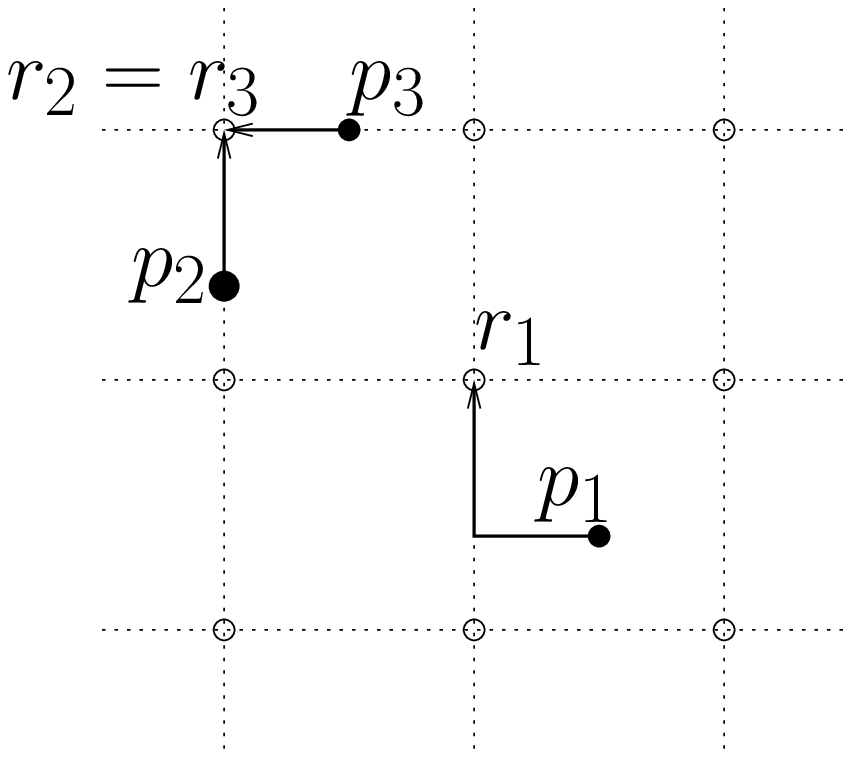}
\end{center}
\caption{\sf Point $p_1$ is transformed upward and leftward to $\gamma_1$; $p_2$ is transformed upward and $p_3$ is is transformed leftward;  the points transform to the same vertex point.
\label{fig: pointsTogrid}}
\end{figure}
} 
\FigpointTogrid

There is a technical point here:
$\Dlineon$ had a copy in $(u,t_i)$ and we need a ``copy''  in $(u,y_i)$ where $t_i-1 < y \leq t_i$.
That is, we need that the solution of $\RSA$ problem will already includes $(u,y_i)$.

\begin{observation}
The solution of $\RSA$ problem already includes $(u,y_i)$.
\label{obser:SRSA: (u,y) in the solution}
\end{observation}
\def\AppObserSRSAuyInTheSolution{
\proof
To make sure such a copy in $(u,y_i)$ does exist, let us consider the way the copy reached $(u,t_i)$ in $\Dlineon$.
If $\Dlineon$ stored a ``copy'' in $u$ from time $t-1$ to $t$
(see Fig. \ref{fig: onRSAmn execution}(c)
),
then also $(u,y_i)$ belong to the solution.
Otherwise, $\Dlineon$ moved the copy to $(u,t_i)$ over a route
$\calP_\calH$
from some other grid vertex $(w,t_i)$.

Note that $(w,t_i)$ appeared in the transformed algorithm because that algorithm served a point
${p_i} = ({ x_i}, { y_i})$, of a time
$t_i-1  <   { y_i}  < t_i$
(see Fig. \ref{fig: onRSAmn execution}(d)%
).
The transformed algorithm moved this route $\calP_\calH [(w,t),(u,t)]$
earlier in time to
$\tilde{\calP}_\calH [(w,{ y_i}),(u, {y_i})]$ and left copies in those network node until time $t_i$ (see Fig. \ref{fig: onRSAmn execution}(d)).
In particular, it leave a copy also in $u$ from time ${y_i}$ to time $t_i$, hence $(u,y_i)$ is already in the solution of the transform algorithm.
\QED
} 
\AppObserSRSAuyInTheSolution

So far, we described how to transform the delivery phase of $\Dlineon$. The storage phase of $\Dlineon$ does not need to be transformed.
(Actually, $\DMCD$ even has some minor extra difficulty that does not exist in $\RSA$;
consider some request $r_{i-1}=(v_{i-1},t_{i-1})$ in $\DMCD$, and suppose that the
next request $r_{i}=(v_{i},t_{i})$ is at time $t_{i}=t_{i-1}+10$; then time $t+1$ arrives, and $\Dlineon$ must make some decisions, without knowing that the next request will be at time $t_{i-1}+10$; then time $t+2$ arrives, etc; no such notion of time passing (without new points arriving) exists in the definition of $\RSA$; that is, the $Y$ coordinate $y_{i}$ of the next request $p_{i}=(x_{i},y_{i})$ is known right after the algorithm finished handling $p_i=(x_i,y_i)$; the storage phase of the transformed algorithm does not make any use of this extra freedom in $\RSA$
and simulates the ``times'', or the $Y$ coordinates, one by one; note that for that purpose, the transformation of the delivery phase ensured the following property: that if a copy in
$\DMCD$ exists in a replica $(v,t)$ in $\Dlineon$, this replica also contains a copy in the transformed algorithm.)
Denote the solution of $\onRSAn$ on $\calQ$ by $\FRSAn(\calQ)$.
For the pseudo code, see Fig. \ref{figure: onRSAn}.

\def\FigCodeOnRSAn{
\begin{figure}[ht!]
\fboxsep=0.2cm
\framebox[\textwidth]{
\begin{minipage}{0.8\textwidth}

\begin{enumerate}

\item For $p_1$ do:
\begin{enumerate}
    \item compute the translated request $r_1=(v_1,t_1)$ of $p_1$; $\calR\leftarrow\{r_1\}$.
    \item $\FRSAn\leftarrow\{\linev{(0,0),(0,y_1)},\lineh{(0,y_1),(x_1,y_1)}\}$;
\end{enumerate}


\item For each point $p_i\in\calQ\setminus\{p_1\}=\{p_2=(x_2,y_2),...,p_\NN=(y_\NN,y_\NN)\}$ do:

\begin{enumerate}

\item compute the translate request $r_i=(v_i,t_i)$ of $p_i$;
\item $\calR\leftarrow\calR\cup\{r_i\}$.

\item \underline{``Vertical phase''}
    \begin{enumerate}

    \item If $t_i>t_{i-1}$, then for each time $t=t_{i-1},...,t_i-1$ do:

        \begin{enumerate}
        \item ``Simulate'' $\Dlineon$ on $\calR$ to find $\calC_{t+1}$.

        \item $\FRSAn\leftarrow\FRSAn\cup\{\linev{(v,t),(v,t+1)} \mid \rep{v}{t+1}\in\calC_{t+1}\}$.

        \end{enumerate}
    \end{enumerate}

\item \underline{``Horizontal phase''}

    \begin{enumerate}
    \item ``Simulate'' $\Dlineon$ on $\calR$ to find $\DPVon(i)$.\\

    \item $\FRSAn\leftarrow\FRSAn\cup\{\lineh{({\uon_i},y_i),({v_i},y_i)}\}$.

    \item $\FRSAn\leftarrow\FRSAn\cup\{\lineh{({v_i},y_i),(x_i,y_i)}\}$

    \item $\FRSAn\leftarrow\FRSAn\cup\{\linev{({u},y_i),({u},{t_i})} \mid (u,t_i)\in\PVon(i)\}$.

    \end{enumerate}

\end{enumerate}
\item Return $\FRSAn(\calQ)$
\end{enumerate}
\end{minipage}
}
\caption{\label{figure: onRSAn}
Subroutine $\onRSAn$ assumes the knowledge of $n$ and that $\nn/2\leq\xmaxQ\leq \nn$ and $\sqrt[4]{\nn}\leq \NN\leq \nn$.
}
\end{figure}

} 
\FigCodeOnRSAn


\paragraph*{\bf Analysis sketch of the transformed algorithm with known parameters%
\commsingle.\commsingleend}

It is not hard to see that an optimal solution for that instance of $\DMCD$ is ``not that far'' from an optimal solution of the original instance of $\RSA$.
To see that, given an optimal solution of $\RSA$, one can derive a feasible solution of the resulting $\DMCD$ by adding 2 segments of length at most $1$  for each point $p$. (One vertical such segment plus a horizontal one are enough to connect a point $p$ to the replica
$(v,t)$ where we moved $p$).
 The total of those distances is $2\nn$ at most.
On the other hand, an optimal solution of $\RSA$ would need to pay at least $\xmaxQ\geq \nn/2$.
Hence, an optimal solution for $\DMCD$ would have implied a constant approximation of $\RSA$.
Intuitively, an approximation (and a competitive ratio) for $\DMCD$ implies an approximation (and a competitive ratio) of $\RSA$ in a similar way.
%
%
%
For a given Algorithm $A$ for $\RSA$ and a set $\calQ$ of input points, let $\cost(A,\calQ)$ be the cost of $A$ on $\calQ$.
Let $\opt$ be an optimal algorithm for $\RSA$.

\begin{lem}
Assume that $\xmaxQ\leq \nn$ and $\NN\leq \nn$.
Then, $\cost(\onRSAn,\calQ)= O(\loglogratio{n}(\cost(\opt,\calQ)+\nn))$.
If also $\sqrt[4]{\nn}\leq\NN$ and $\nn/2\leq\xmaxQ$, then $\onRSAn$ is  $O(\loglogratio{N})$-competitive for $\RSA$.
\label{lemma: onRSAn is O(sqrt log n)-competitve}
\end{lem}


\def\ProofLEMMAonRSAnLOGcomp{
\proof
It is easy to verify that $\onRSAn$ computes a feasible solution (see the ``technical point'' comments in parentheses in section \ref{subsec: onRSAn}).
%
Consider some input point set $\calQ=\{p_1=(x_1,y_1),...,p_\NN=(x_\NN,y_\NN)\}$ such that $\xmaxQ\leq \nn$
and $\NN\leq\nn$.
Let $\calR=\{r_1=(v_1,t_1),...,r_\NN=(v_\NN,t_\NN)\}$ be the translated instance of the $\MCD$ problem.

Recall how does $\onRSAn(\calQ)$ translate the solution of
\commdouble\\\commdoubleend
$\Dlineon(\calR)$.
An horizontal edge $((u,t_i),(u+1,t_i))\in\PHon(i)$ (that $\Dlineon$ add to its solution when handling request $r_i$, see step (D2) in $\Dlineon$)
is translated into a horizontal line segment $\lineh{(u,y_i),(u+1,y_i)}$.
An arc $((u,t),(u,t+1))\in\Aon$ (of $\Dlineon$'s solution on $\calR$) is translated into a vertical line segment $\linev{(u,t),(u,t+1)}$.
Hence, the total cost of those parts of the solution of $\onRSAn(\calQ)$ is exactly the same as the cost of the solution of $\Dlineon(\calR)$.

Thus, the cost of $\onRSAn$ on $\calQ$ differ from the cost of $\Dlineon$ on $\calR$ only by
two kinds of ``short'' segments (Segment of length at most 1).
For the first kind, recall (technical point in Section \ref{subsec: onRSAn})
that for every moved horizontal path $\tilde{\calP}_\calH[(w,\tilde{y}),(u,\tilde{y})]$,
$\onRSAn$ added a short vertical segment for every network node $w'$ of that path from $(w',\tilde{y})$ to $(w',\lceil y\rceil)$.
The second kind of addition is an horizontal short segment connecting the input point $p=(x,y)$ to $(u,y)$, where $u=\lfloor x\rfloor$.

The total cost of the second kind is bounded by $n$, since $|v_i-x_i|\leq1$.
We claim that the total cost of the short segment of the first kind is $\cost(\Dlineon,\calR)$ at most.
To see that, notice that we have at most 1 such ``short'' segment (shorter than 1) per replica that appears in the solution of $\Dlineon$ on $\calR$.
That solution of $\Dlineon$ contains at least as many edges as it contains replicas.
Formally, the cost of $\onRSAn$ is at most,
\commsingle
\begin{eqnarray*}
\cost(\onRSAn,\calQ)&=&\cost(\Dlineon,\calR)+\sum_{i=1}^{\nn}\big(|\PVon(i)|\cdot(t_i-y_i)
+
|v_i-x_i|\big)\nonumber\\
&\leq& 2\cost(\Dlineon,\calR)+\nn.
\label{ineq: cost onRSAmn Q leq cost onalgmn Z}
\end{eqnarray*}
\commsingleend
\commdouble
\begin{eqnarray*}
&&\cost(\onRSAn,\calQ)=\nonumber\\
&&\cost(\Dlineon,\calR)+\sum_{i=1}^{\nn}\big(|\PVon(i)|\cdot(t_i-y_i)
+
|v_i-x_i|\big)\nonumber\\
&&\leq 2\cost(\Dlineon,\calR)+\nn.
\label{ineq: cost onRSAmn Q leq cost onalgmn Z}
\end{eqnarray*}
\commdoubleend
Thus, by Theorem \ref{thm: Dlineon is frac(log n)(log log n) competitive},
\begin{eqnarray}
\cost(\onRSAn,\calQ)\leq
c_1\loglogratio{\nn}\cdot\cost(\opt,\calR)+\nn,
\label{ineq: c(onRSAn, Q) leq c1 sqrt log c(opt,R)+n}
\end{eqnarray}
where $c_1$ is some constant.
%

Let us look the other direction, from an {\em optimal} solution of $\RSA$ for $\calQ$ to optimal solution of $\DMCD$ for $\calR$.
Recall that $r_i$ can be served from $p_i$ at a cost of 2 (at most).
Hence,
\begin{equation}
\cost(\opt,\calR) \leq \cost(\opt,\calQ)+2\nn.
\label{ineq: cost opt Q geq cost Z' -2M size/ netsize}
\end{equation}
Thus, by Inequalities (\ref{ineq: c(onRSAn, Q) leq c1 sqrt log c(opt,R)+n}) and (\ref{ineq: cost opt Q geq cost Z' -2M size/ netsize}),
\commsingle
\begin{eqnarray}
\label{ineq: c(onRSAn Q) leq O(sqrt log n) (c(opt Q)+n)}
\cost(\onRSAn,\calQ)&\leq&
c_1\loglogratio{n}\cdot(\cost(\opt,\calQ)+2\nn)+\nn\\
&=&O(\loglogratio{n}\cdot(\cost(\opt,\calQ)+\nn).\nonumber
\end{eqnarray}
\commsingleend
\commdouble
\begin{eqnarray}
\label{ineq: c(onRSAn Q) leq O(sqrt log n) (c(opt Q)+n)}
&&\cost(\onRSAn,\calQ)\\
&&\leq c_1\loglogratio{n}\cdot(\cost(\opt,\calQ)+2\nn)+\nn\nonumber\\
&&=O(\loglogratio{n}\cdot(\cost(\opt,\calQ)+\nn).\nonumber
\end{eqnarray}
\commdoubleend
The first statement of the lemma holds.
%
%
Now, let us prove the second statement of the lemma.
Assume that $\xmaxQ/2\leq\NN\leq\nn$ and $\sqrt[4]{\nn}\leq\NN\leq \nn$.
Thus also,
\begin{equation}
\cost(\opt,\calQ) \geq \nn/2.
\label{ineq: cost opt Q geq 2/M}
\end{equation}
Therefore, by Inequalities (\ref{ineq: c(onRSAn Q) leq O(sqrt log n) (c(opt Q)+n)}) and (\ref{ineq: cost opt Q geq 2/M}),
\commsingle
\begin{eqnarray*}
\frac{\cost(\onRSAn,\calQ)}{\cost(\opt,\calQ)}
&\leq&
\frac{c_1\loglogratio{n}\cdot\cost(\opt,\calQ)}
{\cost(\opt,\calQ)}+
\frac{c_1\loglogratio{n}\cdot2\nn+\nn}
{\nn/2}
\\
&\leq& (5c_1+1)\loglogratio{n}
.
\end{eqnarray*}
\commsingleend
\commdouble
\begin{eqnarray*}
&&\frac{\cost(\onRSAn,\calQ)}{\cost(\opt,\calQ)}\\
&&\leq
\frac{c_1\loglogratio{n}\cdot\cost(\opt,\calQ)}
{\cost(\opt,\calQ)}+
\frac{c_1\loglogratio{n}\cdot2\nn+\nn}
{\nn/2}
\\
&&\leq (5c_1+1)\loglogratio{n}.
\end{eqnarray*}
\commdoubleend
The lemma follows, since $\sqrt[4]{\log \nn}\leq \NN$.
\QED
}
\ProofLEMMAonRSAnLOGcomp

Below,  $\onRSAn$ is used as a module in another algorithm, responsible for implementing the assumptions.
In each execution of the other algorithm, $\onRSAn$ is invoked multiple times, for multiple subsets of the input.
Unfortunately, not every time, the other algorithm uses $\onRSAn$, all the assumptions are ensured.
This is the reason of the ``extra'' factor $n \sqrt{\log n}$ in the first part of the above lemma above.
Fortunately, these extra factors of all the invocations are bounded separately later.

\def\FigonRSAnEXE{
\begin{figure}
\begin{center}
\includegraphics[scale=0.4]{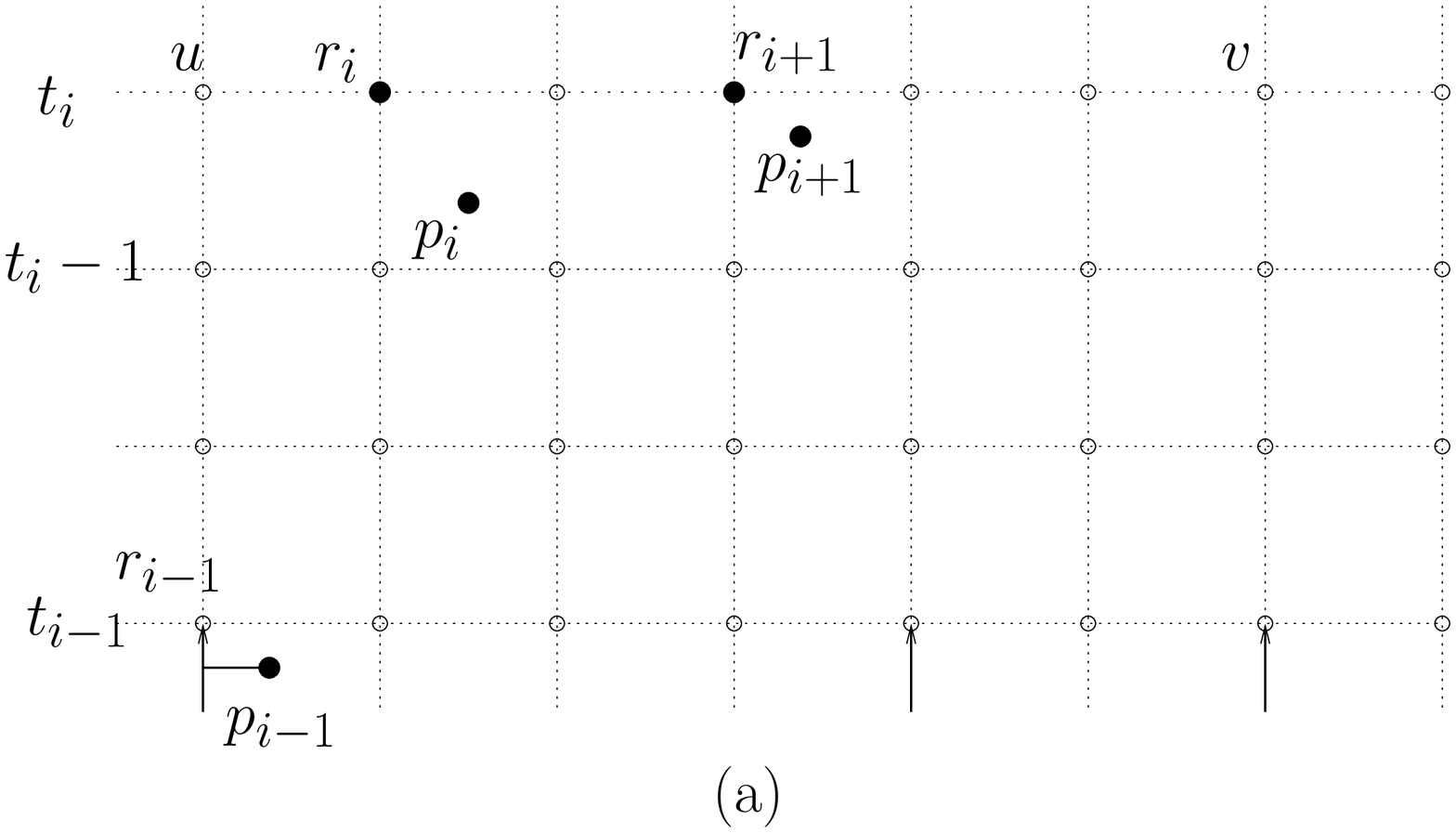}
\hfill
\includegraphics[scale=0.4]{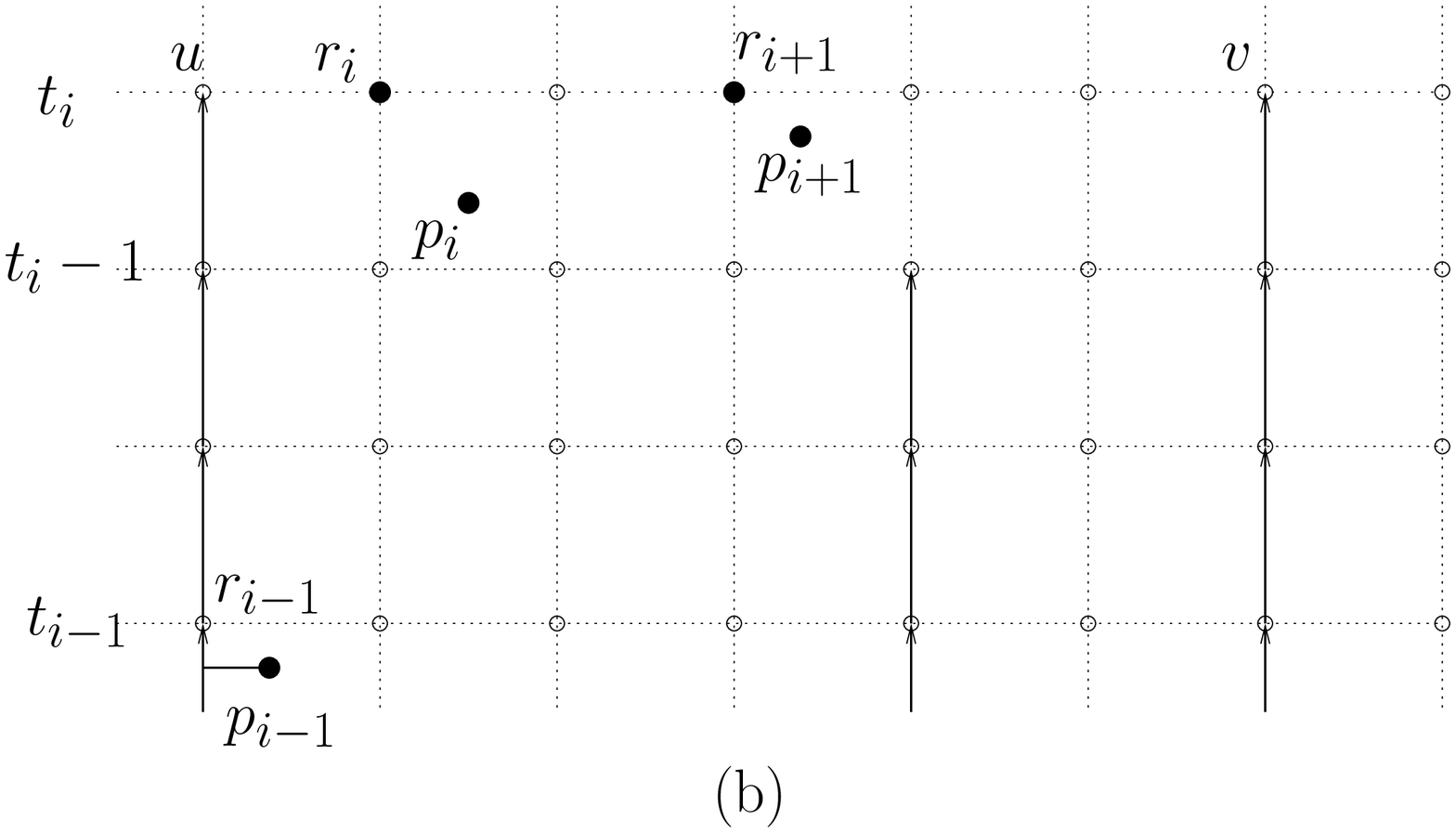}
\hfill
\includegraphics[scale=0.4]{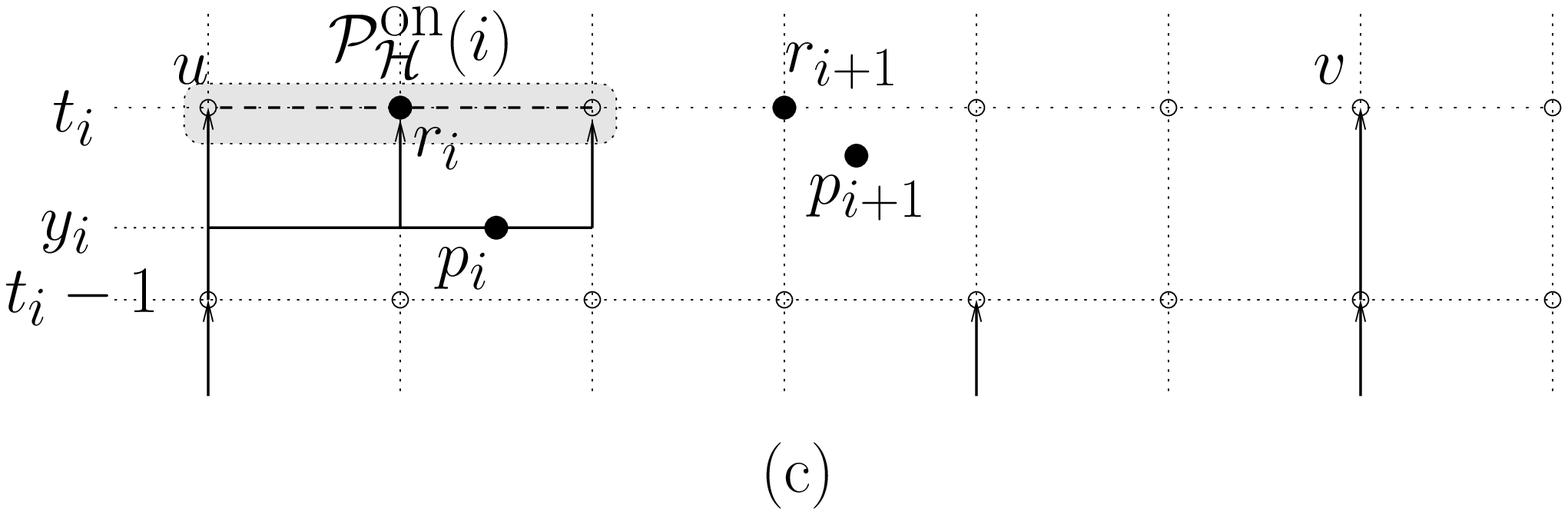}
\hfill
\includegraphics[scale=0.4]{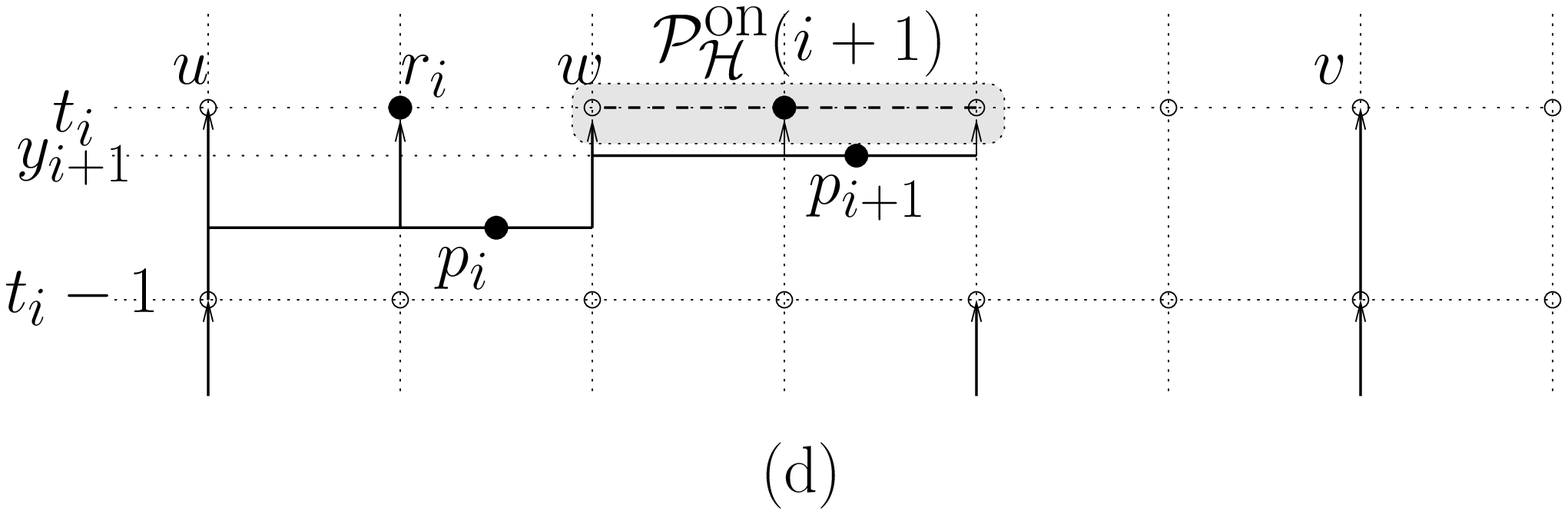}
\end{center}
\caption{\sf Example of execution of $\onRSAn$.
(a) $\onRSAn$'s solution after handling point $p_{i-1}$;
(b) $\onRSAn$ simulates the storage phase of $\Dlineon$ on $\calR$ for times $t=t_{i-1},...,t_i-1$;
(c) $\onRSAn$ handles point $p_i$, moves $\PHon(i)$ from ``time'' $t_i$ to ``time'' $y_i$ (it serves this path from $(u,y_i)$, who ``receives a copy'' when $\Dlineon$ handles time $t_i-1$ in the storage phase),
and ``leaves copies'' at the nodes $\PVon(i)$ from ``time'' $y_i$ to ``time'' $t_i$;
(d) $\onRSAn$ handles point $p_{i+1}$,
 moves $\PHon(i+1)$ from ``time'' $t_i$ to ``time'' $y_{i+1}$ (it serves this path from $(w,y_{i+1})$, who ``receives a copy'' when handling point $p_i$), and
 ``leaves  copies'' at the nodes of $\PVon(i+1)$ from ``time'' $y_{i+1}$ to ``time'' $t_i$.
\label{fig: onRSAmn execution}}
\end{figure}
} 
\FigonRSAnEXE

\subsection{Getting rid of the assumption that {\bf\em M=N}}
\label{subsec: onRSAmn}
We now describe an online algorithm $\onRSAmn$ that is somewhat more general than $\onRSAn$. Algorithm $\onRSAmn$ is {\em not} based on the assumption that the upper bound $M$ on $\xmaxQ$ is also the number of points.
That is, we now do {\em not} assume that $M=N$. Getting rid of this assumption is straightforward.
The new online algorithm $\onRSAmn$ transforms the $X$ coordinate of each input point to the interval $[0,n]$. Algorithm $\onRSAmn$ passes the transformed point to the online algorithm $\onRSAn$ of Section \ref{subsec: onRSAn} that is assumed to be executing in parallel. The transformation of a point is, though, a little more involved, as detailed below.

Later on (in Section \ref{subsec: onRSA}), $\onRSAmn$ will be used by an even more general algorithm in a similar way. For that, it is more convenient for us to define algorithm $\onRSAmn$ a somewhat more general algorithm then is needed by the description so far.
We now assume that the origin is not necessarily $(0,0)$, but is rather some $p_0=(0,y_0)$.
(Meanwhile, we still assume that $x_0=0$).
 Hence, algorithm $\onRSAmn$ translates the $X$ coordinate of each input point $p=(x,y)$ to $\ff(x)= x \cdot \frac{N}{M}$.
To keep the proportion between the axes, the $y$ coordinate $y$ is translated to
$\ff(y)=(y-y_0) \cdot \frac{N}{M}$. (Recall that $y \ge t_0$.)
Finally, the solution of $\onRSAn$ is translated back to the coordinates of $\onRSAmn$ applying the transformation $\ff^-$ to every point of the solution.
(Clearly, this is a polynomial task, since the solution is described using a polynomial number of points).
The pseudo code appears in Fig. \ref{fig:middle-alg}.
By Lemma \ref{lemma: onRSAn is O(sqrt log n)-competitve} and the description of $\onRSAmnp$, it is easy to see the following.

\begin{observation}
Assume that $\xmaxQ\leq \MM$ and $\NN\leq\nn$.
Then, $\cost(\onRSAmn,\calQ) = O(\loglogratio{n}(\cost(\opt,\calQ)+\MM)) $.
If also $\sqrt[4]{\nn}\leq\NN$ and $\MM/2\leq\xmaxQ$, then $\onRSAmn$ is
\commdouble\\\commdoubleend
$O(\loglogratio{N})$-competitive for
$\RSA$.
\label{corollary: onRSAmn is O(sqrt)-cometitve}
\end{observation}

\def\FigCodeOnRSAmn{
\begin{figure}[ht!]
\fboxsep=0.2cm
\framebox[\textwidth]{
\begin{minipage}{0.8\textwidth}

$\bullet$ origin is $p_0=(x_0,y_0)$.

\begin{enumerate}

\item $\calQ'\leftarrow\emptyset$.

\item For each point $p_i\in\calQ$ do:
\begin{enumerate}
    \item $p'_i=(x'_i,y'_i)\leftarrow \ff(p_i,\MM,\nn,y_0)$;
    \item $\calQ'\leftarrow\calQ'\cup\{p'_i\}$.
    \item Call $\onRSAn$ as a subroutine on $\calQ'$ to find $\FRSAn(\calQ')$;
    \item $\FRSAmn\leftarrow\ff^{-1}(\FRSAn(\calQ'),\MM,\nn,y_0)$;
\end{enumerate}

\end{enumerate}
\end{minipage}
}
\caption{\label{fig:middle-alg}
Algorithm $\onRSAmnp$.
}
\end{figure}

} 
\FigCodeOnRSAmn


\subsection{Getting rid of the {knowledge} assumptions}
\label{subsec: onRSA}

To give up the assumption that $\xmaxQ$ is known, we use a standard trick.
We first guess that $\xmaxQ$ is
``about'' twice the $X$ coordinate of the first point.
Whenever the guess for $\xmaxQ$ is proven wrong (some $p_i=(x_i,y_i)$ arrives with $x_i$ larger then our guess for $\xmaxQ$),
we double the guess.
We do not change the solution for the points we already served.
Simply, the points that arrive from now on, are treated as a new instance of $\RSA$, to be solved (by $\onRSAmnp$) by a translation to a new instance of $\DMCD$.
Intuitively, every instance of $\DMCD$ may need to pay an additional cost
that is proportional to our current guess of $\xmaxQ$.
This is justified by the fact that
(1) the new guess is at least double our previous guess of $\xmaxQ$;
and (2) any optimal algorithm would need now to pay $\xmaxQ$ guessed before.
(A minor technical point is that the origin of the new instance of $\RSA$ may not be
$p_0=(0,0)$;
instead, the new origin is $(0,y_{i-1})$, where $y_{i-1}$ is the $y$-coordinates of the last point served.)

For justifying the other assumption, that the number of points is known in advance,
we use a similar trick; however, its justification is more complex.
That is, if the number of points grows larger beyond our current guess, $\nguess$,
we increase our guess of the number of points.
We then start a new instance of $\onRSAmnp$ with the new guess.
(In turn, this leads to a new activation of $\Dlineon$ with $\nguess^{new}$ as the new network size.)
Hence, we start a new $\DMCD$ instance with an increased ``network size''.
The ``new'' guess $\nguess^{new}$ of the number of $\RSA$ points is (not doubled but) the
power of 4 of our ``current'' $\nguess$ (yielding a double exponential sequence).
Each new $\DMCD$ instance  is associated with a cost of $O(\sqrt{\log{\nguess^{new}}}\xmaxQ)$ at most.
Thanks to using a double exponential groth rather than an exponential growth, this would increase the competitive ratio just by a factor of $O(\log\log N)$.
Clearly, one should not increase the guess (of the number of points) more than polynomially each time
(since otherwise, for the last guess $\widetilde{\nguess}$, the value would have been too high compared to the desired $\loglogratio{N}$ competitive ratio.)
%
%
%
Summarizing the above informal description, given an instance of $\RSA$, we use ``guesses'' of $\xmaxQ$ and $N$
to partition the points $\calQ$ into subsets.
Each such subset defines a problem we translate separately to $\DMCD$ via $\onRSAmnp$.

Given an instance of $\RSA$, we now define its partition of multiple instances. For that, we define the partition of $\calQ$ into subsets $\calQ\lrangle{1}$, $\calQ\lrangle{2}$, .... The first $|\calQ\lrangle{1}|$ points will belong to $\calQ\lrangle{1}$, the next $|\calQ\lrangle{2}|$ will belong to $\calQ\lrangle{2}$, etc. We shall also show how to detect online the first point in $\calQ\lrangle{2}$, the first in $\calQ\lrangle{3}$, etc.
Before that, we must tackle some technicality. The original $\RSA$ problem with defined for an origin of $X=0$ and $Y=0$.
However, after solving for the $\RSA$ instance $\calQ\lrangle{1}$, the next point is at $Y$ coordinate that is larger than zero.
Moreover, when solving $\DMCD$, we allowed the origin to be at any node (that is, in any $X$ coordinate).
Hence, it is convenient to generalize
the definition of the $\RSA$ to the setting were the input includes an origin point $p_0=(x_0,y_0)$, in the positive quadrant. The input point set $\calQ\lrangle{k}$ includes only points (in the positive quadrant), whose $y$-coordinates are grater than or equal to $y_0$.

Consider a point set $\calQ=\{p_1,...,p_N\}$.
Algorithm $\onRSA$ partitions $\calQ$ into subsets as follows.
For every $i=1,...,N$ define that
\vspace{-0.0cm}
\begin{eqnarray}
\Mguess(i)=2^{l'},
\label{eq:def: Mguess}
\end{eqnarray}
where $l'=\lceil \log\max\{x_j\mid j=1,...,i\} \rceil$, and
\vspace{-0.0cm}
\begin{eqnarray}
\nguess(i)= \tetration{l^*},
\label{eq:def: nguess}
\end{eqnarray}
where $l^*$ is integer such that $l^*=\min_l(\tetration{l}\geq i)$.
Note that, $\tetration{l+1}=2^{(4\cdot2^l)}=(\tetration{l})^4$. Hence, the growth of the sequence $\tetration{\cdot0},\tetration{\cdot1},\tetration{\cdot2},...$ is for the power of 4.

Let us use the above guesses to generate the subset.
Specifically, we generate a sequence $g_1<g_2<...<g_\tau$ (for some $g_\tau$)
of separators between consecutive subsets.
That is, $\calQ\lrangle{1}=\{p_{g_1},...,p_{{g_2}-1}\}$, then $\calQ\lrangle{2}=\{p_{g_2},...,p_{{g_3}-1}\}$, etc.
A separator is the index of a point where one of the guess fails.
Specifically, let $g_1=1$ and if $\Mguess({g_k})<\Mguess(N)$ or $\nguess(g_k)<\nguess(N)$, then let
\vspace{-0.0cm}
\begin{eqnarray}
g_{k+1}\triangleq&\min_i&\big(\Mguess(g_k)<\Mguess(i)
\mbox{ or }
\nguess(g_k)<\nguess(i)\big).
\label{def: g k+1= separator}
\end{eqnarray}
Define that the guess $\nguess$ of $\calQ\lrangle{k}$ is
$\nn_k=\tetration{(\nguess(g_k))}$ and the guess $\Mguess$ of $\calQ\lrangle{k}$ is
$\MM_k=2^{\Mguess({g_k})}$, for every $k=1,...,\tau$.
The origin points of these subsets are defined as follows:
Let
$y_{last}^k=$ be the $y$-axis of the {\em last} point $p_{g_{k+1}-1}$ in $\calQ\lrangle{k}$
and let $y_0^1=0$ and $y_0^k=y_{last}^{k-1}$ (for $k=2,...,\tau$).
The origin point of $\calQ\lrangle{k}$ is $p_0^k=(0,y_0^k)$, for every $k=1,...,\tau$ (see Fig. \ref{fig: onRSA partition to groups}).

All the above functions can be computed online.
As
sketched,
Algorithm $\onRSA$ handles a point after point, and a subset after subset.
For every point $p_i\in\calQ$, $\onRSA$ finds the subset $\calQ\lrangle{k}$ that $p_i$ belongs to (i.e., $p_i\in\calQ\lrangle{k}$), then $\onRSA$
passes the point to an instance of $\onRSAmnp$ executing (in parallel to $\onRSA$) on $\calQ\lrangle{k}$,
with the origin point $p_0^k=(0,y_0^{k})$, and with the $\Mguess$ parameter $\MM=\MM_k$
and the $\nguess$ parameter $n=n_k$.
Denote the solution
of $\onRSAmnp$
on $\calQ\lrangle{k}$ by $\FRSAmn(\calQ\lrangle{k})$.
The solution of $\onRSA$ is the union of the solutions of $\onRSAmnp$ on all the subsets. That is,
$\onRSA$'s solution is $\FRSA(\calQ)\equiv\cup_{k=1}^{\tau}\RSAFon(\calQ\lrangle{k})$.
The pseudo code of $\onRSA$ is given in Fig. \ref{figure: onRSA}.

\def\FigonRSA{
\begin{figure}[ht!]
\begin{center}
\includegraphics[scale=0.4]{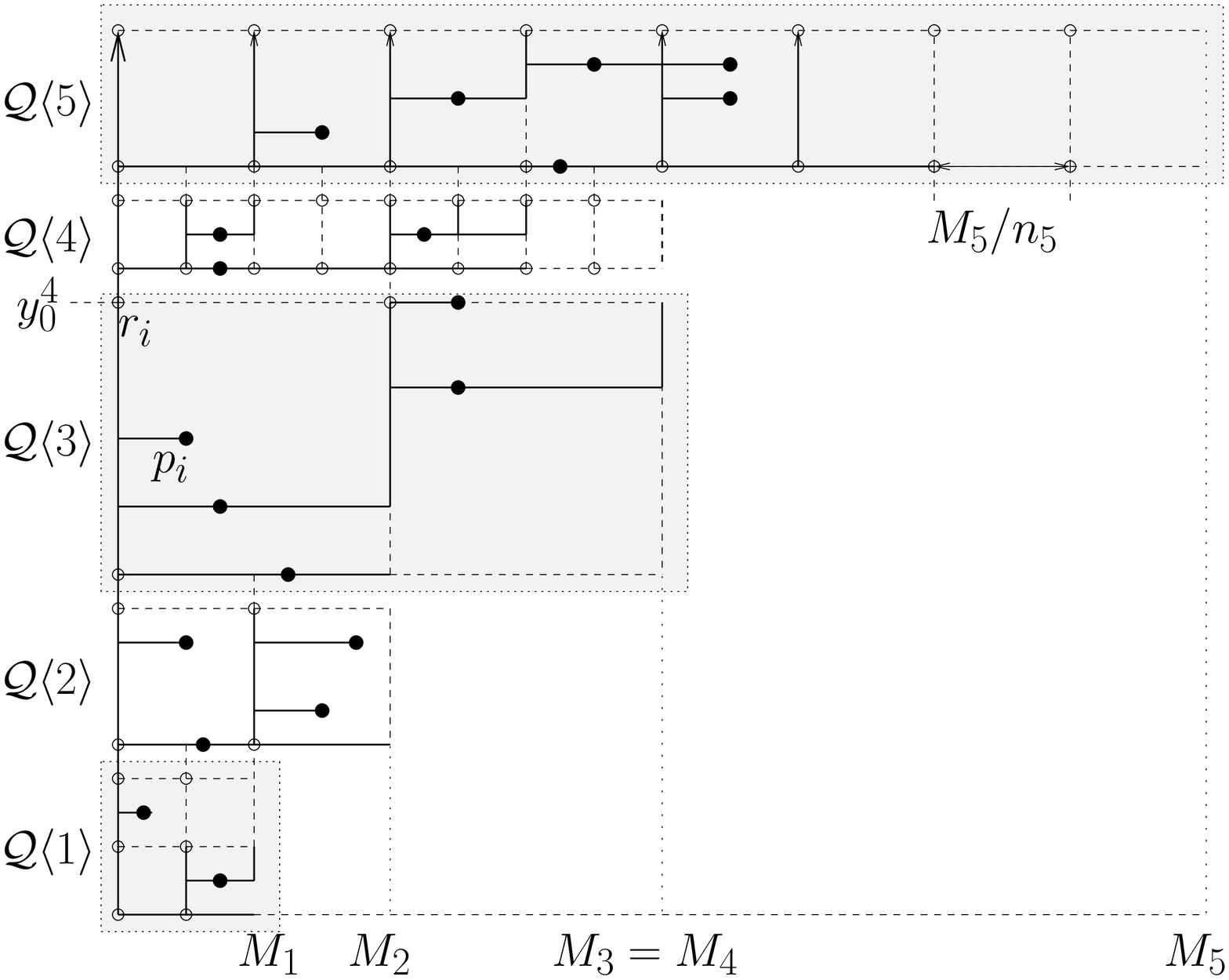}
\end{center}
\caption{\sf Partitioning $\calQ$ into subsets $\calQ\lrangle{1},\calQ\lrangle{2},...,\calQ\lrangle{5}$;
each instance corresponds to a subset $\calQ\lrangle{k}$, origin $(0,y^k_0)$, $\Mguess$ $\MM_k$ and $\nguess$ $\nn_k$.
\label{fig: onRSA partition to groups}
}
\end{figure}
}
\FigonRSA


\begin{figure}[ht!]
\fboxsep=0.2cm
\framebox[\textwidth]{
\begin{minipage}{0.95\textwidth}

\begin{enumerate}

\item when the first point $p_1$ arrives
    \begin{enumerate}
    \item $k\leftarrow 1$; $\calQ\lrangle{1}\leftarrow\{p_1\}$; $\MM_1\leftarrow 2^{\Mguess(1)}$; $\nn_1\leftarrow 4$; $g_1\equiv1$; and origin $p^1_0=(0,0)$.

    \item start an instance of $\onRSAmnp$ on $\calQ\lrangle{1}$;

    \end{enumerate}

\item when an input point arrives $p_i=(x_i,y_i)$ (for $i>1$), /* the points are $i=2,...,N */$
    \begin{enumerate}

    \item if $x_{i} \leq \MM_k$ and $i \leq \nn_k$, then $\calQ\lrangle{k}\leftarrow\calQ\lrangle{k}\cup\{p_{i}\}$.
    \item Otherwise,  ($p_{i}$ ``open a new instance''), then
        \begin{enumerate}
        \item $k\leftarrow k+1$;
        \item $\calQ\lrangle{k}\leftarrow\{i\}$;
        \item $\MM_{k}\leftarrow 2^{\Mguess(i)}$;

        \item $\nn_{k}\leftarrow \tetration{(\nguess(i))}$;

        \item $p^{k}_0 \equiv (0,y_{i-1})$;

        \item $g_{k}\leftarrow i$;
        \item start an instance of $\onRSAmnp$ on $\calQ\lrangle{k}$;
        \end{enumerate}

    \end{enumerate}
\item pass $p_i$ to the instance of $\onRSAmnp$ executing on $\calQ\lrangle{k}$ with origin $p^k_0$; $\MM=\MM_k$; and $\nn=\nn_k$ and compute $\FRSAmn(\{p_{g_k},...,p_{i+1}\})$.


\item  $\FRSA\leftarrow\FRSA\cup\FRSAmn(\{p_{g_k},...,p_{i+1}\})$.\\

\end{enumerate}

\end{minipage}
}
\caption{\label{figure: onRSA}
Algorithm $\onRSA$.
}
\end{figure}

\def\AppThmonRSA{
\begin{thm}
Algorithm
$\onRSA$ is optimal and is
\commdouble \\ \commdoubleend
 $O(\sqrt{\log \NN})$-competitive.
\label{appthm: onRSA}
\end{thm}
\begin{proof}
Consider a point set $\calQ=\{p_1=(x_1,y_1),...,p_\NN=(x_\NN,y_\NN)\}$.
Let $g_1,...,g_{\tau}$ be the separator indices of the subsets $\calQ\lrangle{1},...,\calQ\lrangle{\tau}$.
Recall that, for each subset $\calQ\lrangle{k}\in\{\calQ\lrangle{1},...,\calQ\lrangle{\tau}\}$, we have two types of guesses,
$\MM_k$ for $\MM$-$guess$ and $\nn_k$ for $\nguess$, and an origin at
$p_0^k=(x_0^k,y_0^k)$.
These parameters are computed by $\onRSA$ online.
It is easy to see that
the solution $\FRSAQ$ constructed by $\onRSA$ is feasible (the vertical segment $\linev{(0,0),(0,y_{\NN})}$ is contained in $\FRSAQ$;
thus, all the  origins $p_0^1,...,p_0^{\tau}$ are connected in $\FRSAQ$;
in addition, since subroutine $\onRSAmnp$ solves $\RSA$ for the $k$th instance (with $\calQ\lrangle{k}$), there exists a path from the origin point $p_0^k$ to every point of $\calQ\lrangle{k}$;
finally, every point of $\calQ$ belongs to some subset $\calQ\lrangle{k}$).
Thus also,
\begin{eqnarray}
\cost(\onRSA,\calQ)=\sum_{k=1}^{\tau}\cost(\onRSAmnk{k},\calQ\lrangle{k}),
\label{ineq:thm: cost(RSA Q)= sum of subsets costs}
\end{eqnarray}
since the solution of $\onRSA$ on $\calQ$ is a collection of the solutions of  $\onRSAmn$ on the instances $\calQ\lrangle{1},...,\calQ\lrangle{\tau}$ with the relates origins and guessing parameters.

Clearly, from definition (\ref{def: g k+1= separator}), for the last guess (of $\RSA$), $\MM_\tau$ and $n_\tau$, it holds that
$\MM_\tau/2\leq\xmaxQ\leq \MM_\tau$ and $\sqrt[4]{\nn_\tau}\leq \NN\leq \nn_\tau$. Thus,
\begin{eqnarray}
\cost(\opt,\calQ)\geq \MM_\tau/2
\label{cost(opt Q) geq M/2}
\label{ineq:thm: cost(opt Q) geq M/2}
\end{eqnarray}
and
\begin{eqnarray}
\sqrt{\log \NN}=\Theta(\sqrt{\log \nn_\tau}).
\label{sqrt(log sizeR)=Theta(sqrt(log NetSize)}
\label{ineq:thm: sqrt(log sizeR)=Theta(sqrt(log NetSize)}
\end{eqnarray}

We now compare the actions of an optimal algorithm $\opt$ for $\RSA$ when working on two different instances of $\RSA$.
The first instance $(\opt,\calQ)$ has the input $\calQ$ and the origin at $p_0^1=(0,0)$. The second instance $(\opt,\calQ\lrangle{k})$ has the input $\calQ\lrangle{k}$ and the origin at $p_0^k$.
Both the executions of $\opt$ on the first instance $(\opt,\calQ)$, and the execution on the second instance $(\opt,\calQ\lrangle{k})$, must serve the points in $\calQ\lrangle{k}$. We now compare their costs for serving those points only.

To define those costs, let $t_1^-=0$ and (for every $k=2,...,\tau$),
let $t_k^-$ be the time of the last point in $\calQ\lrangle{k-1}$ (that is,
$t_k^-=y_{g_k-1}$) and $t_k^+$ be the time of the last point in $\calQ\lrangle{k}$.
Given the solution $\FRSAQ$ of $\opt$ on the first instance
$(\opt,\calQ)$, we can speak on the {\em part} of that solution
dealing with $\calQ$. Formally,
let
$\FRSAQ \big|_{\calQ\lrangle{k}}$ be
$\FRSAQ$ projected on the time interval $[t_k^-,t_k^+]$.
The total length of the segments thus defined is
$\cost(\opt, \calQ) \big|_{\calQ\lrangle{k}}$.

The above solution of the first instance $(\opt,\calQ)$  may serve the points of $\calQ$ from any other parts of solution $\FRSAQ$ projected on times earlier than $t_k^-$ (or, of course, from points of $\calQ\lrangle{k}$).
Note that (from the definition of $M_k$), the $X$ coordinate of those earlier points is within the interval
$[0,M_k]$.
On the other hand, $\opt$ performing on the second instance $(\opt,\calQ\lrangle{k})$ may serve points in $\calQ\lrangle{k}$ only from either the origin of the second instance (or from points in $\calQ$). This means that
$\cost(opt,\calQ\lrangle{k}) \ge \cost(\opt, \calQ) \big|_{\calQ\lrangle{k}}$.

However, let us now consider a revision of the $\RSA$ problem (with the instance
$(\opt,\calQ\lrangle{k})$. In the revised problem, a copy is given initially not only in the origin $p_0^k$ but also along the path of length $M_k$ at time $t_k^-$. Clearly, an optimal solution for that problem is not worse than
$\cost(\opt, \calQ) \big|_{\calQ\lrangle{k}}$, since it can simulate
 $\FRSAQ \big|_{\calQ\lrangle{k}}$. On the other
hand, clearly, an optimal solution for the revised problem is better than
$\cost(\opt, \calQ\lrangle{k})$ by an additive factor of $M_k$.
This leads for the following inequality,
\begin{eqnarray}
\cost(\opt,\calQ) + \sum_{k=1}^{\tau} \MM_k \geq \sum_{k=1}^{\tau} \cost(\opt,\calQ\lrangle{k}).
\label{ineq: thm cost(opt, Q) geq sum cost Q'<k>}
\end{eqnarray}
Rewriting  Observation \ref{corollary: onRSAmn is O(sqrt)-cometitve},
\commsingle
\begin{eqnarray}
\cost(\onRSAmnk{k},\calQ\lrangle{k})\leq
c_1\sqrt{\log \nn_k}\left( \cost(\opt,\calQ\lrangle{k}) + \MM_k\right),
\label{ineq: thm cost onRSAmn Q<k> leq sqrt}
\end{eqnarray}
\commsingleend
\commdouble
\begin{eqnarray}
&&\cost(\onRSAmnk{k},\calQ\lrangle{k})\leq
\\
&&c_1\sqrt{\log \nn_k}\left( \cost(\opt,\calQ\lrangle{k})-\MM_k+2\MM_k\right),
\label{ineq: thm cost onRSAmn Q<k> leq sqrt}
\end{eqnarray}
\commdoubleend
where $c_1$ is some constant.
Therefore,
\commsingle
\begin{eqnarray}
\cost(\onRSA,\calQ)
&\leq&\nonumber
c_1\sum_{k=1}^{\tau}\sqrt{\log \nn_k}\cdot\cost(\opt,\calQ\lrangle{k})+ c_1\sum_{k=1}^{\tau}\sqrt{\log \nn_k} \MM_k\\
&\leq&
c_1\sqrt{\log \nn_\tau}\cdot\cost(\opt,\calQ)+ (c_1+1)\sum_{k=1}^{\tau}\sqrt{\log \nn_k} \MM_k,
\label{ineq:thm: cost (SRSA,Q) leq sqrt opt Q + sum log nk mk}
\end{eqnarray}
\commsingleend
where the first inequality holds by Inequality (\ref{ineq: thm cost onRSAmn Q<k> leq sqrt}) and Inequality (\ref{ineq:thm: cost(RSA Q)= sum of subsets costs});
the second inequality holds by Inequality (\ref{ineq: thm cost(opt, Q) geq sum cost Q'<k>}) and since $\nn_1\leq...\leq \nn_{\tau}$.

Inequalities (\ref{cost(opt Q) geq M/2}) and (\ref{sqrt(log sizeR)=Theta(sqrt(log NetSize)})
bounds the right hand side above. Hence,
it is left to bound $\sum_{k=1}^{\tau}(\sqrt{\log \nn_k}\cdot\MM_k)$ from above  as a function of $\MM_\tau$ and $\sqrt{\log \nn_\tau}$.
Recall that in some cases $\onRSA$ started a new instance $\onRSAmnk{k}$ (for ``new'' subset $\calQ\lrangle{k}$)
when $\MM_k$ larger than $2\cdot{\MM_{k-1}}$ and white in some other cases $\MM_k=\MM_{k-1}$, but $\nn_k=(\nn_{k-1})^4$.
Let $\calK_{\MM}^{guess}=\{1\}\cup\{k\in\{2,...,\tau\} \mid \MM_k>\MM_{k-1}\}$ be the indices
of instances of $\RSA$ of the first case above ($\MM_k>\MM_{k-1}$).
Similarly, let $\calK_{\nn}^{guess}=\{1\}\cup\{k\in\{2,...,\tau\} \mid \nn_k>\nn_{k-1}\}$
be the indices of instances of the second case ($\nn_k>\nn_{k-1}$).
We have
\commsingle
\begin{eqnarray*}
\sum_{k=1}^{\tau}\sqrt{\log \nn_k}\cdot\MM_k \leq
\sqrt{\log \nn_\tau}\left(\sum_{k\in\calK_{\MM}^{guess}}\MM_k \right)
+
\MM_\tau\left(\sum_{k\in\calK_{\nn}^{guess}}\sqrt{\log \nn_k}\right),
\end{eqnarray*}
\commsingleend
\commdouble
\begin{eqnarray*}
&&\sum_{k=1}^{\tau}\sqrt{\log \nn_k}\cdot\MM_k\\
&& \leq
\sqrt{\log \nn_\tau}\left(\sum_{k\in\calK_{\MM}^{guess}}\MM_k \right)
+
\MM_\tau\left(\sum_{k\in\calK_{\nn}^{guess}}\sqrt{\log \nn_k}\right),
\end{eqnarray*}
\commdoubleend
since $\{1,...,\tau\}= \calK_{\MM}^{guess}\cup\calK_{\nn}^{guess}$
($k=1$ belongs to both sets; for every $k>1$, either $\MM_k>\MM_{k-1}$ or $\nn_k>\nn_{k-1}$, thus $k\in\calK_{\MM}^{guess}\cup\calK_{\nn}^{guess}$),
$\nn_1\leq...\leq \nn_{\tau}$ and $\MM_1\leq...\leq \MM_{\tau}$.
Recall that $\MM_{k}\geq2\MM_{k-1}$ for each $k\in\calK_{\MM}^{guess}\setminus\{1\}$.
Thus,
\begin{eqnarray*}
\sum_{k\in\calK_{\MM}^{guess}}\MM_k
\leq
\MM_\tau\sum_{i=0}^{ \infty} 2^{-i}\leq 2\MM_\tau.
\end{eqnarray*}
Similarly, $\nn_{k}= (\nn_{k-1})^4$, for each $k\in\calK_{\nn}^{guess}\setminus\{1\}$.
Thus,
\begin{eqnarray*}
\sum_{k\in\calK_{\nn}^{guess}}\sqrt{\log \nn_k}
\leq
\sqrt{\log \nn_k}\sum_{i=0}^{\infty}2^{-i}
\leq
2\sqrt{\log n_\tau},
\label{ineq: tetartion: sum sqrt log tetartion}
\end{eqnarray*}
since, $\sqrt{\log 2^{2^{2l}}}=2^l$.
Hence,
\begin{eqnarray}
\sum_{k=1}^{\tau}\sqrt{\log \nn_k}\cdot\MM_k \leq 4\MM_\tau\sqrt{\log \nn_\tau},
\label{ineq:thm: sum Mk log nk leq Mtau log ntau}
\end{eqnarray}
The theorem follows (from inequalities
(\ref{ineq:thm: cost(opt Q) geq M/2}), (\ref{ineq:thm: sqrt(log sizeR)=Theta(sqrt(log NetSize)}),
(\ref{ineq:thm: cost (SRSA,Q) leq sqrt opt Q + sum log nk mk})
and (\ref{ineq:thm: sum Mk log nk leq Mtau log ntau})).
\QED\end{proof}
} 

\begin{theorem}
Algorithm $\onRSA$ is optimal and is $O(\frac{\log N}{\log \log N})$-competitive.
\label{thm: onRSA}
\end{theorem}

\subsection{Optimizing $\DMCD$ for a small number of requests}
\label{sec:Optimal-mcd-for-few-requests}
Algorithm $\Dlineon$ was optimal only as the function of the network size.
Recall that our solution for $\RSA$ was optimal as a function of the number of requests.
We obtain this property for the solution of $\DMCD$ too, by   transforming our $\RSA$ algorithm back to solve $\DMCD$, and obtain the promised competitiveness,
 $O(  \min \{ \frac{\log N}{\log \log N}  ,  \frac{\log n}{\log \log n} \})$.

Algorithm $\Dlineon$ was optimal as the function of the network size (Theorem \ref{thm: Dlineon is frac(log n)(log log n) competitive}).
This means that it may not be optimal in the case that the number of requests is much smaller than the network size. In this section, we use Theorem \ref{thm: onRSA} and algorithm $\onRSA$ to derive an improve algorithm for $\MCD$. This algorithm, $\Dlineonp$, is competitive optimal (for $\DMCD$) for any number of requests.
Intuitively, we benefit from the fact that $\onRSA$ is optimal for any number of points (no notion of network size exists in $\RSA$).

This requires the solution of some delicate point. Given an instance $\DMCD^a$ of $\DMCD$, we would have liked to just translate the set $\calR^a$ of $\DMCD$ requests into a set $\calQ$ of $\RSA$ points and
apply $\onRSA$ on them.
This may be a bit confusing, since $\onRSA$ performs by converting back to $\DMCD$.
Specifically,
recall that
 $\onRSA$ breaks $\calQ$ into several subsets, and translates back first the first subset $\calQ_1$ into an the requests set $\calR^b_1$ of a new instance
 $\DMCD^b_1$ of $\DMCD$.
Then, $\onRSA$ invokes  $\Dlineon$ on this new instance $\DMCD^b_1$.
The delicate point is that $\DMCD^b_1$ is different than $\DMCD_a$.

In particular, the fact that $\calQ_1$ contains only {\em some} of the points of $\calR^a$, may cause $\onRSA$ to ``stretch'' their $X$ coordinates to fit them into the network of $\DMCD_a$.
Going carefully over the manipulations performed by $\onRSA$
reveals that the solution of $\onRSA$ may not be a feasible solution
of $\DMCD$ (even though it applied $\Dlineon$ plus some manipulations).
Intuitively, the solution of $\onRSA$
  may ``store copies'' in places that are not grid vertices in the grid of $\DMCD_a$. Thus the translation to a solution of $\DMCD_1$ is not immediate.

Intuitively, to solve this problem, we translate a solution of $\onRSA$ to a solution of $\DMCD_a$ in a way that is similar to the way we translated a solution of $\Dlineon$ to a solution of $\RSA$. That is,  each request of $\DMCD_a$ we move to a ``nearby'' point of $\onRSA$. This is rather straightforward, given the description of our previous transformation (of Section \ref{subsec: onRSAn}).
The details are left for the full paper.

\begin{theorem}
Algorithm $\Dlineonp$ is optimal and it
\begin{center}
$O(\min\{\frac{\log N}{\log\log N},\frac{\log n}{\log \log n}\})$-competitive.
\end{center}
\end{theorem}

\vspace{-0.2cm}
\section{Lower Bound for $\RSA$}

\label{sec:Lower bound}

In this section, we prove the following theorem, establishing a tight lower bound for $\RSA$ and for $\DMCD$ on directed line networks.
Interestingly,  this lower bound is not far from the one proven by Alon and Azar  for {\em undirected} Euclidian Steiner trees \cite{AlonAzar93}.
Unfortunately,  the lower bound of \cite{AlonAzar93} does not apply to our case since their construct uses edges directed in what would be the wrong direction in our case (from a high $Y$ value to a low one).

\begin{theorem}
The competitive ratio of any deterministic online algorithm for
$\DMCD$ in directed line networks is $\Omega(\frac{\log n}{\log\log n})$,
implying also an
$\Omega(\frac{\log N}{\log\log N})$ lower bound for $\RSA$.
\label{thm: lower bound for RSA and MCD on directed}
\end{theorem}
\proof
We first outline the proof.
Informally, given a deterministic online algorithm $\onalgrsa$, we construct an adversarial input sequence.
Initially, the request set includes the set $\diag=\{(k,k)\mid 0\leq k\leq n\}$.
That is, at each time step $t$, the request $(t,t)$ is made.
In addition, if the algorithm leaves ``many copies'' then the lower bound is easy.
Otherwise, the algorithm leaves ``too few copies'' from some time $t-1$ until time $t$.
For each such time, the adversary makes another request at $(t-k,t)$ for some $k$ defined later.
The idea is that the adversary can serve this additional request from the diagonal copy at $(t-k,t-k)$ paying the cost of $k$.
On the other hand,
the algorithm is not allowed at time $t$ to decide to serve from $(t-k,t-k)$.
It must serve from a copy it did leave.
Since the algorithm left only ``few'' copies to serve time $t$ the replica, $(t,t-k)$ can be chosen at least at distance $k(\log n)$ from any copy the algorithm did leave.
Hence, the algorithm's cost for such a time $t$ is $\Omega(\log n)$ times greater than that of the adversary.

More formally, let $\delta =\lceil\log n\rceil$.
Partition the line at time $t \in \set{n/2,\ldots,n}$ into
$\lfloor\log_\delta n -1\rfloor$ intervals:
$I_i(t) =
(t-\delta^{i+1},t-\delta^i]$, where $i \in \set{1,2, \ldots,
\lfloor\log_\delta n-1\rfloor}$.
(Note that the intervals are well defined, since $\lfloor\log_\delta n-1\rfloor\leq \lfloor\log_\delta t\rfloor$, for every $n/2\leq t\leq n$, which implies that $\delta^i\leq t$ for every $i=1,...,\lfloor\log_\delta n-1\rfloor$.)
Given an online algorithm $\onalgrsa$, the adversary constructs the set of
requests $\calR$ as follows.  Initially, $\calR \gets \diag$.
For each
time $t\geq n/2$, denote by $V_{\alg}(t)$ the set of
nodes that hold the movie for time $t$ (just before $\onalgrsa$ receives
the requests for time $t$).
The adversary may add a request at $t$ according to $V_{\alg}(t)$.
In particular,
if $\onalgrsa$ leaves a copy in at least one of the nodes of every such intervals $I_i(t)$, for $i=1,...,\lfloor\log_\delta n-1\rfloor$,
then the only adversary request for time $t$ is $(t,t)$ (while $\onalgrsa$ left copies in at least $\lfloor\log_\delta n -1\rfloor$ nodes).
Otherwise, the adversary adds the request $(t-\delta^{i^*},t)$ to $\calR$, where $i^*$ is an arbitrary index such that $I_{i^*}(t)\cap V_{\alg}(t) = \emptyset$.
%
%
That is, the adversary request set of time $t$ is $\{(t,t)\}$ in the first case and  $\{(t-\delta^{i^*},t),(t,t)\}$ in the second case.

For each time $t= \lfloor n/2\rfloor,...,n$,
one of the following two cases hold:
{\bf (1)}
$\onalgrsa$ pays at least $\lfloor\log_\delta n-1\rfloor=\Omega(\frac{\log n}{\log \log n})$ for storing at least $\lfloor\log_\delta n-1\rfloor$ copies from time $t-1$ to time $t$,
while the adversary pays just $2=O(1)$ (to serves request $(t,t)$); or
{\bf (2)} $\onalgrsa$ pays, at least, $\delta^{i^*+1}-\delta^{i^*}=\Omega(\delta^{i^*+1})$ for delivering a copy  to $(t-\delta^{i^*}_t,t)$ from some node outside the interval $I_{i^*}(t)$,
while the adversary pays $O(\delta^{i^*})$ for storing the movie in node $t-\delta^{i^*}$ from time $t-\delta^{i^*}$ to time $t$
(that is, serving from replica $(t-\delta^*,t-\delta^*)$ on the diagonal) and additional two edges (to serve request $(t,t)$).
Thus, in that case, $\onalgrsa$ pays at least $O(\log n)$ times more than the adversary.
This establishes Theorem \ref{thm: lower bound for RSA and MCD on directed}.
\QED


\def\thepage{}
{\small

}

\end{document}